\documentclass[11pt]{article}

\usepackage[utf8]{inputenc}
\usepackage[english]{babel}
\usepackage{graphicx}         
\usepackage{hyperref}
\usepackage{braket}
\usepackage{amsmath}
\usepackage{amssymb}
\usepackage{amsthm}
\usepackage{mathtools}
\usepackage{bbm}		
\usepackage{array}		
\usepackage{color}
\usepackage[usenames,dvipsnames]{xcolor}
\usepackage{caption}
\usepackage{enumitem}
\usepackage{subeqnarray}	
\usepackage{wrapfig}

\allowdisplaybreaks

\numberwithin{equation}{section}

\theoremstyle{plain}\newtheorem{definition}{Definition}[section]
\newtheorem{lem}[definition]{Lemma}
\newtheorem{proposition}[definition]{Proposition}
\newtheorem{cor}[definition]{Corollary}
\theoremstyle{remark}\newtheorem{remark}[definition]{Remark}
\theoremstyle{plain}\newtheorem{claim}{Claim}

\theoremstyle{plain}\newtheorem{assumption}{Assumption}
\theoremstyle{plain}\newtheorem{theorem}{Theorem}

\newcommand{\lemit}[1]{\begin{enumerate}[label={(\alph*)}, ref={\thelem\alph*}]{#1}\end{enumerate}}
\newcommand{\corit}[1]{\begin{enumerate}[label={(\alph*)}, ref={\thecor\alph*}]{#1}\end{enumerate}}
\newcommand{\remit}[1]{\begin{enumerate}[label={(\alph*)}, ref={\theremark\alph*}]\item[]{#1}\end{enumerate}}

\renewcommand{\tilde}[1]{\widetilde{#1}}

\newcommand{\ls}{\lesssim}

\newcommand{\lr}[1]{\left\langle #1 \right\rangle}
\newcommand{\lrt}[1]{\left\langle #1 \right\rangle^{(t)}}

\newcommand{\norm}[1]{\lVert#1\rVert}
\newcommand{\onorm}[1]{\lVert#1\rVert_\mathrm{op}}

\newcommand{\D}{\mathcal{D}}
\newcommand{\R}{\mathbb{R}}

\newcommand{\N}{\mathbb{N}}
\newcommand{\fH}{\mathfrak{H}}
\newcommand{\Fock}{\mathcal{F}}
\newcommand{\Number}{\mathcal{N}}
\newcommand{\id}{\mathbbm{1}}
\newcommand{\cJ}{\mathcal{J}}
\newcommand{\cS}{\mathcal{S}}
\newcommand{\fV}{\mathfrak{V}}
\newcommand{\cU}{\mathcal{U}}
\newcommand{\cQ}{\mathcal{Q}}

\newcommand\mydots{,\makebox[1em][c]{.\hfil.\hfil.},}
\newcommand\mycdots{\makebox[1em][c]{$\cdot$\hfil$\cdot$\hfil$\cdot$}}

\newcommand{\Tr}{\mathrm{Tr}}
\renewcommand{\d}{\mathop{}\!\mathrm{d}}
\newcommand{\dx}{\d x}
\newcommand{\dy}{\d y}
\newcommand{\dz}{\d z}
\newcommand{\dt}{\d t}
\newcommand{\ds}{\d s}

\let\textl\l
\renewcommand{\l}{\ell}
\renewcommand{\i}{\mathrm{i}}
\newcommand{\e}{\mathrm{e}}
\newcommand{\hc}{\mathrm{h.c.}}
\newcommand{\sym}{\mathrm{sym}}
\newcommand{\trap}{\mathrm{trap}}

\newcommand{\bj}{\boldsymbol{j}}
\newcommand{\bm}{\boldsymbol{m}}

\newcommand{\HN}{H^N}
\newcommand{\PsiN}{\Psi^N}
\newcommand{\PsiNl}{\PsiN_\l}

\newcommand{\gPsiNo}{\gamma_{{\PsiN}}^{(1)}}
\newcommand{\gPsiNk}{\gamma_{{\PsiN}}^{(k)}}

\newcommand{\pt}{{\varphi(t)}}
\newcommand{\ps}{{\varphi(s)}}

\newcommand{\pz}{{\varphi(0)}}

\newcommand{\hpt}{h^{\pt}}
\newcommand{\mpt}{\mu^{\pt}}

\newcommand{\FH}{{\Fock(\fH)}}
\newcommand{\Fp}{\Fock_\perp}
\newcommand{\FN}{\Fock^{\leq N}}
\newcommand{\FNp}{\Fp^{\leq N}}
\newcommand{\Fpz}{\Fock_{\perp\pz}}
\newcommand{\Fpt}{{\Fock_{\perp\pt}}}
\newcommand{\Fps}{{\Fock_{\perp\ps}}}

\newcommand{\FNpt}{{\Fock^{\leq N}_{\perp\pt}}}
\newcommand{\FNps}{{\Fock^{\leq N}_{\perp\ps}}}

\newcommand{\Npt}{\Number_{\perp\pt}}

\newcommand{\Qpt}{\mathcal{Q}_{\perp\pt}}
\newcommand{\Qps}{\mathcal{Q}_{\perp\ps}}

\newcommand{\ad}{a^\dagger}
\newcommand{\Ad}{A^\dagger}

\newcommand{\UNpt}{\mathfrak{U}_{N,\pt}}
\newcommand{\UNpz}{\mathfrak{U}_{N,\pz}}

\newcommand{\Chi}{{\boldsymbol{\chi}}}
\newcommand{\ChiN}{\Chi^{\leq N}}
\newcommand{\Chil}{\Chi_\l}
\newcommand{\Chim}{\Chi_m}
\newcommand{\Chio}{\Chi_0}
\newcommand{\Chit}{\Chi_1}

\newcommand{\bPhi}{{\boldsymbol{\phi}}}

\newcommand{\chik}{\chi^{(k)}}

\newcommand{\FockH}{\mathbb{H}}
\newcommand{\FockHN}{\FockH^{\leq N}}
\newcommand{\FockHNpt}{\FockHN_{\pt}}
\newcommand{\FockHpt}{\FockH_{\pt}}
\newcommand{\FockHps}{\FockH_{\ps}}

\newcommand{\FockHopt}{\FockHpt^{(0)}}
\newcommand{\FockHops}{\FockHps^{(0)}}
\newcommand{\FockHtpt}{\FockHpt^{(1)}}

\newcommand{\FockHthpt}{\FockHpt^{(2)}}

\newcommand{\FockHnpt}{\FockHpt^{(n)}}
\newcommand{\FockHnps}{\FockHps^{(n)}}

\newcommand{\qpt}{q^{\pt}}
\newcommand{\Kopt}{K^{(1)}_{\pt}}
\newcommand{\Kops}{K^{(1)}_{\ps}}

\newcommand{\Ktpt}{K^{(2)}_{\pt}}
\newcommand{\Ktps}{K^{(2)}_{\ps}}

\newcommand{\Kthpt}{K^{(3)}_{\pt}}
\newcommand{\Kthps}{K^{(3)}_{\ps}}

\newcommand{\Kfpt}{K^{(4)}_{\pt}}
\newcommand{\Kfps}{K^{(4)}_{\ps}}

\newcommand{\tKopt}{\tilde{K}^{(1)}_{\pt}}
\newcommand{\tKtpt}{\tilde{K}^{(2)}_{\pt}}

\newcommand{\Wpt}{W_{\pt}}

\newcommand{\boldKzpt}{\mathbb{K}^{(0)}_\pt}
\newcommand{\boldKopt}{\mathbb{K}^{(1)}_\pt}
\newcommand{\boldKtpt}{\mathbb{K}^{(2)}_\pt}
\newcommand{\boldKthpt}{\mathbb{K}^{(3)}_\pt}
\newcommand{\boldKfpt}{\mathbb{K}^{(4)}_\pt}
\newcommand{\boldKtptbar}{\overline{\mathbb{K}^{(2)}_\pt}}
\newcommand{\boldKthptbar}{\overline{\mathbb{K}^{(3)}_\pt}}

\newcommand{\boldKops}{\mathbb{K}^{(1)}_\ps}

\newcommand{\Uo}{U^{(0)}_\varphi}

\newcommand{\BogV}{\mathcal{V}}
\newcommand{\BogU}{\,\mathcal{U}_\BogV}
\newcommand{\BogUz}{\,\mathcal{U}_{\BogV_0}}
\newcommand{\BogUts}{\,\mathcal{U}_{\BogV(t,s)}}
\newcommand{\BogUtz}{\,\mathcal{U}_{\BogV(t,0)}}
\newcommand{\BogUt}{\,\mathcal{U}_{\BogV_t}}

\newcommand{\gbar}{\overline{g}}
\newcommand{\Ubar}{\overline{U}}
\newcommand{\Vbar}{\overline{V}}

\newcommand{\omzz}{\omega^{(-1,-1)}}
\newcommand{\omzo}{\omega^{(-1,1)}}
\newcommand{\omoz}{\omega^{(1,-1)}}
\newcommand{\omoo}{\omega^{(1,1)}}
\newcommand{\omlj}{\omega^{(\l,j)}}

\newcommand{\omzjo}{\omega^{(-1,j_1)}}
\newcommand{\omzjt}{\omega^{(-1,j_2)}}
\newcommand{\omzjth}{\omega^{(-1,j_3)}}
\newcommand{\omzjf}{\omega^{(-1,j_4)}}

\newcommand{\omojo}{\omega^{(1,j_1)}}
\newcommand{\omojt}{\omega^{(1,j_2)}}
\newcommand{\omojth}{\omega^{(1,j_3)}}

\newcommand{\asj}{a^{\sharp_j}}
\newcommand{\asl}{a^{\sharp_\l}}
\newcommand{\asjo}{a^{\sharp_{j_1}}}
\newcommand{\asjt}{a^{\sharp_{j_2}}}

\newcommand{\gChiot}{\gamma_{\Chio(t)}}
\newcommand{\gChioz}{\gamma_{\Chio(0)}}

\newcommand{\aChiot}{\alpha_{\Chio(t)}}
\newcommand{\aChioz}{\alpha_{\Chio(0)}}

\newcommand{\gChiNt}{\gamma_{\ChiN(t)}}
\newcommand{\bChiNt}{\beta_{\ChiN(t)}}
\newcommand{\bzo}{\beta_{0,1}}

\newcommand{\alnmm}{\mathfrak{a}^{(\l)}_{n,m,\mu}}

\newcommand{\xn}{x^{(n)}}
\newcommand{\xk}{x^{(k)}}

\newcommand{\lN}{\lambda_N}

\makeatletter
\DeclareFontFamily{OMX}{MnSymbolE}{}
\DeclareSymbolFont{MnLargeSymbols}{OMX}{MnSymbolE}{m}{n}
\SetSymbolFont{MnLargeSymbols}{bold}{OMX}{MnSymbolE}{b}{n}
\DeclareFontShape{OMX}{MnSymbolE}{m}{n}{
    <-6>  MnSymbolE5
   <6-7>  MnSymbolE6
   <7-8>  MnSymbolE7
   <8-9>  MnSymbolE8
   <9-10> MnSymbolE9
  <10-12> MnSymbolE10
  <12->   MnSymbolE12
}{}
\DeclareFontShape{OMX}{MnSymbolE}{b}{n}{
    <-6>  MnSymbolE-Bold5
   <6-7>  MnSymbolE-Bold6
   <7-8>  MnSymbolE-Bold7
   <8-9>  MnSymbolE-Bold8
   <9-10> MnSymbolE-Bold9
  <10-12> MnSymbolE-Bold10
  <12->   MnSymbolE-Bold12
}{}

\let\llangle\@undefined
\let\rrangle\@undefined
\DeclareMathDelimiter{\llangle}{\mathopen}
                     {MnLargeSymbols}{'164}{MnLargeSymbols}{'164}
\DeclareMathDelimiter{\rrangle}{\mathclose}
                     {MnLargeSymbols}{'171}{MnLargeSymbols}{'171}
\makeatother

\newcommand{\ppt}{p^{\pt}}
\newcommand{\Chik}{\Chi_k}
\newcommand{\scp}[2]{\big\langle #1 , #2 \big\rangle}

\newcommand{\HS}{\mathrm{HS}}

\newcommand{\Order}{\mathcal{O}}

\renewcommand{\d}{\mathop{}\!\mathrm{d}}

\renewcommand{\Im}{\mathrm{Im}}

\title{Beyond Bogoliubov Dynamics}

\author{Lea Boßmann\thanks{Fachbereich Mathematik, Eberhard Karls Universität Tübingen, 
	Auf der Morgenstelle 10, 72076 Tübingen, Germany; and Institute of Science and Technology Austria, Am Campus 1, 3400 Klosterneuburg, Austria. \texttt{lea.bossmann@ist.ac.at}},\;
	Sören Petrat\thanks{Department of Mathematics and Logistics, Jacobs University Bremen, Campus Ring 1, 28759 Bremen, Germany; and University of Bremen, Department 3 – Mathematics, Bibliothekstr. 5, 28359 Bremen, Germany. \texttt{s.petrat@jacobs-university.de}},\;
	Peter Pickl\thanks{Mathematisches Institut, Ludwig-Maximilians-Universität München, Theresienstr.\ 39, 80333 München, Germany. \texttt{pickl@math.lmu.de}},\;
	and Avy Soffer\thanks{Department of Mathematics, Rutgers University, 110 Frelinghuysen Road, Piscataway, NJ 08854, USA. \texttt{soffer@math.rugers.edu}}
}

\date{\today}

\begin{document}
\maketitle

\begin{abstract}
\noindent 
We consider a system of $N$ interacting bosons in the mean-field scaling regime and construct corrections to the Bogoliubov dynamics that approximate the true $N$-body dynamics in norm to arbitrary precision. The $N$-independent corrections are given in terms of the solutions of the Bogoliubov and Hartree equations and satisfy a generalized form of Wick's theorem.
We determine the $n$-point correlation functions of the excitations around the condensate, as well as the reduced densities of the $N$-body system, to arbitrary accuracy, given only the knowledge of the two-point correlation functions of a quasi-free state and the solution of the Hartree equation. In this way, the complex problem of computing all $n$-point correlation functions for an interacting $N$-body system is essentially reduced to the problem of solving the Hartree equation and the PDEs for the Bogoliubov two-point correlation functions.
\end{abstract}

\noindent
\textbf{MSC class:} 35Q40, 35Q55, 81Q05, 82C10

\section{Introduction}\label{sec:intro}
We consider a system of $N$ weakly interacting bosons in $\R^d$, $d\geq1$, which are initially prepared in the state
\begin{equation*}
\PsiN_\mathrm{trap}\in \bigotimes\limits_\mathrm{sym}^N L^2(\R^d)=:\bigotimes\limits_\mathrm{sym}^N\fH =:\fH^N_\sym\,,
\end{equation*}
which is either the ground state or an appropriate low-energy excited eigenstate
of the Hamiltonian
\begin{equation}\label{H:N:trap}
\HN_\mathrm{trap} := \sum\limits_{j=1}^N\left(-\Delta_j+V_\mathrm{trap}(x_j)\right)+\lN\sum\limits_{1\leq i<j\leq N} v(x_i-x_j)\,.
\end{equation}
Here, $V_\mathrm{trap}:\R^d\to\R$   denotes a suitable trapping potential,  $v:\R^d\to\R$ is some bounded pair interaction, and 
\begin{equation}
\lN:=\frac{1}{N-1}\,.
\end{equation}
This model describes bosons in the mean-field or Hartree regime, which is characterized by weak, long-range interactions.
At time $t=0$, we remove the trapping potential and study the dynamics generated by the Hamiltonian
\begin{equation}\label{HN}
\HN:=\sum\limits_{j=1}^N\left(-\Delta_j\right)+\lN\sum\limits_{1\leq i<j\leq N}v(x_i-x_j) \,,
\end{equation}
which is determined by the $N$-body Schrödinger equation,
\begin{equation}\label{SE}
\i\partial_t \PsiN(t) = \HN \PsiN(t)\,, \qquad \PsiN(0)=\PsiN_\mathrm{trap}\,.
\end{equation}
The dynamics is non-trivial since the eigenstates of $\HN_\mathrm{trap}$ are not eigenstates of $\HN$ anymore. As a consequence of the high-dimensional configuration space $\R^{dN}$ and since all particles become correlated under the time evolution, it is practically impossible to  solve \eqref{SE} analytically or numerically for any reasonably large particle number $N$.
It is therefore of physical relevance as well as of mathematical interest to derive and analyze suitable approximations to the solution $\PsiN(t)$ of \eqref{SE}. 
In this paper, we propose a perturbative  scheme which yields an approximation of $\PsiN(t)$ to any order in $N^{-1}$.
Our construction is such that all corrections to correlation functions and expectation values of bounded operators are given in terms of the two-point correlation functions of a quasi-free state. This crucially reduces the complexity of the $N$-body problem and makes it possible to numerically compute these physically relevant quantities to arbitrary precision.\medskip

It is well known that the ground state as well as the low-energy excited eigenstates of $\HN_\mathrm{trap}$ exhibit Bose--Einstein condensation in the state $\varphi_\mathrm{trap}\in\fH$, which is given by the minimizer of the corresponding Hartree energy functional \cite{benguria-lieb1983,Lieb_Yau1987}.
To construct a norm approximation of the many-body state, one decomposes $\PsiN_\mathrm{trap}$ into the condensate part and excitations from the condensate as in \cite{lewin2015_2},
\begin{equation}\label{eqn:decomp}
\PsiN_\mathrm{trap}=\sum\limits_{k=0}^N\varphi_\mathrm{trap}^{\otimes (N-k)}\otimes_s\chik_\mathrm{trap}\,, \qquad 
\ChiN_\mathrm{trap}:=\big(\chik_\mathrm{trap}\big)_{k=0}^N\in\FN_{\perp\varphi_\mathrm{trap}}\,,
\end{equation}
where $\otimes_s$ denotes the symmetric tensor product. The set of $k$-particle excitations $\chik_\mathrm{trap}$ forms a vector $ \ChiN_\mathrm{trap}$ in the truncated Fock space over the orthogonal complement of $\varphi_\mathrm{trap}$.
In \cite{seiringer2011,grech2013,lewin2015_2}, it was shown that $\ChiN_\mathrm{trap}$ is approximated in norm by an appropriate eigenstate $\Chi_0^\mathrm{trap}$ of the Bogoliubov Hamiltonian corresponding to $\HN_\mathrm{trap}$. 
By \eqref{eqn:decomp}, this leads to an approximation of $\PsiN_\mathrm{trap}$ with respect to the $\fH^N$-norm.

This analysis was extended in \cite{spectrum}, where it was shown that the ground state and low-energy excited states of the Hamiltonian $\HN_\mathrm{trap}$ admit an asymptotic expansion in the parameter $\lN^{1/2}$, i.e., it holds that
\begin{equation}\label{Chi:N:trap}
\Big\|\ChiN_\mathrm{trap}-\sum\limits_{\l=0}^a\lN^\frac{\l}{2}\Chil^{\mathrm{trap}}\Big\|_{\FN_{\perp\varphi_\mathrm{trap}}}\leq C\lN^\frac{a+1}{2}
\end{equation}
for $N$-independent coefficients $\Chil^\mathrm{trap}\in\Fock_{\perp\varphi_\mathrm{trap}}$.
Since the leading order $\Chi_0^{\mathrm{trap}}$ is an eigenstate of the Bogoliubov Hamiltonian associated with $\HN_\mathrm{trap}$, it is related via a Bogoliubov transformation to a state with fixed particle number. In particular, if $\PsiN_\mathrm{trap}$ is the ground state of $\HN_\mathrm{trap}$, $\Chio^\mathrm{trap}$ is quasi-free, i.e., it arises as Bogoliubov transformation of the vacuum. 
The higher orders in the expansion \eqref{Chi:N:trap} can be written as
\begin{equation}\label{Chi:l:trap}
\Chil^\mathrm{trap}=
\sum\limits_{\substack{0\leq n\leq 3\l\\n+\l\text{ even}}}\;\,
\sum\limits_{\bj\in\{-1,1\}^n}
\int\dx^{(n)}\,
{\mathfrak{a}}^{(\bj)}_{\l,n}\big(x^{(n)}\big) a^{\sharp_{j_1}}_{x_1}\,\mycdots a^{\sharp_{j_n}}_{x_n}\, \Chio^\mathrm{trap}
\end{equation}
for some known coefficients ${\mathfrak{a}}^{(\bj)}_{\l,n}$, a multi-index $\bj=(j_1\mydots j_n)$, and where we abbreviated  $a^{\sharp_{-1}}:=a$ (annihilation operator), $a^{\sharp_1}:=\ad$ (creation operator) and $x^{(n)}:=(x_1\mydots x_{n})$. \medskip

The property of Bose--Einstein condensation is preserved by the time evolution \eqref{SE}, which was first shown in the 1970s and 1980s by Hepp, Ginibre, Velo, and Spohn \cite{hepp, ginibre1979, ginibre1979_2, spohn1980}. Interest in the question was revived in the early 2000s through the work of Bardos, Golse, and Mauser \cite{bardos2000}, and a series of papers by Erd\H{o}s, Schlein, and Yau \cite{erdos2001,erdos2006,erdos2007,erdos2007_2,erdos2009,erdos2009_3,erdos2010}. Since then, many more results have been obtained, e.g., in \cite{adami2004, adami2007, froehlich2007, rodnianski2009, frohlich2009,knowles2010, kirkpatrick2011,pickl2011,pickl2015,benedikter2015, chen2013_2, sohinger15, anapolitanos2016, jeblick2016, jeblick2017,brennecke2017, jeblick2018}, see also \cite{benedikter_lec, golse_lec} for overviews of the topic. The optimal convergence rate was proven in \cite{erdos2009,chenlee:2011,chen2011,mitrouskas2016}, where it was shown that for any $t\in\R$, there exists some $C(t)>0$ such that
\begin{equation}\label{eqn:RDM:convergence}
\Tr\left|\gPsiNo(t)-|\pt\rangle\langle\pt|\right|\leq C(t) N^{-1}\,,
\end{equation}
where $\gPsiNo(t)$ denotes the one-particle reduced density matrix of $\PsiN(t)$ and where $\pt$ is the solution of the Hartree equation \eqref{hpt} with initial datum $\varphi(0)=\varphi_\mathrm{trap}$. We refer to \cite{LSSY,Froehlich_Lenzmann_2003_rev,lewin2014} for an overview of the results and for further references.

To characterize the $N$-body dynamics on the level of the wave function, one decomposes, analogously to \eqref{eqn:decomp},
\begin{equation}
\PsiN(t)=\sum\limits_{k=0}^N\pt^{\otimes (N-k)}\otimes_s\chik(t)\,, \qquad \ChiN(t):=\big(\chik(t)\big)_{k=0}^N\in\FNpt\,.
\end{equation}
In \cite{lewin2015}, it was shown that $\ChiN(t)$ is approximated in norm by the solution $\Chio(t)$ of the  Bogoliubov (or Bogoliubov--de~Gennes) equation,
\begin{equation}\label{eqn:Bogoliubov:equation}
\i\partial_t\Chio(t)=\FockHopt\Chio(t)\,, \qquad
\Chio(0)=\Chi_0^\mathrm{trap}\,,
\end{equation}
where $\FockHopt$ denotes the Bogoliubov Hamiltonian corresponding to $\HN$.  The dynamics \eqref{eqn:Bogoliubov:equation} describes the formation of pair correlations, which gives the first order correction to the mean-field dynamics. 
Fluctuations around the mean-field were first analyzed by Ginibre and Velo \cite{ginibre1979,ginibre1979_2} and by Grillakis, Machedon, and Margetis \cite{grillakis2010,grillakis2011,grillakis2013,grillakis2017} in a slightly different setting, and further results in this direction were obtained, e.g., in \cite{lewin2015,nam2015, boccato2015,mitrouskas2016,chong2016,nam2017,nam2017_2,kuz2017, brennecke2017_2,petrat2017}.
Moreover, dynamical central limit theorems were derived in
\cite{benarous2013,buchholz2014,rademacher2019}.
\medskip

In our main result (Theorem \ref{thm:norm:approx}), we prove that the bound \eqref{Chi:N:trap} is preserved by the time evolution~\eqref{SE}: for any $a\in\N_0$, we show that
\begin{equation}\label{eqn:intro:thm1}
\Big\|\ChiN(t)-\sum\limits_{\l=0}^a\lN^\frac{\l}{2}\Chil(t)\Big\|_{\FNpt}\leq C_a(t)\lN^\frac{a+1}{2}
\end{equation}
for explicitly computable, $N$-independent coefficients $\Chil(t)\in\Fpt$.
The higher order corrections ($\l\geq1$) to the leading order term $\Chio(t)$ are given by
\begin{equation}\label{eqn:intro:thm2}
\Chil(t)=\sum\limits_{\substack{0\leq n\leq 3\l\\ n+\l\text{ even}}}
\sum\limits_{\bj\in\{-1,1\}^{n}}
\int\dx^{(n)}
\mathfrak{C}^{(\bj)}_{\l,n}(t;x^{(n)})\,\,\asjo_{x_1}\,\mycdots\,a^{\sharp_{j_{n}}}_{x_{n}}\,\Chio(t)
\end{equation}
for some explicitly known, $N$-independent coefficients $\mathfrak{C}^{(\bj)}_{\l,n}$ (see Corollary \ref{thm:even_more_explicit_form}), which are constructed iteratively from the initial coefficients $\mathfrak{a}^{(\bj)}_{\l,n}$ from \eqref{Chi:l:trap}.
To derive \eqref{eqn:intro:thm2}, it is crucial to note that the unitary time evolution generated by the Bogoliubov Hamiltonian acts as a Bogoliubov transformation on creation and annihilation operators.

The main advantage of \eqref{eqn:intro:thm2} lies in the fact that it essentially reduces the computation of the higher order corrections $\Chil(t)$ to the problem of solving first the well-studied Hartree equation, and second the Bogoliubov equation, which is equivalent to solving a $2\times2$ matrix differential equation (see Section \ref{subsec:Bogoliubov}). Hence, \eqref{eqn:intro:thm2} solves the dynamics of the mean-field Bose gas as complete as seems possible, to any order in $1/N$, in terms of functions that can be retrieved with a reasonable computational effort which is completely independent of $N$.
\medskip

Although we decided to focus on the situation where the initial state $\PsiN(0)$ is an  eigenstate of the Hamiltonian $\HN_\mathrm{trap}$, our results remain valid in a more general setting, also including, e.g., mixtures of such states. In particular, it is not necessary to assume that the interaction $v$ is of positive type, which is required for the static result \eqref{Chi:N:trap}.
\medskip

As one application, our approximation scheme yields computationally accessible higher order corrections to the trace norm convergence \eqref{eqn:RDM:convergence} of the reduced density matrices. In Theorem \ref{thm:RDM}, we prove that
\begin{equation}\label{eqn:intro:thm:RDM}
\Tr\left|\gamma_{\PsiN}^{(1)}(t)-\sum\limits_{\l=0}^a\lN^\l\gamma_\l^{(1)}(t) \right| \leq C_a(t) \lN^{a+1}\,
\end{equation}
for $N$-independent one-body operators $\gamma_\l^{(1)}(t)$ (see \eqref{eqn:gamma:higher:orders} for a definition). While the leading order in \eqref{eqn:intro:thm:RDM},
\begin{equation}
\gamma_0^{(1)}(t)=|\pt\rangle\langle\pt|\,,
\end{equation}
recovers Bose--Einstein condensation with optimal rate as in \eqref{eqn:RDM:convergence},
none of the higher order terms in the expansion of $\gamma_{\PsiN}^{(1)}(t)$ has---to the best of our knowledge---been known so far, neither in the mathematics nor in the physics community. For example, the first correction to $\gamma_0^{(1)}$ is given by
\begin{equation}
\gamma_1^{(1)}(t)
=|\pt\rangle\langle\bzo(t)|
+|\bzo(t)\rangle\langle\pt|
+\gChiot-\Tr_\fH\gChiot|\pt\rangle\langle\pt|\,,
\end{equation}
where $\bzo(t,x)$ is the solution of 
\begin{eqnarray}\label{intro_red_dens_corr}
\i\partial_t\bzo(t) 
&=&\left(\hpt+\Kopt\right)\bzo(t)+\Ktpt\overline{\bzo(t)} \nonumber\\
&& +\big(\Kthpt\big)^*\aChiot+\Tr_1\big(\Kthpt\gChiot\big)+\Tr_2\big(\Kthpt\gChiot\big)\,.
\end{eqnarray}
Here, the integral operators $K^{(i)}_{\pt}$ are defined in terms of $\pt$ and $v$ in \eqref{K}, and  we used the notation $\Tr_1A:=\int\dz A(z,\,\cdot\,;z)$ and $\Tr_2 A:=\int\dz A(\,\cdot\,,z;z) $, 
for an operator $A:\fH\to\fH^2$. Moreover, $\gamma_{\Chio(t)}$ and $\alpha_{\Chio(t)}$ are the Bogoliubov two-point correlation functions,
\begin{equation}\label{eqn:intro:gamma:alpha}
\gamma_{\Chio(t)}(x,y):=\lr{\Chio(t),\ad_y a_x\Chio(t)}\,,\qquad
\alpha_{\Chio(t)}(x,y):=\lr{\Chio(t),a_xa_y\Chio(t)}\,,
\end{equation}
which are determined by the initial data through a system of two coupled  PDEs \eqref{eqn:PDE} on $\R^{2d}$. In  particular, it is possible to solve the PDEs numerically to arbitrary accuracy. 
\medskip

Finally, \eqref{eqn:intro:thm2}, leads to an asymptotic expansion of the $n$-point correlation functions of the excitations,
\begin{equation}
\lrt{a^{\sharp_{j_1}}_{x_1}\,\mycdots\,a^{\sharp_{j_n}}_{x_n}}_N :=
\lr{\ChiN(t),a^{\sharp_{j_1}}_{x_1}\,\mycdots\,a^{\sharp_{j_n}}_{x_n}\ChiN(t)}_{\FNpt}\,.
\end{equation}
Let us for simplicity consider only the situation where $\PsiN(0)=\PsiN_\mathrm{trap}$ is the ground state of $\HN_\mathrm{trap}$. In this case, we prove in Corollary \ref{thm:correlation functions_simplified} that
\begin{subequations}\label{eqn:intro:Wick}
\begin{eqnarray}
\lrt{a^{\sharp_{j_1}}_{x_1}\cdots a^{\sharp_{j_{2n}}}_{x_{2n}}}_N
&=& \sum_{\l=0}^{a} \lN^{\l} \sum_{m=0}^{2\l} \lr{\Chim(t) ,a^{\sharp_{j_1}}_{x_1}\cdots a^{\sharp_{j_{2n}}}_{x_{2n}} \Chi_{2\l-m}(t)} + \Order(\lN^{a+1})\,,\label{eqn:intro:Wick:even}\\
\lrt{a^{\sharp_{j_1}}_{x_1}\cdots a^{\sharp_{j_{2n+1}}}_{x_{2n+1}}}_N
&=& \sum_{\l=0}^{a-1} \lN^{\l+\frac12} \sum_{m=1}^{2\l+1} \lr{\Chim(t) ,a^{\sharp_{j_1}}_{x_1}\cdots a^{\sharp_{j_{2n+1}}}_{x_{2n+1}} \Chi_{2\l+1-m}(t)} + \Order(\lN^{a+\frac12})\,.\quad\qquad\label{eqn:intro:Wick:odd}
\end{eqnarray}
\end{subequations}
The expansion \eqref{eqn:intro:Wick} essentially reduces the $N$-body problem to the problem of computing the two-point correlation functions of the quasi-free state $\Chio(t)$, since \eqref{eqn:intro:thm2} leads to a generalization of Wick's rule for the ``mixed'' $n$-point correlation functions of the corrections $\Chil(t)$:
If $\l + n + k$ is odd, it holds that
\begin{subequations}\label{eqn:intro:mixed:n:point:fctns}
\begin{equation}\label{nlk_correlations_odd_intro}
\lr{\Chil(t),\asjo_{x_1}\,\mycdots  a^{\sharp_{j_n}}_{x_{n}}\Chik(t)}_\Fpt=0\,,
\end{equation}
and if $\l+n+k$ is even,
\begin{eqnarray}\label{eqn:intro:prop:compute:corr:fctns}
&&\hspace{-0.3cm}\lr{\Chil(t),a^{\sharp_{j_1}}_{x_1}\cdots a^{\sharp_{j_n}}_{x_{n}}\Chi_k(t)}_\Fpt \\
&& =\hspace{-2.5mm} \sum\limits_{\substack{b=n\\\text{even}}}^{n+3(\l+k)} \sum\limits_{\bm\in\{-1,1\}^b}\sum\limits_{\sigma\in P_{b}}\prod\limits_{i=1}^{b/2}
\int\dy^{(b)}\mathfrak{D}^{(\bj;\bm)}_{\l,k,n;b}(t;\xn;y^{(b)}) \scp{\Chio(t)}{a_{y_{\sigma(2i-1)}}^{\sharp_{m_{\sigma(2i-1)}}}a_{y_{\sigma(2i)}}^{\sharp_{m_{\sigma(2i)}}}\Chio(t)}_\Fpt ,\nonumber
\end{eqnarray}
\end{subequations}
where $\mathfrak{D}^{(\bj;\bm)}_{\l,k,n;b}$ is determined by the coefficients $\mathfrak{C}$ from \eqref{eqn:intro:thm2} and $P_{b}$ is a set of pairings (see Corollary \ref{lem_compute_correlation_functions} for the precise statement).
In particular, by \eqref{eqn:intro:prop:compute:corr:fctns}, the right-hand side of  \eqref{eqn:intro:Wick} is explicitly given in terms of the Bogoliubov two-point correlation functions \eqref{eqn:intro:gamma:alpha}, which are obtained by solving a system of two coupled PDEs.

A possible application of this result is a precise computation of the condensate depletion. By \eqref{eqn:intro:Wick:even},
\begin{eqnarray}
&&\hspace{-1cm}\lr{\ChiN(t),\Number\ChiN(t)}_{\FN} \nonumber\\
& =& \lr{\Chio(t),\Number \Chio(t)} + \lN \bigg( \lr{\Chi_1(t),\Number \Chi_1(t)} + \lr{\Chio(t),\Number \Chi_2(t)} + \lr{\Chi_2(t),\Number \Chio(t)} \bigg)\nonumber\\
&& + \Order(\lN^{2}).
\end{eqnarray}
Making use of \eqref{eqn:intro:prop:compute:corr:fctns}, the next order correction beyond Bogoliubov theory can (analytically or numerically) be evaluated. For dilute Bose gases, the condensate depletion has been studied in various settings, and the predictions of Bogoliubov theory
were experimentally confirmed \cite{xu2006, chang2016, lopes2017}.
\medskip

In conclusion, our approximation scheme provides a computationally efficient perturbative algorithm for the analysis of the  mean-field dynamics, which yields $N$-independent corrections.
In the NLS and Gross--Pitaevskii scaling regimes, where the interaction scales as $ N^{-1+3\tilde{\beta}}v(N^{\tilde{\beta}}x)$ for $0<\tilde{\beta} \leq 1$ (in three dimensions), the correlations are larger due to the singular interaction. We conjecture that our method can be extended to regimes with $\tilde{\beta}>0$, at least for small values of $\tilde{\beta}$. However, the $N$-independence of the corrections might be hard to achieve. We also expect that our results can be extended to a class of pair interactions including the repulsive Coulomb potential in three dimensions, but this is certainly harder to prove.
\medskip

We conclude with an overview of closely related results. To the best of our knowledge, the first derivation of higher order corrections is due to Ginibre and Velo \cite{ginibre1980,ginibre1980_2}. They consider the classical field limit $\hbar\to0$ of the dynamics generated by a Hamiltonian $\mathcal{H}$ on Fock space, which is such that $\mathcal{H}\big|_{\fH^N}=\HN$  upon identification $\hbar\equiv 1/N$. Working in the Heisenberg picture, they expand $\mathcal{H}$ in a power series of $\tfrac{1}{\sqrt{N}} a(t)-\pt$, effectively separating the (classical) motion of the condensate from the  excitations,  and study the unitary time evolution $W(t,s)$ of the operators $b(t)=C(\sqrt{N}\varphi(0))^* \big(a(t)-\sqrt{N}\pt \big)C(\sqrt{N}\varphi(0))$. Here, $C$ is the Weyl operator, which implements the $c$-number substitution for coherent initial data (similar to \eqref{eqn:decomp}). The authors construct a Dyson expansion of the unitary group $W(t,s)$ in terms of the time evolution generated by the Bogoliubov Hamiltonian and prove that the expansion is Borel summable for bounded interaction potentials \cite{ginibre1980} and strongly asymptotic for a class of unbounded potentials \cite{ginibre1980_2}.

The main differences to our work are that Ginibre and Velo consider a Hamiltonian on Fock space (instead of $\fH^N$), accordingly use coherent states as initial data, and expand the time evolution operator $W(t,s)$ in a perturbation series (instead of $\PsiN$). 
In contrast, we provide explicit formulas for physically relevant initial data, compute correlation functions and reduced densities, and make use of the connection between Bogoliubov transformations and Bogoliubov maps to reduce the computational complexity.
\medskip

A similar result in the $N$-body setting was derived in \cite{corr} by expanding the $N$-body time evolution in a comparable Dyson series. Instead of the Bogoliubov time evolution, the expansion is in terms of an auxiliary time evolution $\tilde{U}_\varphi(t,s)$ on $\fH^N$, whose generator has a quadratic structure comparable to the Bogoliubov Hamiltonian (sometimes called particle number preserving Bogoliubov Hamiltonian).
However, the auxiliary time evolution $\tilde{U}_\varphi(t,s)$ is a rather inaccessible object, which implicitly still depends on $N$.
In particular, it is not clear to what extent the computation of physical quantities of interest is less complex with respect to the time evolution $\tilde{U}_\varphi(t,s)$ than with respect to the full $N$-body problem.
The approximation scheme we propose in this paper can be understood as an improvement of this result, where we modified the construction precisely such as to make it accessible to computations. \medskip

A related approach to obtain higher order corrections in the mean-field regime was introduced by Paul and Pulvirenti in \cite{paul2019}. In that work, the authors approach the problem from a kinetic theory perspective and consider the dynamics of the reduced density matrices of the $N$-body state. Their approach is formally similar to ours, since Bogoliubov theory in the sense of linearization of the Hartree equation is used for the expansion (but without projecting  onto the excitation Fock space), an even-odd structure similar to \eqref{eqn:intro:mixed:n:point:fctns} is observed, and an $a$-dependent but $N$-independent number of operations is required for the construction. In comparison, the main advantage of our approach is that our approximations \eqref{eqn:intro:thm2} are completely $N$-independent. \medskip

A multi-scale expansion of the ground state of the Bose gas in the mean-field limit was introduced by Pizzo \cite{pizzo2015,pizzo2015_2,pizzo2015_3} via a constructive approach using Feshbach maps. In \cite{spectrum}, perturbation theory for the ground state and low-energy excited states corresponding to our dynamical results is made rigorous (see the discussion above and in Section~\ref{sec:results:PT}). We refer to \cite{spectrum} for further literature in the static case.
\medskip

In addition to \cite{paul2019,corr}, our work is particularly inspired by the works of Grillakis, Machedon and Margetis \cite{grillakis2010, grillakis2011, grillakis2013, grillakis2017}, Nam and Napi\'orkowski \cite{nam2015,nam2017}, and some of the classical references \cite{bach1994,robinson1965,derezinski}, which emphasize the concepts of Bogoliubov transformations and quasi-free states. For example, in \cite{nam2015,nam2017} the authors prove that describing the excitations around the Hartree solution $\pt$ as solutions of the $N$-independent Bogoliubov equation yields a computationally efficient approximation to the many-body dynamics: since the Bogoliubov time evolution preserves quasi-freeness, all time evolved $n$-point correlation functions are---thanks to Wick's rule---determined by the initial two-point correlation functions, up to order $N^{-1/2}$. The result presented here can be understood as a continuation of these ideas. Our corrections $\Chil(t)$ to the Bogoliubov description are still completely determined by the initial two-point correlation functions of the leading order contribution $\Chio(0)$, and in this sense, the corrections $\Chil(t)$ can be understood as generalized quasi-free states which satisfy a generalized form of Wick's rule.

\subsection*{Notation}
\begin{itemize}
\item Expressions that are independent of  $N, t, \PsiN(0),\pz$, but may depend on all fixed quantities of the model such as $v$, are referred to as constants. 
While we explicitly indicate the dependence of constants on the order $a$ of the approximation by writing $C_a$ or $C(a)$, note that these constants may vary from line to line.
We use the notation
\begin{equation}
A\ls B
\end{equation} 
to indicate that there exists a constant $C>0$ such that $A\leq CB$.
\item
Unless indicated otherwise, $\pt$ denotes the solution of \eqref{hpt} with initial datum $\pz$.
\item If a vector in $\Fock$ is written as a direct sum, it is always understood with respect to the decomposition $\Fock=\FN\oplus\Fock^{>N}$. For $\bPhi_1,\bPhi_2,\bPhi\in\Fock$, we use the notation
\begin{equation}
\lr{\bPhi_1,\bPhi_2}_{\FN}:=\sum\limits_{k=0}^N\lr{\phi_1^{(k)},\phi_2^{(k)}}_{\fH^k}\,,\qquad
\norm{\bPhi}_{\FN}^2:=\sum\limits_{k=0}^N\norm{\phi^{(k)}}_{\fH^k}^2\,.
\end{equation}
\item We abbreviate
\begin{equation}
\xk:=(x_1\mydots x_k)\,,\qquad \d\xk:=\dx_1\mycdots\dx_k
\end{equation} 
for $k\geq 1$ and $x_j\in{\R^d}$, and use the notation
\begin{equation}
a^{\sharp_1}:=\ad\,,\qquad a^{\sharp_{-1}}:=a\,.
\end{equation}
We denote by $\bj=(j_1\mydots j_n)$ a multi-index and define $|\bj|:=j_1+\dots+j_n$.
\item The set of pairings of $\{1\mydots 2a\}$ used for Wick's rule is defined as
\begin{equation}\label{Pairing_set_def}
P_{2a}:=\{\sigma\in\mathfrak{S}_{2a}:\sigma(2j-1)<\min\{\sigma(2j),\sigma(2j+1)\} \;\forall\; 1\leq j\leq 2a\},
\end{equation}
where $\mathfrak{S}_{2a}$ denotes the symmetric group on the set $\{1\mydots 2a\}$.
\item $\mathcal{L}(V,W)$ denotes the space of bounded operators from $V$ to $W$, and the corresponding norm is sometimes abbreviated as $\onorm{A} := \norm{A}_{\mathcal{L}(V,W)}$. We denote the Hilbert--Schmidt norm by $\| A \|_{\HS}$.
\end{itemize}

\section{Preliminaries}
\subsection{Framework}
We begin with stating our assumptions on the Hamiltonian $\HN$ from $\eqref{HN}$ and on the initial condensate wave function.

\begin{assumption}\label{assumption:v}
Assume that the interaction potential $v:\R^d\to\R$ is bounded and even, and let $\pz\in H^1(\R^d)$ with $\norm{\pz}_{\fH}=1$. 
\end{assumption}

As a consequence of Assumption \ref{assumption:v},  the $N$-body Hamiltonian $\HN$ generates a unique family of unitary time evolution operators, which  leaves $\mathcal{D}(\HN)=H^2({\R^d})$ invariant. \medskip

The evolution of the condensate is described by the solution $\pt$ of the Hartree equation,
\begin{equation}\label{hpt}
\i \partial_t \pt = \left(-\Delta + v*|\pt|^2 - \mpt\right)\pt
=: \hpt\pt,
\qquad \varphi(0)=\varphi_\trap\,,
\end{equation}
where the real-valued phase factor $\mpt$ is given by
\begin{equation}\label{mpt}
\mpt = \tfrac{1}{2} \int\limits_{\R^d} \left(v*|\varphi(t)|^2\right)(x) |\varphi(t,x)|^2\dx \,.
\end{equation} 
The solution of \eqref{hpt} in $H^1(\R^d)$ is unique and exists globally \cite[Corollaries 4.3.3 and 6.1.2]{cazenave}. Moreover, it holds that $\norm{\pt}_\fH=\norm{\pz}_\fH=1$.
\medskip

We define the  (truncated) Fock spaces
\begin{equation}\label{Fock:space}
\FNpt:= \bigoplus_{k=0}^N\bigotimes_\sym^k \{\pt\}^\perp\,,\quad 
\Fpt:=\bigoplus_{k=0}^\infty\bigotimes_\sym^k \{\pt\}^\perp\,,\quad
\FH:=\bigotimes\limits_{k=0}^\infty\bigotimes\limits_\sym^k\fH\,,
\end{equation}
for $\{\varphi\}^\perp:=\left\{\phi\in \fH: \lr{\phi,\varphi}_{\fH}=0\right\}$. 
Whenever the context is unambiguous, we abbreviate
\begin{equation}
\Fock:=\FH\,.
\end{equation}
The Fock vacuum is denoted by $|\Omega\rangle=(1,0,0,\dots)$.
Note that the subspaces $\FNpt$ and $\Fpt$ depend on time via the condensate wave function, while the full Fock space $\Fock$ is time-independent.
We refer to $\Fpt$ as \emph{excitation Fock space}, and to $\FNpt$ as \emph{truncated excitation Fock space}.
The one-particle sector $\{\pt\}^\perp$ of the subspace $\Fpt$ can be understood as a hyperplane in $\fH$, which, due to the time evolution of the condensate wave function $\pt$, ``rotates'' through $\fH$; the higher sectors are (symmetrized) direct sums of these hyperplanes.
\\

The creation and annihilation operators on $\Fock$ are defined as
\begin{eqnarray}
(\ad(f)\bPhi)^{(k)}(x_1\mydots x_k)&=&\frac{1}{\sqrt{k}}\sum\limits_{j=1}^kf(x_j)\phi^{(k-1)}(x_1\mydots x_{j-1},x_{j+1}\mydots x_k)\,,\quad k\geq 1\,,\\
(a(f)\bPhi)^{(k)}(x_1\mydots x_k)&=&\sqrt{k+1}\int\d x\overline{f(x)}\phi^{(k+1)}(x_1\mydots x_k,x)\,,\quad k\geq 0
\end{eqnarray} 
for $f\in\fH$ and $\bPhi\in\Fock$.
They can be expressed in terms of the operator-valued distributions $\ad_x$, $a_x$,
satisfying the canonical commutation relations 
\begin{equation}\label{eqn:CCR}
[a_x,\ad_y]=\delta(x-y)\,,\qquad [a_x,a_y]=[\ad_x,\ad_y]=0\,,
\end{equation}
as
\begin{equation}
\ad(f)=\int\d x f(x)\,\ad_x\,,\qquad a(f)=\int\d x\overline{f(x)}\,a_x\,.
\end{equation}
The second quantization in $\Fock$  of a $k$-body operator $A^{(k)}$ is  defined as
\begin{equation}
\d\Gamma(A^{(k)}) = \int \dx^{(k)} \int \dy^{(k)} A^{(k)}(x^{(k)};y^{(k)})\, \ad_{x_1} \ldots \ad_{x_k}  a_{y_1} \ldots a_{y_k},
\end{equation}
and the number operator on $\Fock$ is denoted by 
\begin{equation}
\Number=\int\ad_x a_x\dx\,,\qquad (\Number\bPhi)^{(k)}=k\phi^{(k)}\,.
\end{equation}
The  number of excitations around the time-evolved condensate $\pt^{\otimes N}$ is counted by the excitation number operator,
\begin{equation}
\Npt:=\Number-\ad(\pt)a(\pt) = \d\Gamma(\qpt)\,,
\end{equation}
where $\qpt$ denotes the projection onto the orthogonal complement of $\pt$, i.e.,
\begin{equation}
\ppt:=|\pt\rangle\langle\pt|\,,\qquad \qpt:=\id_\fH-\ppt\,.
\end{equation}
Thus, when restricted to vectors $\bPhi(t)\in\Fpt$, $\Number$ counts the number of excitations, 
\begin{equation}
\Number|_{\Fpt}=\Npt|_{\Fpt}\,.
\end{equation}
Note that $\Number$ is not time-dependent, while $\Npt$ depends on time.
The relation between the $N$-body state $\PsiN(t)$ and the corresponding excitation vector $\ChiN(t)$ is given by the unitary map 
\begin{eqnarray}\label{map:U}
\UNpt:\fH_\sym^N \to  \FNpt\;, \quad
\PsiN(t)  \mapsto  \UNpt\PsiN(t):=\ChiN(t)\,.
\end{eqnarray}
Equivalently, $\UNpt$ can be interpreted as a partial isometry from $\fH_\sym^N$ to $\Fock$, where $\UNpt^*$ is extended by zero outside the truncated excitation Fock space $\FNpt$.
The product state $\pt^{\otimes N}$ is mapped to the vacuum, i.e.,
\begin{equation}
\UNpt \pt^{\otimes N}=(1,0,0,\dots)\,.
\end{equation}
Written in terms of creation and annihilation operators, the map $\UNpt$ acts as
\begin{equation}\label{LNSS:map:U}
\UNpt\PsiN(t)=\bigoplus\limits_{j=0}^N\big(\qpt\big)^{\otimes j}\left(\frac{a(\pt)^{N-j}}{\sqrt{(N-j)!}}\PsiN(t)\right)
\end{equation}
(\cite[Proposition 4.2]{lewin2015_2}).
This leads for $f,g\in\{\pt\}^\perp$ to
\begin{subequations}\label{eqn:substitution:rules}
\begin{eqnarray}
\UNpt \ad(\pt)a(\pt)\UNpt^*&=&N-\Npt\,,\\
\UNpt \ad(f)a(\pt) \UNpt^*&=&\ad(f)\sqrt{N-\Npt}\,,\\
\UNpt \ad(\pt)a(g)\UNpt^*&=&\sqrt{N-\Npt}\, a(g)\,,\\
\UNpt \ad(f)a(g)\UNpt^*&=&\ad(f)a(g)
\end{eqnarray}
\end{subequations}
as identities on $\FNpt$. Note that since $\Npt=\Number$ on $\Fpt$, we can equivalently replace $\Npt$ by $\Number$ in \eqref{eqn:substitution:rules}. 
\medskip

\subsection{Excitation Hamiltonian}
The map $\UNpt$ decouples the evolution \eqref{hpt} of the condensate wave function from the evolution of the excitations $\ChiN(t)=\UNpt\PsiN(t)$, which is determined by
\begin{equation}\label{eqn:SE:Fock}
\i\partial_t\ChiN(t)=\FockHNpt\ChiN(t)\,, 
\end{equation}
with excitation Hamiltonian
\begin{equation}\label{H:Fock}
\FockHNpt=\i(\partial_t\UNpt)\UNpt^*+\UNpt\HN\UNpt^*:\FNpt\to\FN\,.
\end{equation}
For later convenience, we write $\FockHNpt$ as restriction to $\FNpt$ of a Hamiltonian $\FockHpt$ which is defined on the full Fock space $\Fock$, 
\begin{equation}
\FockHNpt\;=\;\FockHpt\big|_{\FNpt}.
\end{equation}

The transformation rules \eqref{eqn:substitution:rules} can be used to obtain an explicit formula for $\FockHpt$ from \eqref{H:Fock} (see, e.g., \cite[Section 4]{lewin2015_2} and \cite[Section 4.2 and Appendix B]{lewin2015}). Written in a way that is more convenient for our analysis, it is given as
\begin{eqnarray}
\FockHpt
&=& \boldKzpt + \frac{N-\Number}{N-1}\boldKopt \nonumber\\
&&+\left(\boldKtpt\frac{\sqrt{\left[(N-\Number)(N-\Number-1)\right]_+}}{N-1}+ \frac{\sqrt{\left[(N-\Number)(N-\Number-1)\right]_+}}{N-1}\boldKtptbar\right)\nonumber \\
&&+\left(\boldKthpt\frac{\sqrt{\left[N-\Number\right]_+}}{N-1}+ \frac{\sqrt{\left[N-\Number\right]_+}}{N-1}\boldKthptbar\right)  +\frac{1}{N-1}\boldKfpt\,,\label{eq_H_N_Fock}
\end{eqnarray}
where $[\cdot]_+$ denotes the positive part and where we used the shorthand notation
\begin{subequations}\label{eqn:K:notation}
\begin{eqnarray}
\boldKzpt&=&\int\dx\,\ad_x \hpt_x a_x\,,\label{eqn:K:notation:0}\\
\boldKopt&=&\int\dx_1\dx_2\,\Kopt(x_1;x_2)\ad_{x_1} a_{x_2}\,,\label{eqn:K:notation:1}\\
\boldKtpt&=&\tfrac12\int\dx_1\dx_2\,\Ktpt(x_1,x_2)\ad_{x_1}\ad_{x_2}\,,\label{eqn:K:notation:2}\\
\boldKthpt&=&\int\dx^{(3)}\,\Kthpt(x_1,x_2;x_3)\ad_{x_1}\ad_{x_2}a_{x_3}\,\label{eqn:K:notation:3}\\
\boldKfpt&=&\tfrac12\int\dx^{(4)}\,\Kfpt(x_1,x_2;x_3,x_4)\ad_{x_1}\ad_{x_2}a_{x_3}a_{x_4}\,,\label{eqn:K:notation:4}\\
\overline{\mathbb{K}^{(j)}_\pt}&:=& \left(\mathbb{K}^{(j)}_\pt\right)^*\,,\qquad j=2,3\label{eqn:K:notation:5}
\end{eqnarray}
\end{subequations}
with
\begin{subequations}\label{K}
\begin{align}
\label{K:1}
& \Kopt:\{\pt\}^\perp\to \{\pt\}^\perp, \qquad
 \Kopt :=   \qpt   \tKopt  \qpt \,, 
 \\[7pt]
\label{K:2}
& \Ktpt \in \{\pt\}^\perp\otimes \{\pt\}^\perp,\qquad
 \Ktpt(x_1,x_2) := 
   \big(\qpt_1\qpt_2 \tKtpt\big)(x_1,x_2)\,,
\\[7pt]
& \Kthpt:\{\pt\}^\perp \to \{\pt\}^\perp\otimes \{\pt\}^\perp,\;\;\;
\nonumber\\[3pt]
&\hspace{1.3cm}\psi\mapsto (\Kthpt\psi)(x_1,x_2):=\qpt_1\qpt_2 \Wpt(x_1,x_2)\varphi(x_1)(\qpt\psi)(x_2)
\label{K:3}\,,\\[5pt]
& \Kfpt:\{\pt\}^\perp\otimes \{\pt\}^\perp\to\{\pt\}^\perp\otimes \{\pt\}^\perp,\nonumber\\[3pt]
&\hspace{1.3cm}\psi\mapsto (\Kfpt\psi)(x_1,x_2):=\qpt_1\qpt_2\Wpt(x_1,x_2)(\qpt_1\qpt_2\psi)(x_1,x_2)\,.\label{K:4}
\end{align}
\end{subequations}
Here, $\tKopt$ is the operator with kernel
\begin{equation}
\tKopt(x_1;x_2) := \overline{\varphi(t,x_2)}v(x_1-x_2)\varphi(t,x_1)\,,
\end{equation}
$\tKtpt$ denotes the vector
\begin{equation}
 \tKtpt(x_1,x_2) := v(x_1-x_2)\varphi(t,x_1)\varphi(t,x_2)\,,
\end{equation} 
and $\Wpt$ is the multiplication operator on $\fH\otimes\fH$ defined by
\begin{equation}\label{eqn:W(t,x,y)}
\Wpt(x,y) := v(x-y) - \left(v*|\pt|^2\right)(x) - \left(v*|\pt|^2\right)(y) + 2\mpt.
\end{equation}
The notation is understood such that the projections $\qpt_1,\qpt_2$ act on the respective functions on their right. For example, the function $\Kthpt\psi$  is obtained from $\psi\in\fH$ by taking the tensor product of $\qpt\psi$ and $\pt$, acting on it with the multiplication operator $\Wpt$, and finally projecting the resulting function on $\fH\otimes\fH$ onto the subspace $\{\pt\}^\perp\otimes\{\pt\}^\perp$.
In Appendix~\ref{appendix:Hamiltonian}, we provide a derivation of \eqref{eq_H_N_Fock}.

Let us remark that there is more than one way to extend  $\FockHNpt$ to the full Fock space $\Fock$. In particular, we could have defined $\FockHpt$ in terms of $\Npt$ instead of $\Number$ since the two operators coincide on $\FNpt$. We chose this way of defining $\FockHpt$ for later convenience.
\medskip

Expanding the $N$-dependent expressions in \eqref{eq_H_N_Fock} in a Taylor series, we write the excitation Hamiltonian $\FockHpt$ as a power series in $\lN^{1/2}$  with $N$-independent (operator-valued) coefficients (see  Section~\ref{subsec:proofs:duhamel} for a proof):

\begin{lem}\label{lem:expansion:HNpt}
Let $a\in\N_0$.
In the sense of operators on $\Fock$, it holds that
\begin{eqnarray}\label{eqn:lem:expansion:HNpt}
\FockHpt
&=&\sum\limits_{n=0}^{a}\lN^{\frac{n}{2}}\FockHnpt +\lN^{\frac{a+1}{2}}\mathcal{R}^{(a)}\,,
\end{eqnarray}
where
\begin{subequations}\label{FockHnpt}
\begin{eqnarray}
\FockHopt&:=&  \boldKzpt + \boldKopt + \boldKtpt + \boldKtptbar\,,\label{FockHopt}\\
\nonumber\\
\FockHtpt&:=& \boldKthpt + \boldKthptbar\,,\label{FockHtpt}\\
\nonumber\\
\FockHthpt
&:=& -(\Number-1)\boldKopt -\boldKtpt(\Number-\tfrac12)-(\Number-\tfrac12)\boldKtptbar+\boldKfpt\,,\label{FockHthpt}\\
\FockHpt^{(2n-1)}
&:=&c_{n-1}\left(\boldKthpt (\Number-1)^{n-1} +  (\Number-1)^{n-1}\boldKthptbar\right)\,,\label{FockHt:2n}\\
\nonumber\\
\FockHpt^{(2n)}&:=&
\sum\limits_{\nu=0}^n d_{n,\nu}\left(\boldKtpt(\Number-1)^\nu +(\Number-1)^\nu\boldKtptbar\right)\label{FockHt:2n+1}
\end{eqnarray}
\end{subequations}
for $n\geq2$ and with
\begin{eqnarray}\label{eqn:taylor:coeff}
c^{(\l)}_0&:=&1\,,\quad c^{(\l)}_n := \frac{(\l-\frac12)(\l+\frac12)(\l+\frac32)\mycdots(\l+n-\frac32)}{n!}\,,\quad c_n:=c_n^{(0)} \quad (n\geq 1)\,,\\
\label{eqn:taylor:coeff:2}
d_{n,\nu}&:=&\sum\limits_{\l=0}^{\nu} c_\l^{(0)} c_{\nu-\l}^{(0)} c_{n-\nu}^{(\l)}\quad (n\geq \nu \geq0)\,.
\end{eqnarray}
The remainder $\mathcal{R}^{(a)}$ satisfies
\begin{equation}
\norm{\mathcal{R}^{(a)}\bPhi}_\Fock\leq C(a) \norm{(\Number+1)^\frac{a+4}{2}\bPhi}_\Fock
\end{equation}
for some constant $C(a)>0$ and any $\bPhi\in\Fock$.
\end{lem}

The leading order $\FockHopt$ is the Bogoliubov Hamiltonian.
The higher orders $\FockHnpt$ contain  powers of $\Number$, which originate from the Taylor expansions of the square roots in $\FockHpt$. 

The operator $\FockHpt$ preserves the truncation of $\FN$ (see Lemma \ref{lem:HN:preserves:truncation}), whereas this property is lost upon  expansion of the square roots. For example, the second term of $\FockHopt$ acting on the $N$-particle sector creates two new particles, resulting in a non-zero $(N+2)$-component.
The single ``fluxes'' out of the truncated Fock space cancel with the remainder, such that the whole expression is truncation preserving.
Moreover, the effect is very small since the occupation of the $N$-particle sector is negligible for large $N$ (see Lemma~\ref{lem:moments:Chil:N}). \medskip

\subsection{Assumptions}
In this section, we present and discuss our assumptions on the initial $N$-body wave function.

\begin{assumption}\label{ass:initial:data}
Let Assumption \ref{assumption:v} hold and let $\tilde{a}\in\N_0$. Let $\PsiN(0)\in \fH^N_\sym$, define $\ChiN(0)=\UNpz\PsiN(0)$, and assume that there exists a constant $C(\tilde{a}) > 0$ such that
\begin{equation}\label{eqn:initial:expansion}
\left\|\ChiN(0)-\sum\limits_{\l=0}^{\tilde{a}} \lN^\frac{\l}{2}\Chil(0)\right\|_{\FN}\leq C(\tilde{a}) \, \lN^{\frac{\tilde{a}+1}{2}}\,,
\end{equation}
where the functions $\Chil(0)$ are defined as follows:
\begin{itemize}
\item  
Let $\tilde{\nu}\in\N_0$, let  $\BogUz$  be a Bogoliubov transformation on $\Fpz$, and let $\{f_j\}_{j=1}^{\tilde{\nu}}\subset\{\pz\}^\perp$ be some orthonormal system (or the empty set for $\tilde{\nu}=0$). Define
\begin{equation}\label{def:Chi0:0}
\Chi_0(0) : = \BogUz  a^{\dagger}\big(f_1 \big) \,\mycdots\, a^{\dagger}\big(f_{\tilde{\nu}}\big) \,|\Omega\rangle\,.
\end{equation}
\item 
For $1\leq\l\leq \tilde{a}$, define iteratively 
\begin{eqnarray}
\Chil(0)
&:=&\sum\limits_{n=1}^\l \;
\sum\limits_{\substack{0\leq m\leq n+2\\m+n\text{ even}}}\;\,
\sum\limits_{\mu=0}^m
\int\dx^{(\mu)}\dy^{(m-\mu)}
\mathfrak{a}^{(\l)}_{n,m,\mu}\big(x^{(\mu)};y^{(m-\mu)}\big) \nonumber\\
&&\hspace{5cm} \times \ad_{x_1}\,\mycdots \ad_{x_\mu} a_{y_1}\,\mycdots a_{y_{m-\mu}}\, \Chi_{\l-n}(0)\,,\label{def:Chil:0} 
\end{eqnarray}
where $\mathfrak{a}^{(\l)}_{n,m,\mu}\big(x^{(\mu)};y^{(m-\mu)}\big)$ are the (Schwartz) kernels of some $N$-independent operators.
\end{itemize}
\end{assumption}

We stated \eqref{def:Chil:0} in this iterative form for later convenience (in particular, this simplifies the proof of Corollary \ref{thm:even_more_explicit_form}). By solving the iteration and bringing the result into normal order, one can show that \eqref{def:Chil:0}  is equivalent to the statement that
\begin{equation}\label{eqn:Chil:0:normal:order}
\Chil(0)
=
\sum\limits_{\substack{0\leq m\leq 3\l\\m+\l\text{ even}}}\;\,
\sum\limits_{\mu=0}^m
\int\dx^{(\mu)}\dy^{(m-\mu)}
\tilde{\mathfrak{a}}^{(\l)}_{m,\mu}\big(x^{(\mu)};y^{(m-\mu)}\big)\ad_{x_1}\,\mycdots \ad_{x_\mu} a_{y_1}\,\mycdots a_{y_{m-\mu}}\, \Chi_{0}(0)
\end{equation}
for some bounded operators $\tilde{\mathfrak{a}}^{(\l)}_{m,\mu}:\fH^{m-\mu}\to\fH^\mu$.

Since the operators $\mathfrak{a}^{(\l)}_{n,m,\mu}$ are $N$-independent, the only $N$-dependence in the initial state is due to the powers of $\lN^{1/2}$ in the series expansion. While $\Chio(0)$ is normalized, this is not the case for $\Chil(0)$ for $\l\geq 1$, hence $\PsiN(0)$ need not be normalized. To recover a normalized initial $N$-body wave function $\PsiN(0)$, it suffices to divide $\ChiN(0)$ and all corrections $\Chil(0)$ by a normalization factor $\norm{\ChiN(0)}_{\FN}=\norm{\ChiN(t)}_{\FN}$. This situation is often encountered in perturbation theory (see, e.g., \cite{sakurai}).

Our analysis generalizes to the case where $\Chio(0)$ is given as a linear combination of Bogoliubov transformed states with different particle numbers $\tilde{\nu}$. 
We refrain from including such initial states in Assumption \ref{ass:initial:data} to simplify the formulas (especially Corollary \ref{lem_compute_correlation_functions}).\medskip

In \cite{spectrum}, it is shown that Assumption~\ref{ass:initial:data} is satisfied if $\PsiN(0)$ is the ground state  or a suitable low-energy eigenstate of $\HN_\trap$ as in \eqref{H:N:trap}. We formulate the precise statement and discuss degeneracies and an example in Appendix~\ref{appendix:ini_data}.

By Assumption \ref{ass:initial:data},  finite moments of $\Number$ with respect to $\Chil(0)$ and $\ChiN(0)$ are bounded uniformly in $N$ (see Section \ref{subsec:proofs:cor:ass} for a proof):

\begin{lem}\label{cor:ass}
Let Assumption  \ref{ass:initial:data} hold for some $\tilde{a}\in\N_0$ and denote $\ChiN(0)=\UNpz\PsiN(0)$.
\corit{
\item \label{cor:ass:moments_l}
Then for all $b\geq 0$ and $0\leq\l\leq\tilde{a}$ there exists a constant  $C(\l,b) > 0$ such that
\begin{equation}\label{eqn:cor:ass:moments_l}
\lr{\Chil(0),\left(\Number+1\right)^{b}\Chil(0)}_{\Fock}\leq C(\l,b)\,.
\end{equation}
\item \label{cor:ass:moments_leqN}
For all $0 \leq b \leq \tilde{a}+1$ there exists a constant $C(b) > 0$ such that
\begin{equation}
\lr{\ChiN(0), \left(\Number+1\right)^{b}\ChiN(0)}_{\FN}\leq C(b)\,.
\end{equation}
}
\end{lem}

\section{Results: Dynamical Perturbation Theory}\label{sec:results:PT}

\subsection{Norm Approximation to Any Order}
\label{subsec:results:explicit}
By Assumption \ref{ass:initial:data}, the initial excitation vector $\ChiN(0)\oplus0$ admits an expansion in powers of $\lN^{1/2}$. With the formal ansatz
\begin{equation}\label{eqn:ansatz:ChiN(t)_later}
\ChiN(t)\oplus0=\sum\limits_{\l=0}^\infty \lN^{\frac{\l}{2}} \Chil(t)\,,
\end{equation}
the Schrödinger equation \eqref{eqn:SE:Fock} leads to the set of equations
\begin{eqnarray}\label{eqn:differential:form:Chil}
\i\partial_t\Chil(t)&=&\FockHopt\Chil(t)+\sum\limits_{n=1}^\l\FockHnpt\Chi_{\l-n}(t)\,.
\end{eqnarray}
The decomposition of $\ChiN(t)\oplus 0$ into the wave functions $\Chil(t)$ is according to orders in $\lN^{1/2}$ and does not relate to the number of excitations contained in $\Chil(t)$, as $\Chil(t)$ is not necessarily a state with a fixed particle number. Due to the inhomogeneity in \eqref{eqn:differential:form:Chil}, the norm of $\Chil(t)$ changes with time for $\l \geq 1$. However, the norm of the perturbation series is conserved in any order, i.e., $\sum_{m = 0}^{\l} \langle \Chi_{\l-m}(t), \Chi_m(t) \rangle$ does not change in time.

For $\l=0$, \eqref{eqn:differential:form:Chil} recovers the Bogoliubov equation \eqref{eqn:Bogoliubov:equation}. 
The time evolution generated by $\FockHopt$ is well-posed and acts as a Bogoliubov transformation $\BogUts$ on $\Fock$. The corresponding Bogoliubov map $\BogV(t,s)$ on $\fH\oplus\fH$ is determined by the differential equation
\begin{equation}\label{eqn:V(t,s):early}
\i\partial_t\BogV(t,s)=\mathcal{A}(t)\BogV(t,s)\,, \qquad
\BogV(s,s)=\id
\end{equation}
with
\begin{equation}
\BogV(t,s)=\begin{pmatrix}
U_{t,s} &\Vbar_{t,s} \\ V_{t,s} & \Ubar_{t,s} \end{pmatrix}\,,\qquad
\mathcal{A}(t)=\begin{pmatrix} 
\hpt+\Kopt& -\Ktpt \\ 
\overline{\Ktpt}& -\left(\hpt+\overline{\Kopt}\right) \end{pmatrix} \,,
\end{equation}
see Lemma~\ref{lem:time:dep:BT}. In Section \ref{subsec:Bogoliubov}, we give an overview of Bogoliubov transformations and collect the corresponding rigorous results.
Written in integral form, \eqref{eqn:differential:form:Chil} motivates the following definition:

\begin{definition}\label{def:Chil}
Let $\BogV(t,s)$ be the solution of \eqref{eqn:V(t,s):early}, denote by $\BogUts$ the corresponding Bogoliubov transformation on $\Fock$, and define $\FockHnpt$ as in \eqref{FockHnpt}. For any $\l\in\N_0$, we define iteratively
\begin{equation}\label{eqn:int:form:Chil}
\Chil(t): = \BogUtz\Chil(0)-\i\sum\limits_{n=1}^\l\,\int\limits_0^t\BogUts\, \FockHnps\,\Chi_{\l-n}(s)\ds \,.
\end{equation}
\end{definition}

By unitarity of $\BogUts$, this is equivalent to the formula
\begin{equation}\label{eqn:int:form:Chil:H:tilde}
\Chil(t) = \BogUtz\Chil(0) - \i \sum\limits_{n=1}^\l \int\limits_0^t \tilde{\mathbb{H}}^{(n)}_{t,s}\, \mathcal{U}_{\BogV(t,s)} \Chi_{\l-n}(s) \ds,
\end{equation}
where
\begin{equation}\label{definition_of_Bog_trafo_Hns}
\tilde{\mathbb{H}}^{(n)}_{t,s} := \BogUts\FockH^{(n)}_\ps\BogUts^*\,.
\end{equation}
Iterating \eqref{eqn:int:form:Chil:H:tilde} yields the following result:

\begin{proposition}\label{thm:duhamel}
Let $\l\in\N_0$, let Assumption \ref{assumption:v} hold and let $\Chi_n(0)\in\Fpz\cap\D(\Number^{3(\l-n)/2})$ for $0\leq n\leq\l$. Then $\Chil(t)\in\Fpt$ and
\begin{eqnarray}
\Chil(t) &=& \BogUtz \Chil(0) \nonumber\\
&&+ \sum_{n=0}^{\l-1} \sum_{m=1}^{\l-n} \sum_{\substack{\bj \in \N^m \\ |\bj| = \l - n}} (-\i)^m \int\limits_0^t \ds_1 \int\limits_0^{s_1} \ds_{2}\, \mycdots\hspace{-5pt} \int\limits_0^{s_{m-1}} \ds_m \, \tilde{\mathbb{H}}^{(j_1)}_{t,s_1} \mycdots \,\tilde{\mathbb{H}}^{(j_m)}_{t,s_m} \, \BogUtz \Chi_n(0)\,,\qquad\label{solution_chil_integral_form}
\end{eqnarray}
where $\tilde{\mathbb{H}}^{(n)}_{t,s}$ from \eqref{definition_of_Bog_trafo_Hns} can be computed as
\begin{equation}\label{eqn:tilde:H}
\tilde{\mathbb{H}}^{(n)}_{t,s} = \sum\limits_{\substack{2\leq p\leq n+2\\p+n\text{ even}}} \sum\limits_{\bj\in\{-1,1\}^p} \int\dx^{(p)} \mathfrak{A}^{(\bj)}_{n,p}(t,s;x^{(p)})\, \asjo_{x_1}\,\mycdots\,a^{\sharp_{j_p}}_{x_p}\,.
\end{equation}
The coefficients $\mathfrak{A}^{(\bj)}_{n,p}$ are given in terms of the kernels of $\Kopt$ to $\Kfpt$ and the matrix entries $U_{t,s}$ and $V_{t,s}$ of $\BogV(t,s)$. They are explicitly stated in \eqref{eqn:A:coeff}. 
\end{proposition}

We postpone the proof of Proposition \ref{thm:duhamel} to Section \ref{subsec:proofs:duhamel}.
The addition of sufficiently many corrections $\lN^{\l/2}\Chil(t)$ approximates the  excitation vector $\ChiN(t)$ in norm to arbitrary precision. This is our main result:

\begin{theorem}\label{thm:norm:approx}
Let $a\in\N_0$, let Assumption \ref{ass:initial:data} hold for $\tilde{a}=a$ and denote $\ChiN(t)=\UNpt\PsiN(t)$.
For $\Chil(t)$ as in Definition \ref{def:Chil}, there exists a constant $C(a)>0$ such that
\begin{equation}\label{eqn:thm:norm:approx}
\Big\|\ChiN(t)-\sum\limits_{\l=0}^a \lN^{\frac{\l}{2}}\Chil(t)\Big\|_{\FN}\leq \e^{C(a) t} \lN^{\frac{a+1}{2}}
\end{equation}
for all $t\in \R$ and sufficiently large $N$. Consequently, the solution $\PsiN(t)$ of \eqref{SE} is approximated with increasing accuracy by the sequence $\left\{\PsiNl(t)\right\}_\l\subset\fH^N$ of $N$-body wave functions
\begin{equation}\label{eqn:PsiN_l}
\PsiNl(t):=\sum\limits_{k=0}^N\pt^{\otimes(N-k)}\otimes_s\chi_\l^{(k)}(t)\,,
\end{equation}
i.e., 
\begin{equation}\label{eqn:thm:norm:approx_Psi}
\Big\|\PsiN(t)-\sum\limits_{\l=0}^a\lN^{\frac{\l}{2}}\PsiNl(t)\Big\|_{\fH^N}\leq \e^{C(a)t} \lN^{\frac{a+1}{2}}\,
\end{equation}
for sufficiently large $N$.
\end{theorem}
We prove Theorem \ref{thm:norm:approx} in Section \ref{subsec:proofs:norm:approx} via a Gronwall argument.

\begin{remark}
\remit{
\item As explained in the introduction, Theorem \ref{thm:norm:approx} is comparable to the results obtained in \cite{ginibre1980,ginibre1980_2,corr}, all of which correspond to  expansions of the unitary group governing the dynamics of the excitations.
In contrast, we focus on the dynamics of initial data satisfying Assumption \ref{ass:initial:data}. This simplifies the approximation since fewer terms are required at a given order $a$ because the state is expanded simultaneously with the Hamiltonian.

\item More precisely, the error in \eqref{eqn:thm:norm:approx_Psi} is of the form 
$$(C_1(a+1))^{C_2(a+1)^2} \e^{C_3(a+1)^2 t} \lN^\frac{a+1}{2}$$
for some constants $C_1,C_2,C_3>0$. Hence, \eqref{eqn:thm:norm:approx_Psi} is not 
uniform in $a$, and, in particular, does not imply convergence of the perturbation series for fixed $N$ as $a\to\infty$. Since our expansion is designed as a large $N$ expansion, we  do not expect this to hold true.  We remark that Borel summability was shown for a comparable but combinatorially simpler expansion in~\cite{ginibre1980}.

\item Theorem \ref{thm:norm:approx} holds under more general conditions than  Assumption \ref{ass:initial:data}. More precisely, it suffices to assume that the initial excitation vector $\ChiN(0)$ admits an asymptotic expansion of the form \eqref{eqn:initial:expansion}, where the coefficients $\Chil(0)$ must satisfy \eqref{eqn:cor:ass:moments_l} for $b=4a+4$ but need not be given by \eqref{def:Chi0:0} and \eqref{def:Chil:0}.
\item We expect that our result can be extended to interactions of the form $\lN N^{d\tilde{\beta}}v(N^{\tilde{\beta}} x)$ for sufficiently small values of $\tilde{\beta}$, comparable to the range $\tilde{\beta}\in[0,1/(4d))$ covered in \cite{corr}.
\item We made no effort to use dispersive estimates and always bound the solution $\pt$ of \eqref{hpt} by its $\fH$-norm, which is possible since $\norm{v}_\infty\ls 1$. As a consequence, our estimates work in all dimensions $d\geq 1$ and apply in the focusing as well as in the defocusing case. 
\item In \eqref{eqn:thm:norm:approx}, one could equivalently identify $\ChiN$ with $\ChiN\oplus 0$ and use the norm on the full Fock space since the contribution of the sectors with more than $N$ particles to $\norm{\Chil(t)}_\Fock$ is negligible (Lemma~\ref{lem:moments:Chil:N}).
}
\end{remark}

\subsection{Simplified Form of the Perturbation Series}\label{sec:explicit_form}

By Assumption~\ref{ass:initial:data}, Proposition~\ref{thm:duhamel} yields a formula for the corrections $\Chil(t)$ in terms of only the initial condition $\Chio(0)$. 
\begin{cor}\label{thm:even_more_explicit_form}
Let Assumption~\ref{ass:initial:data} hold for some $\tilde{a}\in\N_0$ and let $\l\leq\tilde{a}$. Then it holds for $\Chil(t)$ as in Definition \ref{def:Chil} that
\begin{equation}\label{eqn:Chil:explicit:final}
\Chil(t)=\sum\limits_{\substack{0\leq n\leq 3\l\\ n+\l\text{ even}}}
\sum\limits_{\bj\in\{-1,1\}^{n}}
\int\dx^{(n)}
\mathfrak{C}^{(\bj)}_{\l,n}(t;x^{(n)})\,\,\asjo_{x_1}\,\mycdots\,a^{\sharp_{j_{n}}}_{x_{n}}
\Chio(t)\,,
\end{equation}
where $\Chio(t)$ is the solution of the Bogoliubov equation \eqref{eqn:Bogoliubov:equation} and where the $N$-independent coefficients $\mathfrak{C}_{\l,n}^{(\bj)}$ are defined in \eqref{eqn:C:coeff}. 
\end{cor}

We prove Corollary \ref{thm:even_more_explicit_form} in Section~\ref{subsec:proofs:explicit} by iteratively constructing the coefficients $\mathfrak{C}_{\l,n}^{(\bj)}$. They are given in terms of the kernels $\Kopt$ to $\Kfpt$ and the matrix entries $U_{t,s}$ and $V_{t,s}$ of the solution $\BogV(t,s)$ of \eqref{eqn:V(t,s):early}.
For example, the first order correction to $\Chio(t)$ is given as 
\begin{equation}
\Chi_1(t)
=\left(\sum\limits_{j\in\{-1,1\}}
\int\dx\, \mathfrak{C}^{(j)}_{1,1}(t;x)\,a^{\sharp_j}_{x}\,
+ \sum\limits_{\bj\in\{-1,1\}^3}
\int\dx^{(3)}
\mathfrak{C}^{(\bj)}_{1,3}(t;x^{(3)})\,\asjo_{x_1}\,\asjt_{x_2}\,a^{\sharp_{j_{3}}}_{x_{3}}\right)\Chio(t)
\end{equation}
with
\begin{subequations}
\begin{eqnarray}
\mathfrak{C}_{1,1}^{(j)}(t;x)
&=& \int\dy \left(\mathfrak{a}^{(1)}_{1,1,0}(y)\omega^{(-1,j)}_{t,0}(y;x)
+\mathfrak{a}^{(1)}_{1,1,1}(y)\omega^{(1,j)}_{t,0}(y;x)\right)\,,\\
\mathfrak{C}_{1,3}^{(\bj)}(t;x^{(3)})
&=&\sum\limits_{n=0}^3\int\dy^{(3)}\mathfrak{a}^{(1)}_{1,3,n}(y^{(3)})\prod\limits_{\nu=1}^n\omega_{t,0}^{(1,j_\nu)}(y_\nu;x_\nu)	\prod\limits_{\mu=\nu+1}^3\omega_{t,0}^{(-1,j_\mu)}(y_\mu;x_\mu)\nonumber\\
&& -\i\int\limits_0^t\ds
\int\dy^{(3)}
\Bigg(\Kthps(y^{(3)})\omojo_{t,s}(y_1;x_1)\omojt_{t,s}(y_2;x_2)\omzjth_{t,s}(y_3;x_3) \nonumber\\
&&+\big(\Kthps\big)^*(y^{(3)})\omojo_{t,s}(y_1;x_1)\omzjt_{t,s}(y_2;x_2)\omzjth_{t,s}(y_3;x_3)\bigg)
\end{eqnarray}
\end{subequations}
for $\mathfrak{a}^{(1)}_{n,m,\mu}$ from Assumption \ref{ass:initial:data}.
Here, we define
\begin{subequations}\label{abbrv:omega}
\begin{eqnarray}
\omzz_{t,s}(x;y):= U^*_{t,s}(x;y)\,,&\qquad& \omzo_{t,s}(x;y):= V^*_{t,s}(x;y)\,,\\
\omoz_{t,s}(x;y):= V_{t,s}(y;x)\,,&\qquad& \omoo_{t,s}(x;y):= U_{t,s}(y;x)\,.
\end{eqnarray}
\end{subequations}
Making use of \eqref{def:Chi0:0}, one can equivalently express \eqref{eqn:Chil:explicit:final} as
\begin{equation}
\Chil(t)=\mathcal{U}_{\BogV(t,0)\BogV_0} \sum\limits_{\substack{0\leq n\leq 3\l\\ n+\l\text{ even}}}
\sum\limits_{\bj\in\{-1,1\}^{n}}
\int\dx^{(n+\nu)}
\tilde{\mathfrak{C}}^{(\bj)}_{\l,n}(t;x^{(n+\nu)})\,\,\asjo_{x_1}\,\mycdots\,a^{\sharp_{j_{n}}}_{x_{n}}\,\ad_{x_{n+1}}\,\mycdots\,\ad_{x_{n+\tilde{\nu}}} \,|\Omega\rangle,
\end{equation}
where the coefficients $\tilde{\mathfrak{C}}^{(\bj)}_{\l,n}$ additionally depend on the functions $f_1\mydots f_{\tilde{\nu}}$ from Assumption~\ref{ass:initial:data}. Hence, $\Chil(t)$ is a Bogoliubov transformed sum of states with finitely many particles.

\subsection{Generalized Wick Rule}\label{sec_gen_Wick}

In this section, we study the mixed $n$-point correlation functions 
\begin{equation}
\lrt{a^{\sharp_1}_{x_1}\cdots a^{\sharp_n}_{x_n}}_{\l,k}:=\lr{\Chil(t),a^{\sharp_1}_{x_1}\cdots a^{\sharp_n}_{x_n}\Chik(t)}_{\Fock}.
\end{equation}
For example, in the simplest case where $\Chio(0)$ is quasi-free, Corollary~\ref{thm:even_more_explicit_form} and Wick's rule yield
\begin{eqnarray}
&&\hspace{-0.8cm}\big\langle{a^{\sharp_j}_x}\big\rangle^{(t)}_{0,1}\nonumber\\
&=&\hspace{-5mm}\sum\limits_{m\in\{-1,1\}}\int\dy\,\mathfrak{C}^{(m)}_{1,1}(t;y)\big\langle{a^{\sharp_j}_x a^{\sharp_m}_y}\big\rangle^{(t)}_{0,0}
+\sum\limits_{\bm\in\{-1,1\}^3}\int\dy^{(3)}\mathfrak{C}^{(\bm)}_{1,3}(t;y^{(3)})\nonumber\\
&& \times\bigg(\big\langle{a^{\sharp_j}_x a^{\sharp_{m_1}}_{y_1}}\big\rangle^{(t)}_{0,0}
\big\langle a^{\sharp_{m_2}}_{y_2}a^{\sharp_{m_3}}_{y_3}\big\rangle^{(t)}_{0,0}
+\big\langle{a^{\sharp_j}_x a^{\sharp_{m_2}}_{y_2}}\big\rangle^{(t)}_{0,0}\big\langle{a^{\sharp_{m_1}}_{y_1}a^{\sharp_{m_3}}_{y_3}}\big\rangle^{(t)}_{0,0}
+\big\langle{a^{\sharp_j}_x a^{\sharp_{m_3}}_{y_3}}\big\rangle^{(t)}_{0,0}\big\langle{a^{\sharp_{m_1}}_{y_1}a^{\sharp_{m_2}}_{y_2}}\big\rangle^{(t)}_{0,0}
\bigg)\nonumber\\
&=&\hspace{-5mm}\sum\limits_{\substack{b=2,4\\\bm\in\{-1,1\}^b}}\hspace{-2mm}
\int\dy^{(b)}\mathfrak{C}^{(m_2\mydots m_b)}_{1,b-1}(t;y_2\mydots y_b)\delta(y_1-x)\delta_{j,m_1}
\sum\limits_{\sigma\in P_{b}}\,\prod\limits_{i=1}^{b/2}\big\langle{a_{y_{\sigma(2i-1)}}^{\sharp_{m_{\sigma(2i-1)}}}a_{y_{\sigma(2i)}}^{\sharp_{m_{\sigma(2i)}}}}\big\rangle^{(t)}_{0,0}\quad
\label{eqn:gen:Wick:example}
\end{eqnarray}
and
\begin{equation}
\big\langle{a^{\sharp_{j_1}}_{x_1}a^{\sharp_{j_2}}_{x_2}}\big\rangle^{(t)}_{0,1} =0\,.
\end{equation}
To simplify the notation in \eqref{eqn:gen:Wick:example}, we integrate/sum also over the fixed variable/index $x$/$j$, which is taken into account by the delta distribution/Kronecker delta.
This leads to the following generalization of Wick's theorem, which is proven in Section \ref{subsec:proofs:wick}.

\begin{cor}[Generalized Wick Rule]\label{lem_compute_correlation_functions}
Let Assumption~\ref{ass:initial:data} hold for some $\tilde{a}\in\N_0$ and let $n \in \N$, $k,\l\leq\tilde{a}$ and $\bj \in \{-1,1\}^n$.
\corit{
\item \label{cor:gen:Wick:odd}
If $k+\l+n$ odd, then
\begin{equation}\label{nlk_correlations_odd}
\lrt{a^{\sharp_{j_1}}_{x_1}\cdots a^{\sharp_{j_n}}_{x_{n}}}_{\l,k} = 0 \,.
\end{equation}

\item \label{cor:gen:Wick:even}
Let $\tilde{\nu}=0$ in Assumption \ref{ass:initial:data}, i.e.,  $\Chio(0)$ is quasi-free. Then for $k+\l+n$ even 
\begin{eqnarray}\label{nlk_correlations_even}
&&\hspace{-1cm}\lrt{a^{\sharp_{j_1}}_{x_1}\cdots a^{\sharp_{j_n}}_{x_{n}}}_{\l,k} \nonumber\\
&&\hspace{-0.8cm}=\sum\limits_{\substack{b=n\\\text{even}}}^{n+3(\l+k)} \sum\limits_{\bm\in\{-1,1\}^b}\,
\sum\limits_{\sigma\in P_{b}}\,\prod\limits_{i=1}^{b/2}
\int\dy^{(b)}\mathfrak{D}^{(\bj;\bm)}_{\l,k,n;b}(t;\xn;y^{(b)})\lrt{a_{y_{\sigma(2i-1)}}^{\sharp_{m_{\sigma(2i-1)}}}a_{y_{\sigma(2i)}}^{\sharp_{m_{\sigma(2i)}}}}_{0,0}\qquad
\end{eqnarray}
for $P_{b}$ the set of pairings defined in \eqref{Pairing_set_def} and with
\begin{eqnarray}
\mathfrak{D}^{(\bj;\bm)}_{\l,k,n;b}(t;\xn;y^{(b)})
&:=& \sum\limits_{\substack{q=\max\{0,b-n-3\l\}\\q+k \text{ even}}}^{\min\{3k,b-n\}} 
\overline{\mathfrak{C}^{(-m_{b-n-q}\mydots -m_1)}_{\l,b-n-q}(t; y_{b-n-q}\mydots y_1)} \nonumber\\
&& \times \mathfrak{C}^{(m_{b-q+1}\mydots m_b)}_{k,q}(t;y_{b-q+1}\mydots y_b)\qquad\qquad\nonumber\\
&&\times \prod\limits_{\mu=1}^n\delta(y_{b-n-q+\mu}-x_\mu)\delta_{m_{b-n-q+\mu},\,j_\mu}\,. \label{eqn:D}
\end{eqnarray}
}
\end{cor}

We stated part (b) only for the case $\tilde{\nu}=0$. If $\tilde{\nu}>0$, one obtains a similar but more complicated formula for $\langle a^{\sharp_1}_{x_1}\cdots a^{\sharp_n}_{x_n} \rangle^{(t)}_{\l,k}$  in terms of the two-point correlation functions of $\BogUz|\Omega\rangle$, where $\BogUz$ denotes the Bogoliubov transformation from Assumption \ref{ass:initial:data}. 
Moreover, one may consider initial data that are finite superpositions of states with
different $\tilde{\nu}$, but in this case the mixed $n$-point correlation functions corresponding to $k+\l+n$ odd do not necessarily vanish.
\medskip

If $\Chio(0)$ is quasi-free, $\langle a^{\sharp_1}_{x_1}\cdots a^{\sharp_n}_{x_n} \rangle^{(t)}_{\l,k}$ is given explicitly in terms of the two-point correlation functions of $\Chi_0(t)$,
\begin{equation}
\gamma_{\Chio(t)}(x,y)=\lr{\Chio(t),\ad_y a_x\Chio(t)}_\Fock\,,\qquad
\alpha_{\Chio(t)}(x,y)=\lr{\Chio(t),a_x a_y\Chio(t)}_\Fock\,,
\end{equation}
which can be obtained from the two-point correlation functions of $\Chio(0)$. Evaluating the action of the Bogoliubov transformation $\BogUtz$ on creation and annihilation operators (see Section \ref{subsec:Bogoliubov}), one computes
\begin{subequations}
\begin{eqnarray}
\gChiot(x,y)
&=& \lr{\Chio(0),\BogUtz^*\,\ad_y\,\BogUtz\BogUtz^*\,a_x\,\BogUtz\Chio(0)} \nonumber\\
&=& \Big(\Vbar_{t,0}\gChioz^T\Vbar^*_{t,0}+U_{t,0}\gChioz U^*_{t,0}-\Vbar_{t,0}\aChioz^*U_{t,0}^*
-U_{t,0}\aChioz\Vbar_{t,0}^*\Big)(x,y) \nonumber\\
&& +\left(\Vbar_{t,0}\Vbar_{t,0}^*\right)(x,y)\,,\label{computing_gamma_t} \\
\aChiot(x,y)
&=& \left(U_{t,0}\aChioz\Ubar_{t,0}^* +\Vbar_{t,0}\aChioz^* V_{t,0}^*
-U_{t,0}\gChioz V_{t,0}^* - \Vbar_{t,0}\gChioz^T\Ubar_{t,0}^*\right)(x,y) \nonumber\\
&& +\left(U_{t,0}V_{t,0}^*\right)(x,y) 
\,,\label{computing_alpha_t} 
\end{eqnarray}
\end{subequations}
where $U_{t,0}$ and $V_{t,0}$ denote the matrix entries of the solution $\BogV(t,0)$ of \eqref{eqn:V(t,s):early}.
Alternatively, one can obtain $\gChiot$ and $\aChiot$ by solving the system of differential equations 
\begin{subequations}\label{eqn:PDE}
\begin{eqnarray}
\i\partial_t\gamma_{\Chio(t)}
&=& \left(\hpt+\Kopt\right)\gamma_{\Chio(t)}-\gamma_{\Chio(t)}\left(\hpt+\Kopt\right) \nonumber\\
&& +\Ktpt\alpha^*_{\Chio(t)}-\alpha_{\Chio(t)}\big(\Ktpt\big)^*,
\label{eqn:gamma:PDE} 
\\
\i\partial_t\alpha_{\Chio(t)}
&=&\left(\hpt+\Kopt\right)\alpha_{\Chio(t)}
 +\alpha_{\Chio(t)}\left(\hpt+\Kopt\right)^T \nonumber\\
&& +\Ktpt+\Ktpt\gamma_{\Chio(t)}^T+\gChiot\Ktpt\,,\label{eqn:alpha:PDE}
\end{eqnarray}
\end{subequations}
which is not restricted to quasi-free initial states $\Chio(0)$  (see \cite[Proposition 4(i)]{nam2015} and \cite[Eq.\ (17b-c)]{grillakis2013}).

\subsection{Perturbation Series for Correlation Functions}
\label{sec:perturbation_theory1}

Theorem~\ref{thm:norm:approx} and Corollary \ref{cor:gen:Wick:odd} imply an approximation of the $n$-point correlation functions of the excitations, which are defined as follows:

\begin{definition}
Let $\Chi\in\Fock$ and $n\in\N$. The $n$-point correlation functions of $\Chi$ are defined as
\begin{equation}
\lr{\Chi,a^{\sharp_1}_{x_1}\cdots a^{\sharp_n}_{x_n}\Chi}_{\Fock}\,,
\end{equation}
where $a^{\sharp_j}\in\{\ad,a\}$ for $j\in\{1\mydots n\}$.
For $\ChiN(t)=\UNpt\PsiN(t)$, we use the short-hand notation
\begin{equation}
\lrt{a^{\sharp_1}_{x_1}\cdots a^{\sharp_n}_{x_n}}_{N}:=\lr{\ChiN(t),a^{\sharp_1}_{x_1}\cdots a^{\sharp_n}_{x_n}\ChiN(t)}_{\FN}\,.
\end{equation}
\end{definition}

Formally, the expansion of $\ChiN(t)\oplus0$ from \eqref{eqn:ansatz:ChiN(t)_later} yields
\begin{eqnarray}
\lrt{a^{\sharp_1}_{x_1}\cdots a^{\sharp_n}_{x_n}}_{N}
& = &\lr{\sum_{\ell=0}^{\infty} \lN^{\frac{\l}{2}} \Chil(t) ,a^{\sharp_1}_{x_1}\cdots a^{\sharp_{n}}_{x_{n}} \sum_{m=0}^{\infty} \lN^{\frac{m}{2}} \Chim(t)}_\Fock \nonumber\\
& =& \sum_{\ell=0}^{a} \lN^{\frac{\l}{2}} \sum_{m=0}^{\ell} \lr{\Chim(t) ,a^{\sharp_1}_{x_1}\cdots a^{\sharp_{n}}_{x_{n}} \Chi_{\ell-m}(t)}_\Fock + \Order\Big(\lN^{\frac{a+1}{2}}\Big).
\end{eqnarray}
By the generalized Wick rule (Corollary \ref{lem_compute_correlation_functions}), all contributions $\langle{a^\sharp_{x_1}\mycdots\,a^\sharp_{x_n}}\rangle^{(t)}_{m,\l-m}$ with $\l+n$ odd vanish. This is made rigorous in the following corollary (see Section \ref{subsec:proofs:corr} for a proof).

\begin{cor}\label{thm:correlation functions_simplified}
Let Assumption \ref{ass:initial:data} hold for some $\tilde{a}\in\N_0$ and let $n,p \in \N_0$ with $n+p\leq \tilde{a}+1$, $t\in\R$ and $B\in\mathcal{L}(\fH^{p},\fH^{n})$. Let $a\leq\frac{\tilde{a}-1}{2}$ if $n+p$ even and $a\leq\frac{\tilde{a}}{2}$ if $n+p$ odd. Then there exists some constant $C(a,n,p)>0$ such that
for $n+p$ even,
\begin{subequations}
\begin{eqnarray}
&&\Bigg| \int\dx^{(n)}\dy^{(p)} B(x^{(n)};y^{(p)}) \Bigg( \lrt{\ad_{x_1}\,\mycdots\,\ad_{x_n}\,a_{y_{1}}\,\mycdots\,a_{y_{p}}}_N \nonumber\\
&&\hspace{4mm} - \sum_{\ell=0}^a \lambda_N^{\ell} \sum_{m=0}^{2\ell} \lrt{\ad_{x_1}\,\mycdots\,\ad_{x_n}\,a_{y_{1}}\,\mycdots\,a_{y_{p}}}_{m,2\ell-m} \Bigg) \Bigg| \;\leq\; \norm{B}_{\mathcal{L}(\fH^p,\fH^n)}\, \e^{C(a,n,p)t}\,\lN^{a+1},\qquad\quad\label{thm_eq_even_corr}
\end{eqnarray}
and for $n+p$ odd,
\begin{eqnarray}
&&\Bigg| \int\dx^{(n)}\dy^{(p)} B(x^{(n)};y^{(p)}) \Bigg( \lrt{\ad_{x_1}\,\mycdots\,\ad_{x_n}\,a_{y_{1}}\,\mycdots\,a_{y_{p}}}_N \nonumber\\
&&\hspace{4mm}- \sum_{\ell=0}^{a-1} \lambda_N^{\l+\frac12} \sum_{m=0}^{2\ell+1} \lrt{\ad_{x_1}\,\mycdots\,\ad_{x_n}\,a_{y_{1}}\,\mycdots\,a_{y_{p}}}_{m,2\ell+1-m} \Bigg) \Bigg| \leq \norm{B}_{\mathcal{L}(\fH^p,\fH^n)}\, \e^{C(a,n,p)t}\,\lN^{a+\frac12},\qquad\quad\label{thm_eq_odd_corr}
\end{eqnarray}
\end{subequations}
where for $a=0$ the sum $\sum_{\ell=0}^{a-1}$ is interpreted as zero.
\end{cor}

If $\Chio(0)$ is quasi-free ($\tilde{\nu}=0$ in Assumption~\ref{ass:initial:data}), the approximation can be simplified further by Corollary \ref{cor:gen:Wick:even}. For example, one obtains for $n+p$ even
\begin{eqnarray}
&&\hspace{-0.8cm}\int\dx^{(n+p)}B(x^{(n+p)}) \lrt{\ad_{x_1}\,\mycdots\,\ad_{x_n}\,a_{x_{n+1}}\mycdots\,a_{x_{n+p}}}_N\nonumber\\
&=&  \sum\limits_{\l=0}^a\lN^\l
\sum\limits_{\substack{b=n+p\\b\text{ even}}}^{n+p+6\l}
\sum\limits_{\boldsymbol{\mu}\in\{-1,1\}^b}
\sum\limits_{\sigma\in P_b}\prod\limits_{i=1}^{b/2}
\int\dx^{(n+p+b)}\tilde{\mathfrak{D}}^{(\bj,\boldsymbol{\mu})}_{n,p;\l,b}(t;x^{(n+p+b)})
\lrt{a_{x_{\sigma(2i-1)}}^{\sharp_{\mu_{\sigma(2i-1)}}} a_{x_{\mu_{\sigma(2i)}}}^{\sharp_{\mu_{\sigma(2i)}}}}_{0,0} \nonumber\\
&&+\mathcal{O}(\lN^{a+1})\e^{C(n,p,a)t}
\end{eqnarray}
for $\bj=\{1\}^n\times\{-1\}^p$ and
where the coefficients $\tilde{\mathfrak{D}}^{(\bj,\boldsymbol{\mu})}_{n,p;\l,b}$ are given in terms of  the kernel of $B$ and \eqref{eqn:D}.

\begin{remark}
\remit{
\item\label{remark_thm3_normal_order} We stated Corollary \ref{thm:correlation functions_simplified} only for normal-ordered correlation functions since any correlation function can be normal-ordered using \eqref{eqn:CCR}.

\item Corollary \ref{thm:correlation functions_simplified} implies for $n+p$ even the $L^2$-bound
\begin{eqnarray}
&&\Bigg\| \lrt{\ad_{x_1}\,\mycdots\,\ad_{x_n}\,a_{y_{1}}\,\mycdots\,a_{y_{p}}}_N  \nonumber\\
&&\qquad\qquad- \sum_{\ell=0}^a \lN^\l \sum_{m=0}^{2\ell} \lrt{\ad_{x_1}\,\mycdots\,\ad_{x_n}\,a_{y_{1}}\,\mycdots\,a_{y_{p}}}_{m,2\ell-m} \Bigg\|_{\fH^{n+p}} \leq \e^{C(a,n,p)t}  \lN^{a+1} \qquad\quad
\end{eqnarray}
and analogously for $n+p$ odd.
This follows directly from choosing $B$ as the Hilbert-Schmidt operator with kernel
\begin{equation*}
B(x^{(n)}; y^{(p)}) = \overline{\lrt{\ad_{x_1}\,\mycdots\,\ad_{x_n}\,a_{y_{1}}\,\mycdots\,a_{y_{p}}}_N} - \sum_{\ell=0}^a \lN^{\l} \sum_{m=0}^{2\ell} \overline{\lrt{\ad_{x_1}\,\mycdots\,\ad_{x_n}\,a_{y_{1}}\,\mycdots\,a_{y_{p}}}_{m,2\ell-m}}
\end{equation*}
and from the fact that $\onorm{B}\leq\norm{B}_{\HS}=\norm{B(\cdot\,;\,\cdot)}_{\fH^{n+p}}$.

\item The $2n$-point correlation function $ \langle\ad_{y_1}\mycdots\,\ad_{y_n} a_{x_1}\mycdots\, a_{x_n}\rangle^{(t)}_N$ can be understood as the integral kernel  $\gamma^{(n)}_{\ChiN(t)}(\xn;y^{(n)})$ of the reduced $n$-particle density matrix of $\ChiN(t)$. Since the trace class operators are the dual of the compact operators, Corollary \ref{thm:correlation functions_simplified} implies that
\begin{equation}\label{even_corr_trace_norm}
\Tr\, \Bigg| \gamma^{(n)}_{\ChiN(t)} - \sum_{\ell=0}^a \lN^{\l} \sum_{m=0}^{2\ell} \gamma_{m,2\ell-m}^{(n)}(t) \Bigg| \leq \e^{C(a,n)t} \lN^{a+1},
\end{equation}
where $\gamma_{m,2\ell-m}^{(n)}(t)$ is the operator on $\fH^n$ with kernel
$\langle\ad_{y_1}\mycdots\,\ad_{y_n} a_{x_1}\mycdots\, a_{x_n}\rangle^{(t)}_{m,2\ell-m}$.
}
\end{remark}

Making use of Proposition~\ref{thm:duhamel}, the mixed $n$-point correlation functions $\langle\ad_{x_1}\,\mycdots\,\ad_{x_n}\,a_{y_{1}}\,\mycdots\,a_{y_{p}}\rangle^{(t)}_{\l,k}$ can be computed explicitly and independently of $N$, given their initial values and  the solutions $\BogV(t,s)$ and $\varphi(t)$ of the two-body problem \eqref{eqn:V(t,s):early} and the one-body problem \eqref{hpt}, respectively. Since the actual correlation functions of the $N$-body problem  are determined by $\langle\ad_{x_1}\,\mycdots\,\ad_{x_n}\,a_{y_{1}}\,\mycdots\,a_{y_{p}}\rangle^{(t)}_{\l,k}$  to any order in $N^{-1/2}$, this implies a drastic reduction of complexity. 

Alternatively, the functions $\langle\ad_{x_1}\,\mycdots\,\ad_{x_n}\,a_{y_{1}}\,\mycdots\,a_{y_{p}}\rangle^{(t)}_{\l,k}$ can be obtained from solving a system of PDEs. A straightforward computation yields
\begin{subequations}\label{PDE_line_all}
\begin{eqnarray}
\label{PDE_line1} \i \partial_t \sum_{m=0}^{\ell} \lrt{a^{\sharp_1}_{x_1}\,\mycdots\,a^{\sharp_n}_{x_n}}_{m,\ell-m} 
&=& \sum_{m=0}^{\ell} \lrt{\Big[a^{\sharp_1}_{x_1}\,\mycdots\,a^{\sharp_n}_{x_n},\FockHopt \Big]}_{m,\ell-m} \\
\label{PDE_line2}&& + \sum_{k=0}^{\l-1} \sum_{m=0}^{k} \lrt{\Big[a^{\sharp_1}_{x_1}\,\mycdots\,a^{\sharp_n}_{x_n},\FockHpt^{(\l-k)} \Big]}_{m,k-m}.\qquad\qquad
\end{eqnarray}
\end{subequations}
The commutator on the right-hand side of \eqref{PDE_line1} contains again $n$ creation and annihilation operators and is evaluated in the same $m$, $\l-m$ inner product as the left-hand side. The commutators in \eqref{PDE_line2} contain at most $n+\l-k$ creation and annihilation operators and are determined by a PDE analogous to \eqref{PDE_line_all}. Since they are evaluated only in inner products with $\Chio(t), \ldots, \Chi_{\l-1}(t)$ and not $\Chi_{\l}(t)$, iterating this procedure yields a closed system of coupled PDEs. For the simplest example with $\l=1$ and $n=1$, see Section~\ref{subsec:results:RDM}.\medskip

As a consequence of Corollary \ref{thm:correlation functions_simplified}, the expectation value of any bounded $k$-body operator $A^{(k)}$ with respect to $\PsiN(t)$ can be computed to arbitrary precision. To see this, recall that
\begin{equation}
\lr{\PsiN(t),A^{(k)}\PsiN(t)}_{\fH^N}=\tbinom{N}{k}^{-1}\lr{\ChiN(t),\UNpt \d\Gamma(A^{(k)})\UNpt^*\ChiN(t)}_{\FN}\,,
\end{equation}
and the latter is as a sum of expressions covered by Corollary \ref{thm:correlation functions_simplified}.
Equivalently, this yields an expansion of the reduced $k$-body density matrices $\gamma^{(k)}_{\Psi_N}(t)$ of $\Psi^N(t)$, which is discussed in the next section for the simplest case $k=1$.

\subsection{One-Particle Reduced Density Matrix}
\label{subsec:results:RDM}

Corollary \ref{thm:correlation functions_simplified} implies an asymptotic expansion of the $k$-particle reduced density matrices $\gPsiNk(t)$ in terms of $\pt$ and the mixed $n$-point correlation functions $\langle a^{\sharp}\mycdots a^{\sharp}\rangle^{(t)}_{\l,k}$. We restrict to the case $k=1$. 
First, one observes that
\begin{eqnarray}\label{eqn:RDM:aux}
\gPsiNo(t)
&=&\ppt\gPsiNo(t)\ppt + \qpt\gPsiNo(t)\qpt + \left(\ppt\gPsiNo(t)\qpt +\hc\right)\nonumber\\
&=&\frac{1}{N}\ppt\lr{\PsiN(t),\ad(\pt)a(\pt)\PsiN(t)}_{\fH^N} + \frac{1}{N}\gChiNt \nonumber\\
&&+ \bigg(\frac{1}{N}\sum\limits_{\l\geq1}|\pt\rangle\langle\varphi_\l(t)|\lr{\PsiN(t),\ad(\varphi_\l(t))a(\pt)\PsiN(t)}_{\fH^N}+\hc\bigg)
\end{eqnarray}
for any orthonormal basis $\{\varphi_\l(t)\}_{\l\geq 0}$ of $\fH$ with $\varphi_0(t)=\pt$, and
where $\gChiNt$ denotes the one-body reduced density matrix of $\ChiN(t)$ with kernel $\gChiNt(x;y)=\langle\ChiN(t),\ad_ya_x\ChiN(t)\rangle$. 
Hence, the substitution rules \eqref{eqn:substitution:rules} yield
\begin{eqnarray}
\gPsiNo(t) 
&=& \ppt + \frac{1}{\sqrt{N}} \left(|\varphi(t) \rangle \langle \beta_{\ChiN(t)}| + |\beta_{\ChiN(t)} \rangle \langle \varphi(t)|\right) \nonumber\\
&&+\frac{1}{N} \left(\gChiNt- \ppt \lr{\ChiN(t),\Number \ChiN(t)}_{\FN}\right),\label{gamma_psi_gamma_chi}
\end{eqnarray}
where
\begin{equation}\label{eqn:bChiNt}
\bChiNt(x) := \lr{\ChiN(t), \sqrt{1-\frac{\Number}{N}} \,a_x \ChiN(t)}_{\FN}.
\end{equation}
Expanding the $N$-dependent expressions in \eqref{gamma_psi_gamma_chi} in powers of $\lN^{1/2}$ and applying Corollary~\ref{thm:correlation functions_simplified} leads to the following result, whose proof is postponed to Section \ref{subsec:proofs:RDM}.

\begin{theorem}\label{thm:RDM}
Let Assumption~\ref{ass:initial:data} hold for some $\tilde{a}\in\N$ and let $a\leq\frac{\tilde{a}-1}{2}$ and $t\in\R$. Then
\begin{equation}\label{eq_red_density_result}
\Tr \Big| \gPsiNo(t) - \sum\limits_{\l=0}^a\lN^\l\gamma_\l^{(1)}(t) \Big| \leq \e^{C(a)t} \lN^{a+1}
\end{equation}
for some constant $C(a)>0$, where
\begin{subequations}\label{eqn:gamma:higher:orders}
\begin{eqnarray}
\gamma_0^{(1)}(t;x;y)&:=& \varphi(t,x)\overline{\varphi(t,y)}\,, \\
\gamma_\l^{(1)}(t;x;y)&:=& \sum\limits_{m=1}^{\l}\Bigg[
\sum\limits_{k=0}^{\l-m}\sum\limits_{n=0}^{2m-1}\tilde{c}_{\l-m,k}\,\bigg( \varphi(t,x)\lrt{ \ad_y(\Number-1)^k}_{n,2m-n-1}\nonumber\\
&&\hspace{4cm}+\lrt{(\Number-1)^ka_x}_{n,2m-n-1}\overline{\varphi(t,y)}
\bigg) \nonumber\\
&&+\sum\limits_{n=0}^{2m-2}\tilde{c}_{\l-m}\left(\lrt{\ad_y a_x}_{n,2m-n-2}
-\varphi(t,x)\overline{\varphi(t,y)}\lrt{\Number}_{n,2m-n-2}\right)\Bigg]\label{gamma_correction_a}\qquad
\end{eqnarray}
\end{subequations}
for $\l\geq1$, where $\tilde{c}_\l=(-1)^\l c_\l^{(3/2)}$ and $\tilde{c}_{\l,k}:=\tilde{c}_{\l-k}c_k^{(0)}$ with $c_\l^{(n)}$ as in \eqref{eqn:taylor:coeff}.
\end{theorem}

Theorem \ref{thm:RDM} implies for any bounded operator $A:\fH\to\fH$ that
\begin{equation}
\left|\Tr_{\fH}A\gPsiNo(t) - \sum\limits_{\l=0}^a\lN^\l\Tr_{\fH}A\gamma^{(1)}_\l(t) \right| \leq \onorm{A}\,\e^{C(a)t} \lN^{a+1}\,.
\end{equation} 
For $a=0$, we recover the well-known statement that $\gPsiNo(t)\approx\ppt$ up to an error of order $N^{-1}$ \cite{erdos2009,chenlee:2011,chen2011,mitrouskas2016}. 
In \cite{rodnianski2009,pickl2011}, an error estimate of order $N^{-1/2}$ was proven by estimating $\beta_{\ChiN(t)}$ uniformly in $N$ without making use of the Bogoliubov approximation. In this case, the Bogoliubov approximation improves the
convergence rate but does not give itself a correction to the one-body reduced density matrix.

The next order correction to the leading order $\gamma_0^{(1)}(t)=\ppt$ is
\begin{eqnarray}
\gamma_1^{(1)}(t;x;y)
&=& \varphi(t,x) \Big(\langle{\ad_y}\rangle^{(t)}_{0,1}+\langle{\ad_y}\rangle^{(t)}_{1,0}\Big)
+\Big(\langle{a_x}\rangle^{(t)}_{1,0}+\langle{a_x}\rangle^{(t)}_{0,1}\Big)\overline{\varphi(t,y)}
+\gChiot(x;y) \nonumber\\
&&- \big(\Tr_\fH\gChiot\big) \,\varphi(t,x)\overline{\varphi(t,y)}\,.
\end{eqnarray}
The two-point correlation function $\gChiot$  is given by the solution  of \eqref{eqn:PDE}, or, if $\Chio(0)$ is quasi-free,  directly by \eqref{computing_gamma_t}. The mixed correlation functions $\langle a^\sharp \rangle^{(t)}_{1,0}$ and $\langle a^\sharp \rangle^{(t)}_{0,1}$ can be obtained from this by Corollary \ref{lem_compute_correlation_functions}.
Alternatively, 
\begin{equation}
\bzo(t,x):=\lrt{a_x}_{0,1}+\lrt{a_x}_{1,0}
\end{equation} 
is determined by a partial differential equation. As a consequence of \eqref{eqn:differential:form:Chil},
\begin{eqnarray}
\i\partial_t\bzo(t,x) 
&=& \lrt{\left[a_x,\FockHopt\right]}_{1,0}+\lrt{\left[a_x,\FockHopt\right]}_{0,1} + \lrt{\left[a_x,\FockHtpt\right]}_{0,0}\nonumber\\
&= &\hpt_x\bzo(t,x) +  \int \dy \Kopt(x;y)\bzo(t,y) + \int \dy\, \Ktpt(x;y)\overline{\bzo(y)} \nonumber\\
&& + \int \dy^{(2)}\bigg( \Kthpt(x,y_1;y_2)+\Kthpt(y_1,x;y_2)\bigg) \gChiot(y_2;y_1) \nonumber\\
&& + \int\dy^{(2)}\big(\Kthpt\big)^*(x;y_1,y_2)\aChiot(y_1,y_2)\,,
\end{eqnarray}
which is \eqref{intro_red_dens_corr} from the introduction.

\section{Bogoliubov Transformations and Quasi-free States}
\label{subsec:Bogoliubov}
In this section, we briefly recall the concepts of Bogoliubov transformations, Bogoliubov maps, and quasi-free states, and prove some of their properties. Our main references are \cite{solovej_lec,lewin2015,nam2015}.
Let us consider
\begin{equation}
F=f\oplus Jg=f\oplus\gbar=\begin{pmatrix}f \\ \overline{g} \end{pmatrix} \in\fH\oplus\fH\,,
\end{equation}
where  $J:\fH\to\fH$,  $(Jf)(x)=\overline{f(x)}$, denotes the complex conjugation map.
The generalized creation and annihilation operators $A(F)$ and $\Ad(F)$ are defined as
\begin{equation}\label{eqn:A}
A(F)=a(f)+\ad(g)\,, \quad \Ad(F)=A(\cJ F)=\ad(f)+a(g)
\end{equation}
for $\cJ=\left(\begin{smallmatrix}0 & J\\J&0\end{smallmatrix}\right)$.
If an operator $\BogV$ on $\fH\oplus\fH$ is such that the map $F\mapsto A(\BogV F)$ has the same properties  as the map $F\mapsto A(F)$, i.e., if
\begin{equation}
\Ad(\BogV F)=A(\BogV\mathcal{J}F)\,,\qquad [A(\BogV F_1),\Ad(\BogV F_2)]= [A( F_1),\Ad( F_2)]\,,
\end{equation}
the operator $\BogV$ is called a \emph{(bosonic) Bogoliubov map}.
This requirement is equivalent to the following definition:

\begin{definition}
A bounded operator 
\begin{equation}
\BogV:\fH\oplus \fH\to \fH\oplus \fH
\end{equation}
is called a Bogoliubov map if it satisfies
\begin{equation}
\BogV^*\mathcal{S}\BogV=\mathcal{S}=\BogV\mathcal{S}\BogV^*\,,\qquad \mathcal{J}\BogV\mathcal{J}=\BogV\,,
\end{equation}
where $\cS:=\left(\begin{smallmatrix}1&0\\0&-1\end{smallmatrix}\right)$.
Equivalently, $\BogV$ is of the block form
\begin{equation}\label{BogV:block:form}
\BogV:=\begin{pmatrix}U & \Vbar\\V & \Ubar\end{pmatrix}\,,\quad U,V:\fH\to \fH\,,
\end{equation}
where $U$ and $V$ satisfy the relations
\begin{equation}\label{eqn:rel:U:V}
U^*U=\id+V^*V\,, \qquad UU^*=\id+\Vbar\,\Vbar^*\,, \quad V^*\Ubar=U^*\Vbar\,,\quad UV^*=\Vbar\,\Ubar^*\,.
\end{equation}
We denote the set of Bogoliubov maps on $\fH\oplus\fH$ as
\begin{equation}
\fV(\fH):=\left\{\BogV\in\mathcal{L}\left(\fH\oplus\fH\right)\,|\,\BogV \text{ is a Bogoliubov map }\right\}.
\end{equation}
\end{definition}

Bogoliubov maps can be implemented on Fock space in the following sense:

\begin{lem}
Let $\BogV\in\fV(\fH)$. Then there exists a unitary transformation $\BogU:\FH\to\FH$ such that
\begin{equation}\label{eqn:unit:impl}
\BogU A(F)\BogU^*=A(\BogV F)
\end{equation}
for all  $F\in\fH\oplus\fH$
if and only if 
\begin{equation}
 \norm{V}_\mathrm{HS(\fH)}^2 :=\Tr(V^*V)<\infty
\end{equation}
(Shale--Stinespring condition).
In this case,   $\BogV$ is called (unitarily) implementable.
\end{lem}
A proof of the lemma is given, for instance, in \cite[Theorem 9.5]{solovej_lec}. 
In the following, we refer to the unitary implementation $\BogU:\FH\to\FH$ of a Bogoliubov map $\BogV\in\fV(\fH)$ as  \emph{Bogoliubov transformation}.
Some relevant properties of Bogoliubov transformations are summarized in the following lemma, which is easily verified by direct computations using \eqref{eqn:rel:U:V}.

\begin{lem}\label{lem:properties:BT}
The Bogoliubov maps $\fV(\fH)$ form a subgroup of the group of isomorphisms on $\fH\oplus\fH$.
In particular, the adjoint and inverse of  $\BogV\in\fV(\fH)$  with block form \eqref{BogV:block:form} are given as
\begin{equation}\label{eqn:BogV:-1}
\BogV^*=\begin{pmatrix}
U^* & V^*\\\Vbar^{\,*} & \Ubar^{\,*} \end{pmatrix}\,,\quad
\BogV^{-1}=\cS\BogV^*\cS=\begin{pmatrix}
U^* & -V^*\\ -\Vbar^{\,*} & \Ubar^{\,*} \end{pmatrix}\,.
\end{equation}
If $V$ is a Hilbert--Schmidt operator, the set of all Bogoliubov transformations,
\begin{equation}
\left\{ \BogU:\FH\to\FH\,|\, \BogV\in\fV(\fH)\right\}\,,
\end{equation}
forms a subgroup of the group of unitary maps on $\FH$.
Moreover, the map
$\BogV\mapsto \BogU$
is a group homomorphism, which, in particular, implies that
\begin{equation}
\mathcal{U}_{\BogV^{-1}}=\left(\BogU\right)^{-1}=\BogU^*\,.
\end{equation}
\end{lem}

We write the operators $U$ and $V$  as integral operators with (Schwartz) kernels $U(x;y)$ and $V(x;y)$, i.e., for any $f\in \fH$,
\begin{equation}
(Uf)(x):=\int U(x;y)f(y)\dy\,, \qquad (Vf)(x):=\int V(x;y)f(y)\dy \,,
\end{equation} 
and the operators $\Vbar$ and $\Ubar$ are to be understood as integral operators with kernels $\overline{V(x;y)}$ and $\overline{U(x;y)}$.
The transformation rule \eqref{eqn:unit:impl} corresponds to
\begin{subequations}\label{eqn:trafo:ax}
\begin{eqnarray}
\BogU \,a_x\,\BogU^*&=&\int \dy\, \overline{U(y;x)}\,a_y
+\int \dy\, \overline{V(y;x)}\,\ad_y\,,\\
\BogU\,\ad_x\,\BogU^*&=&\int \dy\, V(y;x)\,a_y + \int\dy\, U(y;x)\,\ad_y.
\end{eqnarray}
\end{subequations}
The inverse relation to \eqref{eqn:unit:impl} is given by
\begin{equation}\label{U_star_A_U}
\BogU^*A(F)\BogU \;=\; A(\BogV^{-1}F) \,,
\end{equation}
which, by \eqref{eqn:BogV:-1}, leads to the relations
\begin{subequations}\label{eqn:trafo:ax:invers}
\begin{eqnarray}
\BogU^*\,a_x\,\BogU&=&\int\dy\, U(x;y)a_y 
- \int \dy\,\overline{V(x;y)}\ad_y\,,\\
\BogU^*\,\ad_x\,\BogU&=&-\int\dy V(x;y) a_y + 
\int \dy\,\overline{U(x;y)}\ad_y\,.
\end{eqnarray}
\end{subequations}
The Bogoliubov transformed number operator can be bounded as follows:

\begin{lem}\label{lem:number:BT} 
Let $b\in\N$, let $\BogV\in\mathfrak{V}(\fH)$ be unitarily implementable, and denote by $\BogU$ the corresponding Bogoliubov transformation on $\Fock$. Then it holds in the sense of operators on $\Fock$  that
\begin{equation}
\BogU(\Number+1)^b\BogU^* \leq C_\BogV^b\, b^b(\Number+1)^b
\end{equation}
with
\begin{equation}\label{CV}
C_\BogV:=2\norm{V}_\HS^2+\onorm{U}^2+1
\end{equation}
for $\BogV=\begin{pmatrix}U & \Vbar\\ V&\Ubar\end{pmatrix}$.
\end{lem}

\begin{proof}
Let $\l\in\R$.
As a consequence of \eqref{eqn:trafo:ax}, it follows that 
\begin{eqnarray}
&&\hspace{-1cm}(\Number+1)^\l \BogU(\Number+1)\BogU^* \nonumber\\
&=&\int\dy\dz(\Ubar\,\Vbar^*)(y;z)a_ya_z(\Number-1)^\l +\int\dy\dz  (UV^*)(y;z)\ad_y\ad_z(\Number+3)^\l\nonumber\\
&&+\bigg(\int\dy\dz(\Vbar\,\Vbar^*+UU^*)(y;z)\ad_y a_z+\norm{V}^2_{\HS(\fH)}+1\bigg)(\Number+1)^\l \,.
\end{eqnarray}
By Lemma \ref{lem:aux:new} and \eqref{eqn:rel:U:V}, this yields for $\l \in\R_0^+$, $b\in\N$ and $\bPhi\in\Fock$
\begin{eqnarray}
&&\hspace{-1cm}\norm{(\Number+1)^\l \BogU(\Number+1)^b\BogU^*\bPhi}_\Fock\nonumber\\
&=&\norm{(\Number+1)^\l \BogU(\Number+1)\BogU^*\BogU(\Number+1)^{b-1}\BogU^*\bPhi}_\Fock\nonumber\\[5pt]
&\leq&C_\BogV\Big\|(\Number+2)(\Number+3)^\l \BogU(\Number+1)\BogU^*\BogU(\Number+1)^{b-2}\BogU^*\bPhi\Big\|_\Fock \nonumber\\
&\leq&C_\BogV^b\norm{(\Number+2)(\Number+4)\mycdots(\Number+2b)(\Number+2b+1)^\l \bPhi}_\Fock \,.\label{eqn:number:BT:1}
\end{eqnarray}
The choice $\l =0$ proves the lemma for $b$ even.
For $b$ odd, note first that
\begin{eqnarray}
&&\hspace{-1.5cm}\norm{(\Number+1)^\frac12\BogU^*\bPhi}^2_\Fock\nonumber\\
&=&\left(\lr{\bPhi,\int\dy\dz(\Ubar\,\Vbar^*)(y;z)a_ya_z\bPhi}_\Fock+\hc\right)\nonumber\\
&& +\int\dz\lr{\int\dy (V\,V^*+\Ubar\,\Ubar^*)(y;z) a_y\bPhi, a_z\bPhi}_\Fock+\bigg(\norm{V}^2_{\HS}+1\bigg)\norm{\bPhi}^2_\Fock \nonumber\\
&\leq&C_\BogV\norm{(\Number+1)^\frac12\bPhi}^2_\Fock\,,\label{eqn:number:BT:2}
\end{eqnarray}
where we used  that
\begin{eqnarray}
\lr{\bPhi,\int\dy\dz f(y;z) a_ya_z\bPhi}_\Fock
&\leq& \norm{f}_{\fH^2}\sum\limits_{k\geq0}\norm{\sqrt{k+1}\phi^{(k)}}_{\fH^k}\norm{\sqrt{k+2}\,\phi^{(k+2)}}_{\fH^{k+2}}\nonumber\\
&\leq&\norm{f}_{\fH^2}\norm{(\Number+1)^\frac12\bPhi}^2_\Fock\,,
\end{eqnarray}
and  that $\norm{\int\dy (VV^*)(y;z)a_y\bPhi}_\Fock\leq\onorm{VV^*}\norm{a_z\bPhi}_\Fock$, etc., and  $\int\dz\norm{a_z\bPhi}_\Fock^2=\norm{\Number^\frac12\bPhi}_\Fock^2$.
Thus, we find
\begin{eqnarray}
\norm{(\Number+1)^\frac{2b+1}{2}\BogU^*\bPhi}_\Fock
&=&\norm{(\Number+1)^\frac12\BogU^*\BogU(\Number+1)^b \BogU^* \bPhi}_\Fock\nonumber\\ 
&\leq& C_\BogV^{b+\frac12}(2b+1)^{b+\frac12}\norm{(\Number+1)^{b+\frac12}\bPhi}_\Fock
\end{eqnarray}
by \eqref{eqn:number:BT:2} and \eqref{eqn:number:BT:1} with $\l=\frac12$. 
\end{proof}

A concept that is closely connected to Bogoliubov transformations is the notion of quasi-free states.
\begin{definition}\label{def:QF}
A state $\bPhi\in\FH$ with $\norm{\bPhi}_\FH=1$ is called a quasi-free (pure) state if it can be expressed as a Bogoliubov transformation of the vacuum, i.e., if there exists a $\BogV\in\fV(\fH)$ such that 
\begin{equation}
\bPhi=\BogU|\Omega\rangle\,.
\end{equation}
The set of all quasi-free states on $\FH$ is denoted by
\begin{equation}
\mathcal{Q}(\fH):=\left\{\bPhi\in\FH\,:\, \exists\, \BogV\in\fV(\fH) \text{ s.t. } \bPhi=\BogU|\Omega\rangle\right\} \subset \FH\,.
\end{equation}
\end{definition}
Our definition of quasi-free states is sometimes referred to as \emph{even quasi-free} \cite{derezinski}.
We define the set of quasi-free excitation vectors at time $t\in\R$ as
\begin{equation}
\Qpt:=\cQ(\fH)\cap\Fpt=\left\{\bPhi\in\Fpt \;:\; \exists\,\BogV\in\fV(\fH) \text{ \emph{s.t. }} \bPhi=\BogU|\Omega\rangle\right\}\,.
\end{equation}
Clearly, any Bogoliubov transformation leaves the set of quasi-free states invariant, i.e.,
\begin{equation}
\BogU\cQ(\fH)=\cQ(\fH)
\end{equation}
for all $\BogV\in\fV(\fH)$.

\begin{definition}\label{def:Gamma}
Define the generalized one-particle density matrix 
$\Gamma_\bPhi: \fH\oplus\fH\to \fH\oplus\fH$
of $\bPhi\in\FH$ as
\begin{equation}
\Gamma_\bPhi:=\begin{pmatrix}
    \gamma_\bPhi & \alpha_\bPhi \\
    \alpha_\bPhi^* & \mathbbm{1}+\gamma_\bPhi^T
  \end{pmatrix}
  \,,
\end{equation}
where the one-body density matrix $\gamma_\bPhi:\fH\to \fH$ and the pairing density $\alpha_\bPhi:\fH\to\fH$ are defined as the operators with kernels
\begin{equation}
\gamma_\bPhi(x,y) = \lr{\bPhi,\ad_y\,a_x\,\bPhi}_\Fock\,,\qquad
\alpha_\bPhi(x,y) = \lr{\bPhi, a_x\,a_y\,\bPhi}_\Fock\,.
\end{equation}
\end{definition}
Quasi-free states are uniquely determined by their generalized one-particle density matrix, or alternatively by Wick's rule.

\begin{lem}
Let $\bPhi\in\FH$ with $\norm{\bPhi}_\Fock=1$. Then the following are equivalent:
\lemit{
\item \begin{equation}\bPhi\in\cQ(\fH)\,.\end{equation}

\item \label{lem:wick}
\begin{equation}\lr{\bPhi,\Number\bPhi}_\FH<\infty
\end{equation}
and, for all $a^\sharp\in\{\ad,a\}$, $n\geq 1$ and $f_1\mydots f_{2n}\in \fH$,
\begin{subequations}\label{eqn:Wick}
\begin{eqnarray}
\lr{\bPhi,a^\sharp(f_1)\mycdots a^\sharp(f_{2n-1})\bPhi}_{\FH} &=& 0\,,\\
\lr{\bPhi,a^\sharp(f_1)\mycdots a^\sharp(f_{2n})\bPhi}_{\FH} &=& \sum\limits_{\sigma\in P_{2n}}\prod\limits_{j=1}^n \lr{\bPhi,a^\sharp(f_{\sigma(2j-1)})a^\sharp(f_{\sigma(2j)})\bPhi}_{\FH}\,,\;
\end{eqnarray}
\end{subequations}
where $P_{2n}$ denotes the set of pairings as in \eqref{Pairing_set_def}.
The property \eqref{eqn:Wick} is known as Wick's rule.
\item 
\begin{equation}
\gamma_\bPhi\alpha_\bPhi=\alpha_\bPhi\gamma_\bPhi^T\,,\qquad \alpha_\bPhi\alpha_\bPhi^*=\gamma_\bPhi(\id+\gamma_\bPhi)\,.
\end{equation}
\item 
\begin{equation}
\Gamma_\bPhi \mathcal{S}\Gamma_\bPhi=-\Gamma_\bPhi\,.
\end{equation}
}
\end{lem}

The equivalence of $(a\Leftrightarrow b)$ is proven in \cite[Theorem 1.6]{nam2011}, $(a\Leftrightarrow d)$ in \cite[Theorem 1.6]{nam2011} or \cite[Theorem 10.4]{solovej_lec}, and $(c\Leftrightarrow d)$ is a simple computation.\\

Let us now make the connection to the time-dependent setting. The crucial observation is that the time evolution operator $\Uo(t,s)$ of the Bogoliubov equation \eqref{eqn:Bogoliubov:equation} is a time-dependent Bogoliubov transformation.

\begin{lem} \label{lem:time:dep:BT}
Let $s,t\in\R$, $\bPhi(s)\in\FH$, $\varphi(s)\in H^1({\R^d})$ with $\norm{\varphi(s)}_\fH=1$, and denote by $\pt$ the unique solution of \eqref{hpt} with initial datum $\varphi(s)$.
\lemit{
\item \label{lem:time:dep:BT:0}
The Bogoliubov equation \eqref{eqn:Bogoliubov:equation} with initial datum $\bPhi(0)$ in the quadratic form domain $Q(\d\Gamma(1-\Delta))$ has a unique solution $\bPhi\in C^0([0,\infty),\FH) \cap L_{\mathrm{loc}}^\infty([0,\infty),Q(\d\Gamma(1-\Delta)))$. We denote by 
$\Uo(t,s)$
the corresponding unitary time evolution on $\FH$, i.e.,
\begin{equation}\bPhi(t)=\Uo(t,s)\bPhi(s)\,.\end{equation}
\item \label{lem:time:dep:BT:1}
If $\bPhi(s)$ is quasi-free, then $\bPhi(t)$ is quasi-free, i.e.,
\begin{equation}\Uo(t,s)\cQ(\fH)\subseteq\cQ(\fH)\,.\end{equation}
\item \label{lem:time:dep:BT:2}
Let $\bPhi(s)\in\Fps$ such that $\lr{\bPhi(s),\Number\bPhi(s)}<\infty$. Then $\bPhi(t)\in\Fpt$, which, in particular, implies that
\begin{equation}\Uo(t,s)\Qps\subseteq\Qpt\,.\end{equation}
\item \label{lem:time:dep:BT:3}
$\Uo(t,s)$ is a Bogoliubov transformation on $\FH$, i.e., there exists a two-parameter group of Bogoliubov maps $\BogV(t,s)\in\fV(\fH)$ such that 
\begin{equation}\Uo(t,s)=\BogUts\,.\end{equation}
Moreover, $\BogV(t,s)$ is the weak\footnote{\label{footnote_well_posedness}With ``weak'' solution we mean here that $\BogV(t,s)$ is strongly continuous and $i\partial_t \scp{F}{\BogV(t,s) G} = \scp{F}{\mathcal{A}(t)\BogV(t,s)G}$ for all $G \in L^2(\R^d) \oplus L^2(\R^d)$ and $F \in H^2(\R^d) \oplus H^2(\R^d)$.} solution of the differential equation
\begin{equation}\label{eqn:V(t,s)}
\begin{cases}
\displaystyle\i\partial_t\BogV(t,s)&=\mathcal{A}(t)\BogV(t,s)\,\\[10pt]
\quad\displaystyle\BogV(s,s)&=\id\,,
\end{cases}
\end{equation}
where
\begin{equation}\label{eqn:A(t)}
\mathcal{A}(t)=\begin{pmatrix} 
\hpt+\Kopt& -\Ktpt \\ 
\overline{\Ktpt}& -\left(\hpt+\overline{\Kopt}\right) \end{pmatrix} \,.
\end{equation}
}
\end{lem}

\begin{proof}
Parts (a) and (c) are shown in \cite[Theorem 7]{lewin2015}  (see also \cite{ginibre1979,ginibre1979_2,nam2015,bach2016}),
and assertion (b) is proven in \cite[Proposition 4, Step 5]{nam2015}. Part (d) was shown, in a slightly different setting, in \cite[Theorem 2.2]{benarous2013}, and a similar result was obtained in \cite[Lemma 3.10]{bach2016}. We give a simple proof here. Without loss of generality, let us assume that $s=0$. We prove Lemma \ref{lem:time:dep:BT:3} in three steps.

\begin{claim}\label{app:B:claim:1}
Let $\bPhi_0\in\cQ(\fH)$. Then there exists a unitarily implementable  $\BogV_t\in\mathfrak{V}(\fH)$ such that $\bPhi(t):=\Uo(t,0)\bPhi_0=\BogUt\bPhi_0$.
\end{claim}
\begin{proof}\renewcommand{\qed}{}
Since $\bPhi_0$ is quasi-free, there exist Bogoliubov transformations $\cU_{\tilde{\BogV}_0}$ and $\cU_{\tilde{\BogV}_t}$ such that $\bPhi_0=\cU_{\tilde{\BogV}_0}|\Omega\rangle$ and $\bPhi(t)=\cU_{\tilde{\BogV}_t}|\Omega\rangle$ by Lemma \ref{lem:time:dep:BT:1}. This proves the claim with $\BogUt:=\cU_{\tilde{\BogV}_t\tilde{\BogV}_0^{-1}}$.
\end{proof}

\begin{claim} \label{app:B:claim:2}
The Bogoliubov map $\BogV_t$ is the unique weak solution of \eqref{eqn:V(t,s)} in $\mathfrak{V}(\fH)$.
\end{claim}
\begin{proof}\renewcommand{\qed}{}
By \cite[Proposition 4, Step 3]{nam2015}, the pair $(\gamma_{\bPhi(t)},\alpha_{\bPhi(t)})$ of density matrices is the unique weak solution of the set of equations \eqref{eqn:PDE} with initial condition $(\gamma_{\bPhi_0},\alpha_{\bPhi_0})$. Equivalently, by Definition \ref{def:Gamma},  $\Gamma_{\bPhi(t)}$ is the unique weak solution of the equation
\begin{equation}
\i\partial_t\Gamma(t) = \mathcal{A}(t)^*\Gamma(t)-\Gamma(t)\mathcal{A}(t)\,,\qquad
 \Gamma(0)=\Gamma_{\bPhi_0}\,,
\end{equation}
for $\mathcal{A}(t)$ as in \eqref{eqn:A(t)}. As $\lr{F_1,\Gamma_{\bPhi(t)} F_2}_{\fH\oplus\fH}=\lr{\bPhi(t),\Ad(F_2)A(F_1)\bPhi(t)}_\FH$ for $F_1,F_2\in\fH\oplus\fH$, one easily verifies that
\begin{equation}\label{eqn:app:B:Gamma:V}
\BogV_t^*\Gamma_{\bPhi(t)}\BogV_t=\Gamma_{\bPhi_0}\,,
\end{equation}
which implies that $\i\partial_t\left(\BogV_t^*\Gamma_{\bPhi(t)}\BogV_t\right)=0$ and consequently
\begin{eqnarray}
\BogV_t^*\Gamma_{\bPhi(t)}\left(\i\partial_t\BogV_t-\mathcal{A}(t)\BogV_t\right) -\hc=0\,.
\end{eqnarray}
This proves the claim  because there is a one-to-one correspondence between $\BogV_t$ and $\Gamma_{\bPhi(t)}$. Note that the weak well-posedness as in Footnote~\ref{footnote_well_posedness} follows from \cite[Theorem~X.69]{rs2} by a Dyson expansion in the interaction picture.
\end{proof}

\begin{claim}
In the sense of operators on $\FH$, it holds that $\BogUt = \Uo(t,0)$.
\end{claim}
\begin{proof}\renewcommand{\qed}{}
Let $\bPhi_0\in\cQ(\fH)$. We prove by induction over $n\in\N$ that
\begin{equation}\label{eqn:app:B:1}
\i\partial_t\BogUt A(F_1)\mycdots A(F_n) \bPhi_0 \;=\;\FockHopt\BogUt A(F_1)\mycdots A(F_n)\bPhi_0
\end{equation}
for any $n\in\N$ and $F_1\mydots F_n \in \fH\oplus\fH$. By  Lemma \ref{lem:time:dep:BT:0}, this implies that
\begin{equation}
\Uo(t,0) A(F_1) \cdots A(F_n) \bPhi_0 = \BogUt A(F_1) \cdots A(F_n) \bPhi_0\,,
\end{equation} 
and then the claim follows by density.
By Claim \ref{app:B:claim:1}, $\i\partial_t\BogUt\bPhi_0=\FockHopt\BogUt\bPhi_0$, and one easily verifies that
\begin{equation}\label{Bog_commutator_relation}
\Big[ A(F), \FockHopt \Big] = A(\mathcal{A}(t)F)\,,\qquad\i\partial_t\BogUt  A(F)\BogUt^* = \left[\FockHopt,A(\BogV_t F)\right]\,,
\end{equation}
which yields \eqref{eqn:app:B:1} for $n=1$. Given \eqref{eqn:app:B:1} for $n-1$,
\begin{eqnarray}
\i\partial_t\BogUt A(F_1)\mycdots A(F_n)\bPhi_0
&=&\left[\FockHopt,A(\BogV_tF_1)\right]\BogUt A(F_2)\mycdots A(F_n)\bPhi_0 \nonumber\\
&&+A(\BogV_tF_1)\FockHopt\BogUt A(F_2)\mycdots A(F_n)\bPhi_0\nonumber\\
&=&\FockHopt\BogUt A(F_1)\mycdots A(F_n)\bPhi_0\,,
\end{eqnarray} 
which completes the induction.
\end{proof}

\end{proof}

Finally, we provide estimates on the time dependence of the solution $\BogV(t,s)$
of \eqref{eqn:V(t,s)}:
\begin{lem}\label{lem:U:V:norms}
Let $\BogV(t,0)$ denote the (weak) solution of \eqref{eqn:V(t,s)} with initial condition $\BogV(0,0)=\id$, i.e.,
\begin{equation}
\BogV(t,0)=\begin{pmatrix}
U_{t,0} & \overline{V}_{t,0} \\ V_{t,0} & \overline{U}_{t,0} 
\end{pmatrix}\,,\qquad
\BogV(0,0)=\begin{pmatrix}
\id_\fH& 0\\0&\id_\fH
\end{pmatrix}.
\end{equation}
Then there exists a constant $C>0$ such that
\begin{equation}
\onorm{U_{t,0}}\ls \e^{Ct}\,,\qquad
\norm{V_{t,0}}_\mathrm{HS} \ls \e^{Ct}\,.
\end{equation}
\end{lem}
\begin{proof}
Define $\mathfrak{u}(t,s):\fH\to\fH$ as the unitary two-parameter group (in the weak sense as discussed in Footnote~\ref{footnote_well_posedness}) generated by the self-adjoint operator $\hpt+\Kopt$, i.e.,
\begin{equation}
\i\partial_s \mathfrak{u}(t,s)=-\mathfrak{u}(t,s)\left(h^\ps+\Kops\right)\,.
\end{equation}
Then it follows from  \eqref{eqn:V(t,s)} that
\begin{equation}
\i\partial_t\left(\mathfrak{u}(0,t)U_{t,0}\right) =-\mathfrak{u}(0,t)\Ktpt V_{t,0}\,,
\end{equation}
for $\Ktpt$ the Hilbert--Schmidt operator with kernel $\Ktpt(x_1,x_2)$,
hence
\begin{equation}
\onorm{U_{t,0}}=\onorm{\mathfrak{u}(0,t)U_{t,0}}=\Big\|U_{0,0}+\i\int\limits_0^t\mathfrak{u}(0,s)\Ktps V_{s,0}\ds\Big\|_\mathrm{op}
\leq 1+\int\limits_0^t\onorm{\Ktps V_{s,0}}\ds\,.
\end{equation}
Since 
\begin{eqnarray}
\onorm{\Ktps V_{s,0}}^2
&\leq&\norm{\Ktps V_{s,0}}^2_\mathrm{HS} \nonumber\\
&=&\Tr\left(\Ubar_{s,0}\Ubar_{s,0}^*\big(\Ktps\big)^*\Ktps \right)-\Tr\left(\big(\Ktps\big)^*\Ktps\right)\nonumber\\
&\leq& \onorm{U_{s,0}}^2\norm{\Ktps(\cdot,\cdot)}^2_{\fH^2}
\;\ls\; \onorm{U_{s,0}}^2
\end{eqnarray}
by \eqref{eqn:rel:U:V} and Lemma \ref{lem:Kjpt:bounded:op}, the first statement follows with Gronwall's lemma.
For the second part of the lemma, one defines $\overline{\mathfrak{u}}(t,s):\fH\to\fH$ as the unitary group generated by $-\big(\hpt+\overline{\Kopt\big)}$, and analogously to above obtains
\begin{equation}
\overline{\mathfrak{u}}(0,t)V_{t,0}=-\i\int\limits_0^t\overline{\mathfrak{u}}(0,s)\overline{\Ktps}U_{s,0}\ds\,,
\end{equation}
hence the previous result implies that
\begin{equation}
\norm{V_{t,0}}_\mathrm{HS} \leq \int\limits_0^t \norm{\Ktps}_\mathrm{HS}\onorm{U_{s,0}}\ds \ls \e^{Ct}\,.
\end{equation}
\end{proof}

\section{Proofs}
\label{sec:proofs}

In the following, we will abbreviate $\norm{\cdot}\equiv\norm{\cdot}_\Fock$ and $\lr{\cdot\,,\cdot}\equiv\lr{\cdot\,,\cdot}_\Fock$.

\subsection{Auxiliary Estimates}\label{subsec:proofs:aux}
Recall that for any $F:\R_0^+ \to \R$ and $j\in\{-1,1\}$, it holds that 
\begin{equation} \label{eqn:aux:1}
a_x^{\sharp_j} F(\Number)= F(\Number-j)a_x^{\sharp_j} \,,\qquad
F(\Number)a_x^{\sharp_j} = a_x^{\sharp_j} F(\Number+j) \,.
\end{equation} 
Second-quantized operators can be bounded in terms of $\Number$:
\begin{lem}\label{lem:aux:new}
Let $n,p\geq 0$, $f:\fH^p\to\fH^n$ be a bounded operator with (Schwartz) kernel $f(x^{(n)};y^{(p)})$, and $\bPhi\in\Fock$. Then
\begin{eqnarray}
\left\|\int\dx^{(n)}\dy^{(p)}f(x^{(n)};y^{(p)})\ad_{x_1}\,\mycdots\,\ad_{x_n}\,a_{y_{1}}\,\mycdots\,a_{y_{p}}\bPhi\right\|
&\leq& \norm{f}_{\fH^{p}\to\fH^n}\norm{(\Number+n)^\frac{n+p}{2}\bPhi}\nonumber\\
&\leq& n^\frac{n+p}{2}\norm{f}_{\fH^{p}\to\fH^n}\norm{(\Number+1)^\frac{n+p}{2}\bPhi}\,.\qquad
\end{eqnarray}
\end{lem}
\begin{proof}
Abbreviating $z_j^{(n)}=(z_{j_1}\mydots z_{j_n})$, we compute for $k \geq n$
\begin{eqnarray}
&&\hspace{-1cm}\left\|\left[\int\dx^{(n)}\dy^{(p)}f(x^{(n)};y^{(p)})\ad_{x_1}\mycdots\ad_{x_n}a_{y_1}\mycdots a_{y_p}\bPhi\right]^{(k)}\right\|_{\fH^k} \nonumber\\
&\leq& \sqrt{\frac{(k-n+p)!}{k!}} 
\sum\limits_{\substack{j_1\neq\dots\neq j_n\\\in\{1\mydots k\}}}
\left(\int\dz^{(k)}\left|\int\dy^{(p)}f(z_j^{(n)};y^{(p)})\phi^{(k-n+p)}\big(z^{(k)},y^{(p)}\setminus z_j^{(n)}\big)\right|^2\right)^\frac12\nonumber\\
&\leq& (k+p)^\frac{p+n}{2}\norm{f}_{\fH^p\to\fH^n}\norm{\phi^{(k-n+p)}}_{\fH^{k-n+p}}\,,
\end{eqnarray}
and summing over $k\geq0$ completes the proof. 
\end{proof}

We now collect some useful properties of $\FockHpt$, $\FockHnpt$ and $\mathbb{K}^{(j)}_\pt$. 

\begin{lem}\label{lem:Kjpt:bounded:op}
\lemit{
\item For  $\Kopt$ to $\Kfpt$ as defined in \eqref{K}, it holds that
\begin{equation}\begin{split}
&\norm{\Kopt}_{\fH\to\fH}\leq\norm{v}_\infty\,, \qquad
\norm{\Ktpt}_{\fH^2}\leq\norm{v}_\infty\,,\\
&\norm{\Kthpt}_{\fH\to\fH^2}\ls\norm{v}_\infty\,,\qquad
\norm{\Kfpt}_{\fH^2\to\fH^2}\ls\norm{v}_\infty\,.
\end{split}\end{equation}

\item \label{lem:K:a:norms}
Let $b\geq0$. For $\mathbb{K}_\pt^{(j)}$ as in \eqref{eqn:K:notation} and any $\bPhi\in\Fock$, it holds that
\begin{subequations}
\begin{eqnarray}
\norm{(\Number+1)^b\boldKopt\bPhi} &\ls& \norm{(\Number+1)^{b+1}\bPhi}\\
\norm{(\Number+1)^b\mathbb{K}^{(j)}_\pt\bPhi}&\leq& C(b)\norm{(\Number+1)^{b+\frac{j}{2}}\bPhi} \quad\text{ for } j=2,3,4\,,\\[5pt]
\norm{(\Number+1)^b\overline{\mathbb{K}^{(j)}_\pt}\bPhi} &\leq& C(b) \norm{(\Number+1)^{b+\frac{j}{2}}\bPhi} \quad \text{  for  }j=2,3
\end{eqnarray}
\end{subequations}
for some constant $C(b)>0$.

\item \label{lem:H:norms}
Let $n\in\N$ and $b\geq0$. For $\FockHnpt$ as in \eqref{FockHnpt}  and any $\bPhi\in\Fock$, 
\begin{equation}
\norm{(\Number+1)^b\FockHnpt\bPhi} \ls C(n,b) \norm{(\Number+1)^{b+\frac{n+2}{2}}\bPhi}
\end{equation}
for some constant $C(n,b)>0$.
}

\end{lem}

\begin{proof}
Since  $\norm{\pt}_\fH=1$ and $\norm{\Wpt}_\infty \ls \norm{v}_\infty$,
part (a) follows directly from the definition \eqref{K}, and parts (b) and (c) are implied by part (a) and Lemma \ref{lem:aux:new}.
\end{proof}

Finally, note that $\FockHpt$ leaves the subspace $\FN$ invariant because the first and last line in \eqref{eq_H_N_Fock} preserve the particle number and the remaining terms are zero on the sectors with more than $N$ particles. 
Moreover, all terms in $\FockHpt$ except for $\boldKzpt$ preserve the subspace $\Fpt$, hence $\FockHpt$ leads out of the subspace $\Fpt$ but the time evolution generated by $\FockHNpt$ maps $\FNps$ into $\FNpt$.
We summarize these observations in the following lemma:
\begin{lem}\label{lem:HN:preserves:truncation}
For $\FockHpt$ as in \eqref{eq_H_N_Fock}, $\mathbb{K}^{(j)}_\pt$ as in \eqref{eqn:K:notation} and $j\neq 0$,
\begin{eqnarray}
\FockHpt\FN&\subseteq&\FN\,,\\
\mathbb{K}^{(j)}_\pt\Fpt\,,\, \overline{\mathbb{K}^{(j)}_\pt}\Fpt&\subseteq&\Fpt\,\,,\\
\left(\FockHNpt-\boldKzpt\right)\FNpt&\subseteq&\FNpt\,.
\end{eqnarray}
\end{lem}

\subsection{Proof of Lemma \ref{cor:ass}}\label{subsec:proofs:cor:ass}

Combining \eqref{eqn:Chil:0:normal:order}, Lemma \ref{lem:aux:new} and \eqref{def:Chi0:0}, we find for part (a) that
\begin{equation}
\norm{(\Number+1)^\frac{b}{2}\Chil(0)})
\ls C(\l,b)\norm{(\Number+1)^\frac{3\l+b}{2}\BogUz\ad(f_1)\mycdots \ad(f_{\tilde{\nu}})|\Omega\rangle}
\ls C(\l,b)
\end{equation}
by Lemma \ref{lem:number:BT}. Part (b) follows from part (a) and \eqref{eqn:initial:expansion} by a simple triangle argument since
\begin{eqnarray}
&&\hspace{-1cm}\left\|(\Number+1)^{\frac{b}{2}}\ChiN(0)\right\|_{\FN} \nonumber\\
&&\leq\; \left\|(\Number+1)^{\frac{b}{2}}\bigg(\ChiN(0)-\sum\limits_{\l=0}^{a}\lN^\frac{\l}{2}\Chil(0)\bigg)\right\|_{\FN}
+\;\sum\limits_{\l=0}^{a}\lN^\frac{\l}{2} \left\|(\Number+1)^{\frac{b}{2}}\Chil(0)\right\|_{\FN} \nonumber\\
&&\ls\; C(b)\left((N+1)^{\frac{b}{2}}N^{-\frac{a+1}{2}} + 1\right)
\end{eqnarray}
as $\Number\leq N$ in the sense of operators on $\FN$.\qed

\subsection{Proof of Proposition \ref{thm:duhamel}}
\label{subsec:proofs:duhamel}

First, we expand the $N$-dependent square roots in \eqref{eq_H_N_Fock} in a Taylor series and estimate the remainders.
\begin{lem}\label{lem:calculus}
Let $a\in\N_0$ and define $c_\l^{(n)}$ and $d_{\l,j}$ as in \eqref{eqn:taylor:coeff} and \eqref{eqn:taylor:coeff:2}.
Then
\begin{equation}\label{eqn:cln:bound}
|c_\l| \leq \frac{1}{2\l}\,,\qquad |c_\l^{(j)}| \leq 2^{j+\l-1}\,,\qquad |d_{\l,j}|\leq 2^\l(j+1)\,.
\end{equation}
\lemit{
\item \label{lem:calculus:1}
Define the operator $R^{(3)}_a$ on $\Fock$ via the identity
\begin{equation}\label{eqn:lem:calculus:R3:1}
\frac{\sqrt{\big[N-\Number\big]_+}}{N-1}
\;=\; \sum\limits_{\l=0}^a c_\l\, \lN^{\l+\frac12}
(\Number-1)^\l + \lN^{a+\frac32}\,R^{(3)}_a\,.
\end{equation}
Then $[R^{(3)}_a,\Number]=0$ and it holds for any  $\bPhi\in\Fock$ that
\begin{equation}\label{eqn:lem:calculus:R3:2}
\norm{R^{(3)}_a\bPhi}
\;\leq\; 2^{a+1}\norm{(\Number+1)^{a+1}\bPhi}\,.
\end{equation} 

\item \label{lem:calculus:2}
Define the operator $R^{(2)}_a$ on $\Fock$ through 
\begin{equation}\label{eqn:lem:calculus:R2:1}
\frac{\sqrt{\big[(N-\Number)(N-\Number-1)\big]_+}}{N-1}
\;=\; \sum\limits_{\l=0}^a \lN^\l \sum\limits_{j=0}^\l d_{\l,j} (\Number-1)^j\,
+\lN^{a+1} R^{(2)}_a\,.
\end{equation}
Then $[ R^{(2)}_a,\Number]=0$ and it holds for any  $\bPhi\in\Fock$ that
\begin{equation}\label{eqn:lem:calculus:R2:2}
\norm{R^{(2)}_a\bPhi}\;\leq\; (a+1)^24^{a+1} \norm{(\Number+1)^{a+1}\bPhi}\,.
\end{equation} 
}
\end{lem}
To obtain \eqref{eqn:lem:calculus:R3:2} and \eqref{eqn:lem:calculus:R2:2}, note that the sums in \eqref{eqn:lem:calculus:R3:1} and \eqref{eqn:lem:calculus:R2:1} are the Taylor expansions of the corresponding square roots, which converge on $\FNp$ (for $R^{(3)}_a$) and $\Fock_\perp^{\leq N-1}$ (for $R^{(2)}_a$), respectively. On $\Fock^{\geq N}$ and $\Fock_\perp^{> N-1}$, respectively, one observes that $\Number\geq \lN^{-1}$. The full proof is given in Appendix \ref{appendix:calculus}.\medskip

Next, we  make use of this result to prove the expansion of $\FockHpt$ from Lemma \ref{lem:expansion:HNpt}.\medskip

\noindent\textbf{Proof of Lemma \ref{lem:expansion:HNpt}.}
By \eqref{eq_H_N_Fock}, the Taylor expansion of Lemma \ref{lem:calculus} leads to \eqref{eqn:lem:expansion:HNpt} with $\FockHnpt$ as in \eqref{FockHnpt}. The remainders in \eqref{eqn:lem:expansion:HNpt} are given as
\begin{subequations}
\begin{eqnarray}
\mathcal{R}^{(0)}
&:=&\left(\boldKthpt \sqrt{\frac{\left[N-\Number\right]_+}{N-1}}+\hc\right)\nonumber\\
&&+\lN^{\frac12}\left(\boldKfpt-(\Number-1)\boldKopt
+ \left(\boldKtpt R^{(2)}_0+\hc\right)\right)\,,\label{FockHt:R:1}
\\[3pt]
\mathcal{R}^{(1)}
&:=&\boldKfpt-(\Number-1)\boldKopt
+\left(\boldKtpt R^{(2)}_0+\hc\right)+\lN^{\frac12}\left(\boldKthpt R^{(3)}_0+\hc\right),\qquad\label{FockHt:R:2}
\\[3pt]
\mathcal{R}^{(2n)} &:= &
\boldKthpt R^{(3)}_{n-1}+\lN^{\frac12} \boldKtpt R^{(2)}_{n}+\hc,\label{FockHt:R:2n-1}
\\[3pt]
\mathcal{R}^{(2n+1)} &:= &
\boldKtpt R^{(2)}_{n}+\lN^{\frac12}\boldKthpt R^{(3)}_{n}+\hc\label{FockHt:R:2n}
\end{eqnarray}
\end{subequations}
for $n\geq 1$,  with $R^{(2)}_j$ and $R^{(3)}_j$ from Lemma \ref{lem:calculus}.
Hence, Lemmas \ref{lem:K:a:norms} and \ref{lem:calculus} imply
\begin{eqnarray}
\norm{\mathcal{R}^{(2n)}\bPhi}
&\leq& C(n)\left(\norm{(\Number+1)^{n+\frac32}\bPhi}
+\lN^{\frac12}\norm{(\Number+1)^{n+2}\bPhi}\right),\\
\norm{\mathcal{R}^{(2n+1)}\bPhi}
&\leq&C(n)\left(\norm{(\Number+1)^{n+2}\bPhi}
+\lN^{\frac12}\norm{(\Number+1)^{n+\frac52}\bPhi}\right)
\end{eqnarray}
for any $n\in\N_0$, where we used that $\sqrt{\frac{\left[N-\Number\right]_+}{N-1}}\ls1$ in the sense of operators on $\Fock$.
\qed\\

\noindent\textbf{Proof of Proposition \ref{thm:duhamel}.}
First, we construct the explicit form  \eqref{solution_chil_integral_form} of $\Chil(t)$ (Step 1). Second, we conclude that $\Chil(t)\in\Fpt$ (Step 2), and finally  we derive the expression \eqref{eqn:tilde:H} for the operators $\tilde{\FockH}^{(n)}_{t,s}$ (Step 3).\medskip

\noindent\textbf{Step 1}.
Iterating \eqref{eqn:int:form:Chil:H:tilde} and using that $\cU_{\BogV(t,s_1)}\cU_{\BogV(s_1,s_2)}=\cU_{\BogV(t,s_2)}$, we find that
\begin{eqnarray}
\Chil(t) &=& \BogUtz\Chil(0) + (-\i) \sum\limits_{j_1=1}^\l \int\limits_0^t \ds_1\, \tilde{\mathbb{H}}^{(j_1)}_{t,s_1}\, \mathcal{U}_{\BogV(t,s_1)} \Chi_{\l-j_1}(s_1) \nonumber\\
&=& \BogUtz\Chil(0) + (-\i) \sum\limits_{j_1=1}^\l \int\limits_0^t \d s_1\, \tilde{\mathbb{H}}^{(j_1)}_{t,s_1}\, \mathcal{U}_{\BogV(t,0)} \Chi_{\l-j_1}(0) \nonumber\\
&& + (-\i)^2 \sum\limits_{j_1=1}^\l \sum\limits_{j_2=1}^{\l-j_1} \int\limits_0^t \d s_1 \int\limits_0^{s_1} \d s_2\, \tilde{\mathbb{H}}^{(j_1)}_{t,s_1}\, \tilde{\mathbb{H}}^{(j_2)}_{t,s_2}\, \mathcal{U}_{\BogV(t,s_2)} \Chi_{\l-j_1-j_2}(s_2)\,,
\end{eqnarray}
and successive iterations yield
\begin{eqnarray}\label{eqn:aux:Ijt}
\Chil(t)-\BogUtz\Chil(0)\;=\;\sum\limits_{m=1}^\l\sum\limits_{n=m}^\l\sum\limits_{\substack{\bj\in\N^m\\|\bj|=n}}\mathbb{I}^{(\bj)}_t\Chi_{\l-n}(0)
\;=\;\sum\limits_{n=0}^{\l-1}\sum\limits_{m=1}^{\l-n}\sum\limits_{\substack{\bj\in\N^m\\|\bj|=\l-n}}\mathbb{I}^{(\bj)}_t\Chi_n(0)\,,
\end{eqnarray}
where we abbreviated 
\begin{equation}\label{eqn:aux:Ijt:2}
\mathbb{I}^{(\bj)}_t:=(-\i)^m\int\limits_0^t\ds_1\,\mycdots\hspace{-0.1cm}\int\limits_0^{s_{m-1}}\ds_m\tilde{\mathbb{H}}^{(j_1)}_{t,s_1} \mycdots\, \tilde{\mathbb{H}}^{(j_m)}_{t,s_m}\, \mathcal{U}_{\BogV(t,0)}\,.
\end{equation}
\medskip

\noindent\textbf{Step 2.}
Combining Lemmas \ref{lem:number:BT} and \ref{lem:H:norms}, we find that
\begin{eqnarray}
\norm{(\Number+1)^b\tilde{\FockH}^{(j)}_{t,s}\bPhi}
&=&\norm{(\Number+1)^b\BogUts\FockHps^{(j)}\BogUts^*\bPhi}\nonumber\\
&\leq& C(b,j)C_{\BogV(t,s)}^{2b+1+\frac{j}{2}} \norm{(\Number+1)^{b+1+\frac{j}{2}}\bPhi}\nonumber\\
&\lesssim& \e^{C(b,j)t}\norm{(\Number+1)^{b+1+\frac{j}{2}}\bPhi}\,,
\end{eqnarray}
where we used in the last step that
\begin{equation}
C_{\BogV(t,s)}=2\norm{V_{t,s}}_\HS^2 + \onorm{U_{t,s}}^2 + 1 \ls \e^{Ct}
\end{equation}
by Lemma \ref{lem:U:V:norms}. Consequently,  for any $\bj\in\N^m$ with $|\bj|=\l-n$,
\begin{equation}\label{eqn:aux:Ijt:3}
\norm{\tilde{\mathbb{H}}^{(j_1)}_{t,s_1} \mycdots \,\tilde{\mathbb{H}}^{(j_m)}_{t,s_m}\BogUtz\Chi_n(0)}
\ls \e^{C(\l)t}\norm{(\Number+1)^{\frac{\l-n}{2}+m}\Chi_n(0)}
\end{equation}
and we conclude that
\begin{eqnarray}
\norm{\Chil(t)} \ls \sum\limits_{n=0}^{\l-1}\sum\limits_{m=1}^{\l-n}\sum\limits_{\substack{\bj\in\N^m\\|\bj|=\l-n}} t^m \e^{C(\l)t}\norm{(\Number+1)^{\frac{\l-n}{2}+m}\Chi_n(0)}<\infty
\end{eqnarray}
for $\Chi_n(0)\in\D(\Number^{\frac{3}{2}(\l-n)})$. Finally, $\Chil(t)\in\Fpt$ because $\Chi_n(0)\in\Fpz$ and since $\cU_{\BogV(s_1,s_2)}$ maps $\Fock_{\perp\varphi(s_1)}$ to $\Fock_{\perp\varphi(s_2)}$ by Lemma \ref{lem:time:dep:BT}. \medskip

\noindent\textbf{Step 3.}
Recall that by \eqref{eqn:trafo:ax}, $\BogUts$ transforms creation and annihilation operators as
\begin{equation}\label{eqn:trafo:ax:lj}
\BogUts\asl_x\BogUts^* = \sum\limits_{j\in\{-1,1\}}\int\dy\,\omlj_{t,s}(x;y)\asj_y\,,\qquad \l\in\{-1,1\}\,,
\end{equation}
with $\omega^{(\l,j)}_{t,s}$ as in \eqref{abbrv:omega}.
Thus, abbreviating $z_j:=(y_j;x_j)$, we find
\begin{subequations}
\begin{eqnarray}
\BogUts\Number\BogUts^*\hspace{-5pt}&=&\hspace{-5pt}\sum\limits_{\bj\in\{-1,1\}^2}\int\dx^{(2)}\dy\,\omojo_{t,s}(y;x_1)\omzjt_{t,s}(y;x_2)\asjo_{x_1}\asjt_{x_2}\,,\label{eqn:BT:number:op}\\
\BogUts\boldKops\BogUts^*\hspace{-5pt}&=&\hspace{-5pt}\sum\limits_{\bj\in\{-1,1\}^2}\int\dx^{(2)}\dy^{(2)} \Kops(y^{(2)})\omojo_{t,s}(z_1)\omzjt_{t,s}(z_2)\asjo_{x_1}\asjt_{x_2}\,,\qquad\qquad  
\end{eqnarray}
\end{subequations}
and analogously for $\boldKtpt$ to $\boldKfpt$.
Using that $\cU_{\BogV(t,s_1)}\,\cU_{\BogV(s_1,s_2)}\;=\;\cU_{\BogV(t,s_2)}$ and abbreviating
$\nu(k):=\left\lfloor\tfrac{k-2}{2}\right\rfloor$
for $\lfloor r\rfloor =\max\{z\in\mathbb{Z}:z\leq r\}$, we obtain \eqref{eqn:tilde:H} with the coefficients 
\begin{subequations}\label{eqn:A:coeff}
\begin{eqnarray}
&&\hspace{-1cm}\mathfrak{A}^{(\bj)}_{2n-1,k}(t,s;\xk) \nonumber\\
&=& (-1)^{n-\nu(k)-1}c_{n-1} \tbinom{n-1}{\nu(k)}\Bigg[
\int \dy^{(3)}\Kthps(y^{(3)}) \omojo_{t,s}(z_1)\omojt_{t,s}(z_2)\omzjth_{t,s}(z_3)  \nonumber\\
&& \times \prod\limits_{\mu=2}^{(k-1)/2}\left(\int\dy_{2\mu}\omega^{(1,j_{2\mu})}_{t,s}(z_{2\mu}) \omega^{(-1,j_{2\mu+1})}_{t,s}(y_{2\mu};x_{2\mu+1})\right)  \nonumber\\
&&+\prod\limits_{\mu=1}^{(k-3)/2}\left(\int\dy_{2\mu-1}\omega^{(1,j_{2\mu-1})}_{t,s}(z_{2\mu-1})\omega^{(-1,j_{2\mu})}_{t,s}(y_{2\mu-1};x_{2\mu})\right) \nonumber \\
&&\times\int\dy_{k-2}\dy_{k-1}\dy_k\big(\Kthps\big)^*(y_{k-2},y_{k-1},y_{k})  \nonumber\\
&&\times\omega^{(1,j_{k-2})}_{t,s} (z_{k-2})\omega^{(-1,j_{k-1})}_{t,s}(z_{k-1}) \omega^{(-1,j_k)}_{t,s}(z_k)\Bigg]
\end{eqnarray}
for $n\geq1$, and
\begin{eqnarray}
&&\hspace{-1cm}\mathfrak{A}^{(\bj)}_{2n,k}(t,s;\xk)  \nonumber\\
&=&\tfrac12\sum\limits_{\mu=\nu(k)}^n \left((-1)^{\mu-\nu(k)}d_{n,\mu}\tbinom{\mu}{\nu(k)}\right) \Bigg[
\int\dy^{(2)}\Ktps(y^{(2)})\omojo_{t,s}(z_1)\omojt_{t,s}(z_2) \nonumber \\
&&\times\prod\limits_{m=1}^{(k-2)/2}\left(\int\dy_{2m+1}\omega^{(1,j_{2m+1})}_{t,s}(z_{2m+1}) \omega^{(-1,j_{2m+2})}_{t,s}(y_{2m+1};x_{2m+2})\right)  \nonumber\\
&&+
\prod\limits_{m=1}^{(k-2)/2}\left(\int\dy_{2m-1}\omega^{(1,j_{2m-1})}_{t,s}(z_{2m-1}) \omega^{(-1,j_{2m})}_{t,s}(y_{2m-1};x_{2m})\right) \nonumber \\
&&\times \int\dy_{k-1}\dy_k\overline{\Ktps(y_{k-1},y_k)}\omega^{(-1,j_{k-1})}_{t,s}(z_{k-1})\omega^{(-1,k)}_{t,s}(z_k)
\Bigg]
\end{eqnarray}
for $n\geq2$. The remaining coefficients can be expressed as 
\begin{eqnarray}
&&\hspace{-1cm}\mathfrak{A}^{(\bj)}_{2,2}(t,s;x^{(2)})\nonumber\\
&=&\int\dy^{(2)}\Bigg[\Kops(y^{(2)})\omojo_{t,s}(z_1)\omzjt_{t,s}(z_2) + \tfrac14\Ktps(y^{(2)})\omojo_{t,s}(z_1)\omojt_{t,s}(z_2)\nonumber\\
&&+ \tfrac14\overline{\Ktps(y^{(2)})}\omzjo_{t,s}(z_1)\omzjt_{t,s}(z_2)\Bigg]\,,
\end{eqnarray}
\begin{eqnarray}
&&\hspace{-1cm}\mathfrak{A}^{(\bj)}_{2,4}(t,s;x^{(4)})\nonumber\\
&=&-\int\dy_1\dy_3\dy_4\Bigg[ \omojo_{t,s}(z_1)\omzjt_{t,s}(y_1;x_2) \Kops(y_3,y_4)\omojth_{t,s}(z_3)\omzjf_{t,s}(z_4)\nonumber\\
&& + \tfrac12\Ktps(y_1,y_3)\omojo_{t,s}(z_1)\omojt_{t,s}(y_3;x_2) \omojth_{t,s}(y_4;x_3)\omzjf_{t,s}(z_4)\nonumber\\
&& + \tfrac12\omojo_{t,s}(z_1)\omzjt_{t,s}(y_1;x_2)\overline{\Ktps(y_3,y_4)}\omzjth_{t,s}(z_3)\omzjf_{t,s}(z_4)\Bigg]\nonumber\\
&&+\tfrac12\int\dy^{(4)}\Kfps(y^{(4)})\omojo_{t,s}(z_1)\omojt_{t,s}(z_2)\omzjth_{t,s}(z_3)\omzjf_{t,s}(z_4)\,.
\end{eqnarray}
\end{subequations}
\qed

\subsection{Proof of Theorem \ref{thm:norm:approx}}
\label{subsec:proofs:norm:approx}

We begin with proving that under Assumption~\ref{ass:initial:data}, moments of $\Number$ with respect to $\Chil(t)$ and $\ChiN(t)$ remain bounded uniformly in $N$ under the time evolution.

\begin{lem}\label{lem:moments}
Let Assumption \ref{ass:initial:data} hold for some $\tilde{a}\in\N_0$ and let $t\in\R$, $b\in\N_0$ and $0\leq\l\leq\tilde{a}$.
\lemit{
\item  \label{lem:moments:Chil}
There exists a constant $C(\l,b)>0$ such that
\begin{equation}
\lr{\Chil(t), (\Number+1)^b\Chil(t)}_\Fock \;\ls\;  \e^{C(\l,b)t}\,.
\end{equation} 
\item \label{lem:moments:Chil:N}
Let $c\in \N_0$. There exists a constant $C(\l,b,c)>0$ such that
\begin{equation}
\lr{\Chil(t), (\Number+1)^b\Chil(t)}_{\Fock^{\geq N}} \;\ls\; N^{-c} \,\e^{C(\l,b,c)t}\,.
\end{equation}
\item \label{lem:moments:ChilN:full:evolution}
There exists a constant $C(b)>0$ such that for all $0 \leq b \leq \tilde{a}+1$,
\begin{equation}
\lr{\ChiN(t), (\Number+1)^b\ChiN(t)}_{\FN} \;\ls\;  \e^{C(b)t}\,.
\end{equation}
}
\end{lem}

 Part (a) for $\l=0$ and  part (c) are standard results, which are proven, e.g., in \cite[Proposition~3.3]{rodnianski2009}, \cite[Lemma~2.3]{mitrouskas2016}, and \cite[Proposition 4 (iii)]{nam2015}.

\begin{proof}
For part (a), we infer from \eqref{eqn:aux:Ijt}, \eqref{eqn:aux:Ijt:3} and Lemma \ref{lem:number:BT} that  
\begin{eqnarray}
\norm{(\Number+1)^\frac{b}{2}\Chil(t)}
&\leq& \norm{(\Number+1)^\frac{b}{2}\BogUtz\Chil(0)} 
+\sum\limits_{n=0}^{\l-1}\sum\limits_{m=1}^{\l-n}\sum\limits_{\substack{\bj\in\N^m\\|\bj|=\l-n}}\norm{(\Number+1)^\frac{b}{2}\mathbb{I}_t^{(\bj)}\Chi_n(0)}\nonumber\\
&\ls& \e^{C(\l,b)t}\sum\limits_{n=0}^\l\norm{(\Number+1)^{\frac{b+3(\l-n)}{2}}\Chi_n(0)}
\end{eqnarray}
for $\mathbb{I}^{(\bj)}_t$ as in \eqref{eqn:aux:Ijt:2}.
For part (b), note that $\Number\geq N$ on $\Fock^{\geq N}$, hence
\begin{equation}
\left\| (\Number+1)^{\frac{b}{2}}\Chil(t) \right\|_{\Fock^{\geq N}}
\leq \left\| \frac{(\Number+1)^{{\frac{b+c}{2}}}}{N^{\frac{c}{2}}}\Chil(t) \right\|_{\Fock^{\geq N}} 
\ls N^{-{\frac{c}{2}}} \, \e^{C(b,\l,c)t}
\end{equation}
 for any $c\in \N_0$ by part (a).
\end{proof}

\noindent\textbf{Proof of Theorem \ref{thm:norm:approx}.}
Let us abbreviate 
\begin{equation}\label{tilde_abbreviation}
\tilde{\Chi}_a(t):=\ChiN(t)\oplus0-\sum\limits_{\l=0}^a \lN^{\frac{\l}{2}}\Chil(t)\,.
\end{equation}
For simplicity of presentation, we first do a formal computation. Note that $(\FockHNpt\ChiN)\oplus0=\FockHpt(\ChiN\oplus0)$, hence \eqref{eqn:SE:Fock} and \eqref{eqn:differential:form:Chil} imply
\begin{align}\label{time_derivative_chi_tilde}
&\partial_t\norm{\tilde{\Chi}_a(t)}^2_{\FN}\nonumber\\
&\quad= 2\,\Im \lr{\tilde{\Chi}_a(t),\left(\FockHpt \ChiN(t)\oplus0-\sum\limits_{\l=0}^a\lN^{\frac{\l}{2}}\sum\limits_{n=0}^\l \FockHnpt\Chi_{\l-n}(t)\right)}_{\FN} \nonumber\\
&\quad= 2\,\Im \lr{\tilde{\Chi}_a(t),\left(\FockHpt \left(\tilde{\Chi}_a(t)+\sum\limits_{\l=0}^a \lN^{\frac{\l}{2}}\Chil(t) \right)-\sum\limits_{\l=0}^a\lN^{\frac{\l}{2}}\sum\limits_{n=0}^\l \FockHnpt\Chi_{\l-n}(t)\right)}_{\FN}\nonumber\\
&\quad= 2\,\Im \lr{\tilde{\Chi}_a(t),\FockHpt\tilde{\Chi}_a(t)}_{\FN} \nonumber\\
&\quad\quad + 2\sum\limits_{\l=0}^a\lN^{\frac{\l}{2}}\,\Im \lr{\tilde{\Chi}_a(t),\left(\FockHpt \Chil(t) -\sum\limits_{n=0}^\l \FockHnpt\Chi_{\l-n}(t)\right)}_{\FN}.
\end{align}
By self-adjointness of $\FockHpt$ and since $[\FockHpt, \id_{\FN}] = 0$, the first summand vanishes. In a similar way, using only the (well-defined) integral form of \eqref{eqn:SE:Fock} and Definition~\ref{def:Chil}, we rigorously obtain the integrated version of \eqref{time_derivative_chi_tilde}, viz.,
\begin{align}
&\norm{\tilde{\Chi}_a(t)}^2_{\FN} - \norm{\tilde{\Chi}_a(0)}^2_{\FN} \nonumber\\
&\quad= 2\sum\limits_{\l=0}^a\lN^{\frac{\l}{2}}\,\Im \int_0^t \ds\, \lr{\tilde{\Chi}_a(s),\left(\left(\FockHps-\FockHops\right) \Chil(s) -\sum\limits_{n=1}^\l \FockHnps\Chi_{\l-n}(s)\right)}_{\FN}.
\end{align}
By reordering the summation and using Cauchy-Schwarz we find
\begin{align}
&\norm{\tilde{\Chi}_a(t)}^2_{\FN} - \norm{\tilde{\Chi}_a(0)}^2_{\FN} \nonumber\\
&\quad= 2\sum\limits_{\l=0}^a\lN^{\frac{\l}{2}}\,\Im \int_0^t \ds\, \lr{\tilde{\Chi}_a(s),\left(\left(\FockHps-\FockHops\right) \Chil(s) -\sum\limits_{n=1}^{a-\l} \lN^{\frac{n}{2}} \FockHnps\Chi_{\l}(s)\right)}_{\FN} \nonumber\\
&\quad \leq 2 \int_0^t \ds\, \norm{\tilde{\Chi}_a(s)}_{\FN}
\sum\limits_{\l=0}^a\lN^{\frac{\l}{2}}\left\|{ \left(\FockHps - \FockHops - \sum\limits_{n=1}^{a-\l} \lN^{\frac{n}{2}}\FockHnps\right) \Chil(s)}\right\|_{\FN}\nonumber\\
&\quad=2\lN^\frac{a+1}{2} \int_0^t \ds\, \norm{\tilde{\Chi}_a(s)}_{\FN}
\sum\limits_{\l=0}^a\norm{\mathcal{R}^{a-\l} \Chil(s)}_{\FN},
\end{align}
with the definition of $\mathcal{R}^{a}$ from Lemma~\ref{lem:expansion:HNpt}. Then the remainder estimate from Lemma~\ref{lem:expansion:HNpt}, and Lemma~\ref{lem:moments:Chil} yield
\begin{align}
\norm{\tilde{\Chi}_a(t)}^2_{\FN} - \norm{\tilde{\Chi}_a(0)}^2_{\FN} &\leq 2 \lN^\frac{a+1}{2}C(a) \int_0^t \ds\, \norm{\tilde{\Chi}_a(s)}_{\FN}
\sum\limits_{\l=0}^a\norm{(\Number+1)^\frac{a-\l+4}{2} \Chil(s)}\nonumber\\
&\ls \lN^\frac{a+1}{2} \int_0^t \ds\, \e^{C(a)s}\norm{\tilde{\Chi}_a(s)}_{\FN}\,.
\label{eqn:proof:thm1:1}
\end{align}
%where we used Lemmas \ref{lem:expansion:HNpt} and \ref{lem:moments:Chil}, the self-adjointness of $\FockHpt$,  and that
%\begin{eqnarray}
%\sum\limits_{\l=0}^a\lN^{\frac{\l}{2}}\sum\limits_{n=0}^\l\FockHnpt\Chi_{\l-n}(t)
%=\sum\limits_{n=0}^a\sum\limits_{\l=n}^a\lN^{\frac{\l}{2}}\FockHpt^{(\l-n)}\Chi_n(t)
%=\sum\limits_{\l=0}^a\lN^{\frac{\l}{2}}\sum\limits_{n=0}^{a-\l}\lN^{\frac{n}{2}}\FockHnpt\Chil(t)\,.
%\end{eqnarray}
Since $\norm{\tilde{\Chi}_a(0)}^2_{\FN}\leq C(a) N^{-(a+1)}$ by Assumption \ref{ass:initial:data}, Gronwall's lemma implies that
\begin{equation}
\norm{\tilde{\Chi}_a(t)}^2_{\FN} \ls \lN^{a+1}\e^{C(a)t}\,.
%\;\ls\; \left(\norm{\tilde{\Chi}_a(0)}_{\FN}^2+N^{-(a+1)}\int\limits_0^t\e^{C(a)s}\ds\right)\e^{t}
%\;\ls\;\lN^{a+1}\e^{C(a)t}\,.
\end{equation}
\qed

\subsection{Proof of Corollary \ref{thm:even_more_explicit_form}}
\label{subsec:proofs:explicit}
\setcounter{claim}{0}

The proof is an inductive construction of the coefficients in \eqref{eqn:Chil:explicit:final}. We proceed by proving four auxiliary claims; Claim~\ref{claim:4} is the statement of the Corollary.
\begin{claim} \textbf{\emph{(Bogoliubov time evolution of the initial data).}}\label{claim:1}
\begin{equation}\label{eqn:BT:Chil(0)}
\BogUtz\Chil(0) = \sum\limits_{n=1}^\l \sum\limits_{\substack{m=0\\m+n\text{ even}}}^{n+2} \; \sum\limits_{\bj\in\{-1,1\}^m} \int\dx^{(m)}\tilde{\mathfrak{A}}^{(\bj)}_{\l,n,m}\big(t;x^{(m)}\big)
\asjo_{x_1}\,\mycdots \,a_{x_m}^{\sharp_{j_m}} \BogUtz\Chi_{\l-n}(0)\,,
\end{equation}
where
\begin{equation}
\tilde{\mathfrak{A}}^{(\bj)}_{\l,n,m}\big(t;x^{(m)}\big)
 = \sum\limits_{\mu=0}^m 
\int\dy^{(m)}\alnmm(y^{(m)})\prod\limits_{p=1}^\mu\omega_{t,0}^{(1,j_p)}(y_p;x_p)\prod\limits_{q=\mu+1}^{m} \omega_{t,0}^{(-1,j_q)}(y_q;x_q)
\,.\label{eqn:tilde:A}
\end{equation}
\end{claim}

\begin{proof}
Analogous to Step 3 in the proof of Proposition \ref{thm:duhamel}.
\end{proof}\smallskip

\begin{claim} \textbf{\emph{(Induction base case).}}\label{claim:2}
\begin{equation}\label{eqn:Chit(t):aux}
\Chit(t)=\sum\limits_{q=1,3}\sum\limits_{\bj\in\{-1,1\}^q} \int\dx^{(q)} \mathfrak{C}^{(\bj)}_{1,q}(t;x^{(q)}) \asjo_{x_1}\,\mycdots a^{\sharp_{j_q}}_{x_q}\,\Chio(t)\,,
\end{equation}
where
\begin{equation}\label{eqn:C:2:coeff}
\mathfrak{C}^{(\bj)}_{1,q}\big(t;x^{(q)}\big) = \mathfrak{B}^{(\bj)}_{1,q}\big(t;x^{(q)}\big)+ \tilde{\mathfrak{B}}^{(\bj)}_{1,q}\big(t;x^{(q)}\big)
\end{equation}
with
\begin{subequations}\label{eqn:B:base:case}
\begin{eqnarray}
\mathfrak{B}^{(\bj)}_{1,1}(t;x)&=&0\,, \qquad \mathfrak{B}^{(\bj)}_{1,3}(t;x^{(3)})=-\i\int\limits_0^t\mathfrak{A}^{(\bj)}_{1,3}\big(t,s;x^{(3)}\big)\ds\,,\label{eqn:B:2}\\
\tilde{\mathfrak{B}}^{(\bj)}_{1,q}(t;x^{(q)})&=&\tilde{\mathfrak{A}}^{(\bj)}_{1,1,q}(t;x^{(q)})\label{eqn:tilde:B:2}
\end{eqnarray}
\end{subequations}
for $\mathfrak{A}^{(\bj)}_{n,p}$ as in \eqref{eqn:A:coeff} and $\tilde{\mathfrak{A}}^{(\bj)}_{\l,n,m}$ as in \eqref{eqn:tilde:A}.
\end{claim}

\begin{proof}
By \eqref{eqn:int:form:Chil:H:tilde}, we can decompose 
\begin{equation} \label{eqn:Chit(t):int:formula}
\Chit(t)=\BogUtz\Chit(0)+\Chit^\mathrm{int}(t)\,,\qquad
\Chit^\mathrm{int}(t):=-\i\int\limits_0^t\tilde{\FockH}^{(1)}_{t,s}\ds\Chio(t)\,.
\end{equation}
Insertion of \eqref{eqn:BT:Chil(0)} and \eqref{eqn:tilde:H} for $n=1$ yields
\begin{subequations}
\begin{eqnarray}
\Chit^\mathrm{int}(t)&=&\sum\limits_{m=1,3}\sum\limits_{\bj\in\{-1,1\}^m} \int\dx^{(m)}\mathfrak{B}^{(\bj)}_{1,m}\big(t;x^{(m)}\big) \asjo_{x_1}\mycdots\,a_{x_m}^{\sharp_{j_m}}\Chio(t)\,,\\
\BogUtz\Chit(0)&=&\sum\limits_{m=1,3}\sum\limits_{\bj\in\{-1,1\}^m} \int\dx^{(m)}\tilde{\mathfrak{B}}^{(\bj)}_{1,m}\big(t;x^{(m)}\big) \asjo_{x_1}\mycdots\,a_{x_m}^{\sharp_{j_m}}\Chio(t)\label{eqn:U:Chit(0):base:case}
\end{eqnarray}
\end{subequations}
with $\mathfrak{B}_{1,m}^{(\bj)}$ and $\tilde{\mathfrak{B}}_{1,m}^{(\bj)}$ as in \eqref{eqn:B:base:case}.
\end{proof}\smallskip

\begin{claim}\textbf{\emph{(Induction for the Bogoliubov time evolved initial data)}}\label{claim:3}
\begin{equation}\label{eqn:induction:hyp:Chil(0)}
\cU_{\BogV(t,0)}\Chil(0) = \sum\limits_{\substack{0\leq q\leq 3\l\\q+\l\text{ even}}}\; \sum\limits_{\bj\in\{-1,1\}^q} 
\int\dx^{(q)}\tilde{\mathfrak{B}}^{(\bj)}_{\l,q}\big(t;x^{(q)}\big) \asjo_{x_1}\mycdots\, a_{x_q}^{\sharp_{j_q}}\Chio(t)\,,
\end{equation}
where the coefficients $ \tilde{\mathfrak{B}}^{(\bj)}_{\l,q}$  are determined  by the iteration rule
\begin{subequations}\label{eqn:tilde:B}
\begin{eqnarray}
\tilde{\mathfrak{B}}_{0,0}(t)&:=&1\,,\\
\tilde{\mathfrak{B}}^{(\bj)}_{\l,q}\big(t;x^{(q)}\big)&:=&0\quad \text{ if }q>3\l \text{ or } q+\l\text{ odd}\,,
\end{eqnarray} 
and otherwise
\begin{equation}
\tilde{\mathfrak{B}}^{(\bj)}_{\l,q}\big(t;x^{(q)}\big) 
=\sum\limits_{n=1}^{\min\{\l,\frac{3\l+2-q}{2}\}}
\sum\limits_{\substack{m=\max\{0,q-3(\l-n)\}\\m+n\text{ even}}}^{\min\{q,n+2\}}
\tilde{\mathfrak{A}}^{(j_1\mydots j_m)}_{\l,n,m}\big(t;x^{(m)}\big) \tilde{\mathfrak{B}}^{(j_{m+1}\mydots j_q)}_{\l-n,q-m}\big(t;x_{m+1}\mydots x_q\big)
\end{equation}
with $\tilde{\mathfrak{A}}^{(\bj)}_{\l,n,m}$ as in \eqref{eqn:tilde:A}.
\end{subequations}
\end{claim}

\begin{proof}
We prove the hypothesis \eqref{eqn:induction:hyp:Chil(0)} by induction over $\l\in\N$. The base case $\l=1$ is established by \eqref{eqn:U:Chit(0):base:case}.
Assume \eqref{eqn:induction:hyp:Chil(0)} holds for $1,2,\mydots\l-1$. Then it follows from Claim \ref{claim:1} that
\begin{eqnarray}
\BogUtz\Chil(0)
&=&\sum\limits_{n=1}^\l \sum\limits_{\substack{m=0\\m+n\text{ even}}}^{n+2} \; \sum\limits_{\substack{0\leq q\leq 3(\l-n)\\q+\l-n\text{ even}}}  \sum\limits_{\bj\in\{-1,1\}^{m+q}}
\int\dx^{(m+q)}\tilde{\mathfrak{A}}^{(j_1\mydots j_m)}_{\l,n,m}\big(t;x^{(m)}\big) \nonumber\\
&&\times\tilde{\mathfrak{B}}^{(j_{m+1}\mydots j_{m+q})}_{\l-n,q}\big(t;x_{m+1}\mydots x_{m+q}\big)
\asjo_{x_1}\mycdots\, a_{x_{m+q}}^{\sharp_{j_{m+q}}}\Chio(t)\,.\label{eqn:induction:Chil(0):1}
\end{eqnarray}
We now rearrange the sums in \eqref{eqn:induction:Chil(0):1}. 
Abbreviating the integrand  as $I^{(\bj)}_{n;m;q}$, we find that
\begin{eqnarray}
&&\sum\limits_{n=1}^\l  \;\sum\limits_{\substack{0\leq m\leq n+2\\m+n \text{ even}}} \;\sum\limits_{\substack{0\leq q\leq 3(\l-n)\\q+\l-n\text{ even}}} \; \sum\limits_{\bj\in\{-1,1\}^{m+q}} \int\dx^{(m+q)}
I_{n;m;q}^{(\bj)}\nonumber\\
&&\quad = \sum\limits_{n=1}^\l  \;\sum\limits_{\substack{0\leq m\leq n+2\\m+n \text{ even}}} \;\sum\limits_{\substack{m\leq q\leq m+3(\l-n)\\q+\l\text{ even}}} \; \sum\limits_{\bj\in\{-1,1\}^{q}} \int\dx^{(q)}
I_{n;m;q-m}^{(\bj)}\nonumber\\
&&\quad = \sum\limits_{n=1}^\l  \;\sum\limits_{\substack{0\leq q\leq 3\l+2-2n\\q+\l \text{ even}}} 
\;\sum\limits_{\substack{m=\max\{0,q-3(\l-n)\}\\m+n\text{ even}}}^{\min\{q,n+2\}}\;
 \sum\limits_{\bj\in\{-1,1\}^{q}} \int\dx^{(q)}
I_{n;m;q-m}^{(\bj)}\nonumber\\
&&\quad = \;\sum\limits_{\substack{0\leq q\leq 3\l\\q+\l \text{ even}}} \; \sum\limits_{\bj\in\{-1,1\}^{q}} 
\sum\limits_{n=1}^{\min\{\frac{3\l+2-q}{2},\l\}}\;
\;\sum\limits_{\substack{m=\max\{0,q-3(\l-n)\}\\m+n\text{ even}}}^{\min\{q,n+2\}} 
 \int\dx^{(q)}I_{n;m;q-m}^{(\bj)}\,.
\end{eqnarray}
This closes the induction.
\end{proof}
\smallskip
\begin{claim}\textbf{\emph{(Induction for the full corrections).}}\label{claim:4}
\begin{equation}\label{Chil:explicit}
\Chil(t)=\sum\limits_{\substack{0\leq q\leq 3\l\\q+\l\text{ even}}} \sum\limits_{\bj\in\{-1,1\}^q} \int\dx^{(q)}\mathfrak{C}^{(\bj)}_{\l,q}(t;x^{(q)})\,\asjo_{x_1}\,\mycdots a^{\sharp_{j_q}}_{x_q} \Chio(t)\,,
\end{equation} 
where 
\begin{subequations}\label{eqn:C:coeff}
\begin{equation}
\mathfrak{C}^{(\bj)}_{\l,q}\big(t;x^{(q)}\big):=\tilde{\mathfrak{B}}^{(\bj)}_{\l,q}\big(t;x^{(q)}\big)
\end{equation}
if $q=0,1$, and otherwise
\begin{eqnarray}
\mathfrak{C}_{\l,q}^{(\bj)}\big(t;x^{(q)}\big) &=& \tilde{\mathfrak{B}}^{(\bj)}_{\l,q}\big(t;x^{(q)}\big)
-\i\sum\limits_{n=1}^{\min\{\l,\frac{3\l+2-q}{2}\}} \sum\limits_{\substack{p=\max\{2,q-3(\l-n)\}\\ p+n\text{ even}}}^{\min\{q,n+2\}}
 \sum\limits_{\bm\in\{-1,1\}^{q-p}} \int\limits_0^t\ds   \int\dy^{(q-p)} \nonumber\\ 
&&\times\mathfrak{A}^{(j_1\mydots j_p)}_{n,p}(t,s;x^{(p)})\mathfrak{C}^{(\bm)}_{\l-n,q-p}(s;y^{(q-p)} )
 \prod\limits_{\mu=1}^{q-p} \omega^{(m_\mu,j_{p+\mu})}_{t,s}(y_\mu,x_{p+\mu})
\end{eqnarray}
\end{subequations}
with $\mathfrak{A}^{(\bj)}_{n,p}$ as in \eqref{eqn:tilde:A} and $\tilde{\mathfrak{B}}^{(\bj)}_{\l,q}$ as in \eqref{eqn:tilde:B}.
\end{claim}

\begin{proof}
We prove the hypothesis \eqref{Chil:explicit} by induction over $\l\in\N$. The base case $\l=1$ is established in Claim \ref{claim:2}. By \eqref{eqn:int:form:Chil:H:tilde}, it follows that
\begin{equation}\label{eqn:proof:thm2:1}
\Chil(t)=\BogUtz\Chil(0)+\Chil^\mathrm{int}(t)\,,\qquad 
\Chil^\mathrm{int}(t):=-\i\sum\limits_{n=1}^\l\int\limits_0^t \ds\ \tilde{\FockH}^{(n)}_{t,s}\BogUts\Chi_{\l-n}(s)\,,
\end{equation}
and the first term  is given by Claim \ref{claim:3}.
Assume \eqref{Chil:explicit} holds  for $1,2\mydots \l-1$. Then   by  \eqref{eqn:tilde:H},
\begin{eqnarray}
\Chi_\l^\mathrm{int}(t) \hspace{-2mm}
&=&-\i\sum\limits_{n=1}^\l \int\limits_0^t\ds \sum\limits_{\substack{2\leq p\leq n+2\\p+n \text{ even}}} \;\sum\limits_{\substack{0\leq q\leq 3(\l-n)\\q+\l-n\text{ even}}} \; \sum\limits_{\bj\in\{-1,1\}^{p+q}} \int\dx^{(p+q)} \mathfrak{A}^{(j_1\mydots j_p)}_{n,p}(t,s;x^{(p)})\asjo_{x_1}\,\mycdots a^{\sharp_{j_p}}_{x_p} \nonumber\\
&& \times \mathfrak{C}_{\l-n,q}^{(j_{p+1}\mydots j_{p+q})}(s;x_{p+1}\mydots x_{p+q})
\BogUts  a^{\sharp_{j_{p+1}}}_{x_{p+1}} \BogUts^*\,\mycdots \BogUts a^{\sharp_{j_{p+q}}}_{x_{p+q}} \BogUts^*\Chio(t)\nonumber\\
&=&\sum\limits_{n=1}^\l  \;\sum\limits_{\substack{2\leq p\leq n+2\\p+n \text{ even}}} \;\sum\limits_{\substack{0\leq q\leq 3(\l-n)\\q+\l-n\text{ even}}} \; \sum\limits_{\bj\in\{-1,1\}^{p+q}} \int\dx^{(p+q)}
\Bigg(-\i\sum\limits_{\mathbf{m}\in\{-1,1\}^q} \int\limits_0^t\ds \int\dy^{(q)} \nonumber\\
&&\times
 \mathfrak{A}^{(j_1\mydots j_p)}_{n,p}(t,s;x^{(p)}) \mathfrak{C}^{(\mathbf{m})}_{\l-n,q}(s;y^{(q)} )
\bigg(\prod\limits_{\mu=1}^q \omega^{(m_\mu,j_{p+\mu})}_{t,s}(y_\mu;x_{p+\mu})\bigg)
\Bigg) \nonumber\\
&&\times \asjo_{x_1}\,\mycdots a^{\sharp_{j_{p+q}}}_{x_{p+q}}\Chio(t)\,,\label{eqn:proof:thm2:3}
\end{eqnarray}
where we used that
\begin{eqnarray}
&&\hspace{-1cm}\BogUts a^{\sharp_{j_{p+1}}}_{x_{p+1}} \BogUts^*\,\mycdots \BogUts a^{\sharp_{j_{p+q}}}_{x_{p+q}} \BogUts^* \nonumber\\
& =& \sum\limits_{\mathbf{m}\in\{-1,1\}^q} \int\dy^{(q)} \bigg(\prod\limits_{\mu=1}^q\omega_{t,s}^{(j_{p+\mu},m_\mu)}(x_{p+\mu};y_\mu)\bigg) a^{\sharp_{m_1}}_{y_1}\,\mycdots a^{\sharp_{m_q}}_{y_q}
\end{eqnarray}
and subsequently renamed the variables $x_{p+1}\mydots x_{p+q}$ $\leftrightarrow$ $y_{1}\mydots y_{q}$ as well as the indices $j_{p+1}\mydots j_{p+q}$ $\leftrightarrow$ $m_{1}\mydots m_{q}$.
Reordering the sums as in Claim \ref{claim:3}  and adding \eqref{eqn:induction:hyp:Chil(0)} completes the induction.
\end{proof}

\subsection{Proof of Corollary \ref{lem_compute_correlation_functions}}\label{subsec:proofs:wick}

As a consequence of Corollary~\ref{thm:even_more_explicit_form} and since $(a^{\sharp_j})^\dagger=a^{\sharp_{-j}}$, we obtain
\begin{eqnarray}
\lrt{a^{\sharp_{j_1}}_{x_1}\cdots a^{\sharp_{j_n}}_{x_{n}}}_{\l,k} 
&=&\sum\limits_{\substack{0\leq p\leq 3\l\\p+\l\text{ even}}}\,
\sum\limits_{\substack{0\leq q\leq 3k\\q+k\text{ even}}}\,
\sum\limits_{\substack{\bm\in\{-1,1\}^p\\\boldsymbol{\mu}\in\{-1,1\}^q}}
\int\dy^{(p)}\dz^{(q)}\overline{\mathfrak{C}^{(\bm)}_{\l,p}(t;y^{(p)})}\mathfrak{C}^{(\boldsymbol{\mu})}_{k,q}(t;z^{(q)})\nonumber\\
&&\times \lrt{ a_{y_p}^{\sharp_{-m_p}}\,\mycdots\, a_{y_p}^{\sharp_{-m_1}} \asjo_{x_1}\,\mycdots a^{\sharp_{j_n}}_{x_n}\,a_{z_1}^{\sharp_{\mu_1}}\,\mycdots\,a_{z_q}^{\sharp_{\mu_q}}}_{0,0}\nonumber\\
&=&\sum\limits_{\substack{0\leq p\leq 3\l\\p+\l\text{ even}}}\,
\sum\limits_{\substack{0\leq q\leq 3k\\q+k\text{ even}}}\,
\sum\limits_{\bm\in\{-1,1\}^{n+p+q}}
I_{p,q}\nonumber\\
&=&\sum\limits_{\substack{n\leq b\leq n+3(\l+k)\\b+k+l+n \text{ even}}}
\sum\limits_{\bm\in\{-1,1\}^b}\sum\limits_{\substack{q=\max\{0,b-n-3\l\}\\q+k\text{ even}}}^{\min\{3k,b-n\}} I_{b-n-q,q}\,,
\end{eqnarray}
where we abbreviated 
\begin{eqnarray}
I_{p,q}&:=&\int\dy^{(n+p+q)}
\overline{\mathfrak{C}^{(-m_p\mydots -m_1)}_{\l,p}(t;y_p\mydots y_1)}\,
\mathfrak{C}^{(m_{n+p+1}\mydots m_{n+p+q})}_{k,q}(t;y_{n+p+1}\mydots y_{n+p+q})\nonumber\\
&&\times \prod\limits_{\mu=1}^n\delta(y_{p+\mu}-x_\mu)\delta_{m_{p+\mu},j_\mu}
 \lrt{a_{y_1}^{\sharp_{m_1}}\,\mycdots\, a_{y_{n+p+q}}^{\sharp_{m_{n+p+q}}}}_{0,0}\,.
\end{eqnarray}
Consequently,
\begin{equation}
\lrt{a^{\sharp_{j_1}}_{x_1}\cdots a^{\sharp_{j_n}}_{x_{n}}}_{\l,k} 
=\sum\limits_{\substack{n\leq b\leq n+3(\l+k)\\b+k+l+n \text{ even}}}
\sum\limits_{\bm\in\{-1,1\}^b}\int\dy^{(b)}\mathfrak{D}^{(\bj,\bm)}_{\l,k,n;b}(t;\xn;y^{(b)})
\lrt{a_{y_1}^{\sharp_{m_1}}\,\mycdots\, a_{y_{b}}^{\sharp_{m_{b}}}}_{0,0}\label{eqn:proof:wick:1}
\end{equation}
for $\mathfrak{D}^{(\bj,\bm)}_{\l,k,n;b}$ as defined as in \eqref{eqn:D}. 
By Assumption \ref{ass:initial:data}, the inner product in \eqref{eqn:proof:wick:1} is the $b+2\tilde{\nu}$-point correlation function of a quasi-free state. By Lemma \ref{lem:wick}, it vanishes for $b$ odd (or, equivalently, $k+\l+n$ odd) and decomposes into the sum over all possible  pairings for $b$ even.\qed

\subsection{Proof of Corollary \ref{thm:correlation functions_simplified}}
\label{subsec:proofs:corr}

With the  abbreviation  \eqref{tilde_abbreviation}, we find
\begin{eqnarray}
\lrt{\prod_{i=1}^n \ad_{x_i} \prod_{j=1}^{p} a_{y_j}}_N
&=& \lr{\tilde{\Chi}_a(t), \prod_{i=1}^n \ad_{x_i} \prod_{j=1}^{p} a_{y_j} \ChiN(t)}_{\FN} \nonumber\\
&& + \sum_{\ell=0}^{a} \lambda_N^{\frac{\l}{2}} \lr{\prod_{j=0}^{p-1} \ad_{y_{p-j}} \prod_{i=0}^{n-1} a_{x_{n-i}}\Chil(t),\tilde{\Chi}_{a-\l}(t)} 
\nonumber\\
&& + \sum_{\ell=0}^{a} \lambda_N^{\frac{\l}{2}} \sum_{m=0}^{\ell} \lr{\Chim(t), \prod_{i=1}^n \ad_{x_i} \prod_{j=1}^{p} a_{y_j} \Chi_{\l-m}(t)}\,.\label{eqn:proof:thm:corr:fctns}
\end{eqnarray}
Therefore, using Cauchy-Schwarz and Lemma~\ref{lem:aux:new}, it holds for any bounded $B:\fH^p\to\fH^n$ that
\begin{eqnarray}\label{corr_computation_general}
&&\hspace{-1cm}\Bigg| \int\dx^{(n)}\dy^{(p)} B(x^{(n)};y^{(p)}) \Bigg( \lrt{\prod_{i=1}^n \ad_{x_i} \prod_{j=1}^{p} a_{y_j}} - \sum_{\ell=0}^a \lambda_N^{\frac{\l}{2}} \sum_{m=0}^{\ell} \lrt{\prod_{i=1}^n \ad_{x_i} \prod_{j=1}^{p} a_{y_j}}_{m,\ell-m} \Bigg) \Bigg| \nonumber\\
& \leq& \norm{\tilde{\Chi}_a(t)}_{\FN} \bigg\| \int\dx^{(n)}\dy^{(p)} B(x^{(n)};y^{(p)}) \prod_{i=1}^n \ad_{x_i} \prod_{j=1}^{p} a_{y_j} \ChiN(t) \bigg\|_{\FN} \nonumber\\
&&+ \sum_{\ell=0}^{a} \lambda_N^{\frac{\l}{2}} \bigg\| \int\dx^{(n)}\dy^{(p)} B^*(y^{(p)};x^{(n)}) \prod_{j=1}^{p}\ad_{y_j} \prod_{i=1}^n a_{x_i}\Chil(t) \bigg\| \norm{\tilde{\Chi}_{a-\l}(t)} \nonumber\\
&\leq& C(n,p) \onorm{B} \bigg( \norm{\tilde{\Chi}_a(t)}_{\FN} \left\| (\Number+1)^{\frac{n+p}{2}} \ChiN(t) \right\|_{\FN} \nonumber\\
&&\hspace{2.5cm}
+ \sum_{\ell=0}^{a} \lambda_N^{\frac{\l}{2}} \norm{\tilde{\Chi}_{a-\l}(t)}\left\| (\Number+1)^{\frac{n+p}{2}}\Chil(t) \right\| \bigg)\nonumber\\
&\ls& \e^{C(n,p,a)t}\onorm{B} \left(\lN^{\frac{a+1}{2}} +\sum_{\ell=0}^{a} \lambda_N^{\frac{\l}{2}} \lN^{\frac{a-\l+1}{2}}\right),
\end{eqnarray}
where we used for the last estimate Lemmas~\ref{lem:moments:Chil} and \ref{lem:moments:ChilN:full:evolution} and Theorem~\ref{thm:norm:approx}. Finally, the statement follows by expanding up to $2a+1$ for $n+p$ even and up to $2a$ for $n+p$ odd because all contributions to the last line of \eqref{eqn:proof:thm:corr:fctns} with $n+p+\l$ odd vanish by Corollary~\ref{cor:gen:Wick:odd}.
\qed

\subsection{Proof of Theorem \ref{thm:RDM}}\label{subsec:proofs:RDM}

We start from the expression \eqref{gamma_psi_gamma_chi} for the one-particle reduced density matrix, i.e., 
\begin{eqnarray}\label{gamma_psi_gamma_chi_again}
\gPsiNo(t) 
&=& \ppt + \frac{1}{\sqrt{N}} \left(|\varphi(t) \rangle \langle \beta_{\ChiN(t)}| + |\beta_{\ChiN(t)} \rangle \langle \varphi(t)|\right) \nonumber\\
&&+\frac{1}{N} \left(\gChiNt- \ppt \lr{\ChiN(t),\Number \ChiN(t)}_{\FN}\right),
\end{eqnarray}
where
\begin{equation}
\bChiNt(x) := \lr{\ChiN(t), \sqrt{1-\frac{\Number}{N}} \,a_x \ChiN(t)}_{\FN}.
\end{equation}
The idea of the proof is to use expansions of $\sqrt{1-\frac{\Number}{N}}$ and $\frac{1}{N}$ in powers of $\lambda_N$, estimate the remainder terms, and use Corollary~\ref{thm:correlation functions_simplified} to estimate the difference of the microscopic and effective $n$-point correlation functions. Following the calculations in Appendix \ref{appendix:calculus}, we obtain
\begin{subequations}
\begin{eqnarray}
\frac{1}{N}
&=&\lN\frac{1}{1+\lN} \;=\; \sum\limits_{\l=0}^{a-1} \tilde{c}_\l \lN^{1+\l} + \lN^{a+1}R_{1,a}\,,\\
\frac{\sqrt{\left[N-\Number\right]_+}}{N} 
&=&\sum\limits_{\l=0}^a\lN^{\l+\frac12}\sum\limits_{k=0}^\l \tilde{c}_{\l,k}(\Number-1)^k + \lN^{a+\frac32}\mathbb{R}_{2,a}\,,
\end{eqnarray}
\end{subequations}
with $\tilde{c}_\l:=(-1)^\l c_\l^{(3/2)}$, $\tilde{c}_{\l,k}:=\tilde{c}_{\l-k}c_k^{(0)}$, $|R_{1,a}| \leq C(a)$ and $\norm{\mathbb{R}_{2,a}\bPhi}\leq C(a)\norm{(\Number+1)^{a+1}\bPhi}$ for any $\bPhi\in\Fock$ and some constant $C(a)>0$.
For the $|\varphi(t) \rangle \langle \beta_{\ChiN(t)}|$ term in \eqref{gamma_psi_gamma_chi_again} we then find, for any $A\in\mathcal{L}(\fH)$,
\begin{eqnarray}
&&\hspace{-0.7cm}\frac{1}{\sqrt{N}}\int\dx\dy A(x;y)\overline{\varphi(t,x)}\bChiNt(y)\nonumber\\
&=&\hspace{-2.5mm}\sum\limits_{\l=0}^{a}\sum\limits_{k=0}^\l\lN^{\l+\frac12}\tilde{c}_{\l,k}\int\dx \overline{(A^*\varphi)(t,x)}
\lrt{(\Number-1)^k a_x}_N + \mathcal{R}_A^{(1)}\nonumber\\
&=&\hspace{-2.5mm}\sum\limits_{\l=0}^{a}\sum\limits_{k=0}^\l\sum\limits_{m=0}^{a-\l-1}\sum\limits_{n=0}^{2m+1}\lN^{\l+m+1}\tilde{c}_{\l,k}\int\dx\overline{(A^*\varphi)(t,x)}\lrt{(\Number-1)^ka_x}_{n,2m+1-n}
\hspace{-0.1cm}+ \mathcal{R}_A^{(1)} +\mathcal{R}_A^{(2)}\nonumber\\[5pt]
&=&\hspace{-2.5mm}\sum\limits_{\l=0}^{a-1}\lN^{\l+1}\sum\limits_{m=0}^\l\sum\limits_{k=0}^{\l-m}\sum\limits_{n=0}^{2m+1}\tilde{c}_{\l-m,k}\int\dx\overline{(A^*\varphi)(t,x)}\lrt{(\Number-1)^ka_x}_{n,2m+1-n}
\hspace{-0.1cm}+ \mathcal{R}_A^{(1)}+\mathcal{R}_A^{(2)}\qquad\quad
\end{eqnarray}
with
\begin{eqnarray}
\mathcal{R}_A^{(1)}&:=&\lN^{a+\frac32}\int\dx  \overline{(A^*\varphi)(t,x)}\lrt{\mathbb{R}_{2,a} a_x}_N\nonumber\\
\mathcal{R}_A^{(2)}&:=&
\sum\limits_{\l=0}^{a}\lN^{\l+\frac12}\sum\limits_{k=0}^\l\tilde{c}_{\l,k}\int\dx \overline{(A^*\varphi)(t,x)}\Bigg[\lrt{(\Number-1)^ka_x}_N\nonumber\\
&&\qquad\qquad\qquad\qquad-\sum\limits_{m=0}^{a-\l-1}\sum\limits_{n=0}^{2m+1}\lN^{m+\frac{1}{2}}\lrt{(\Number-1)^ka_x}_{n,2m+1-n}\Bigg]\,.
\end{eqnarray}
Analogously, we compute for the second line in \eqref{gamma_psi_gamma_chi_again},
\begin{eqnarray}
&&\hspace{-1cm}\frac{1}{N}\int\dx\dy A(x;y)\left(\lrt{\ad_xa_y}_N+\ppt(y;x)\lrt{\Number}_N\right)\nonumber\\
&=&\sum\limits_{\l=0}^{a-1}\lN^{\l+1}\sum\limits_{m=0}^\l \sum\limits_{n=0}^{2m}\tilde{c}_{\l-m}\int\dx\dy A(x;y)\left(\lrt{\ad_x a_y}_{n,2m-n}+\ppt(y;x)\lrt{\Number}_{n,2m-n}\right)\nonumber\\[5pt]
&& +\mathcal{R}_A^{(3)}+\mathcal{R}_A^{(4)}\,,
\end{eqnarray}
where
\begin{eqnarray}
\mathcal{R}_A^{(3)}&:=&\lN^{a+1} R_{1,a}\int\dx\dy A(x;y)\left(\lrt{\ad_xa_y}_N+\ppt(y;x)\lrt{\Number}_N\right)\,,\\
\mathcal{R}_A^{(4)}
&:=&\sum\limits_{\l=0}^{a-1}\tilde{c}_\l\lN^{\l+1}\int\dx\dy A(x;y)\Bigg[\lrt{\ad_xa_y}_N
-\sum\limits_{m=0}^{a-\l-1}\lN^m\sum\limits_{n=0}^{2m}
\lrt{\ad_xa_y}_{n,2m-n} \Bigg]\nonumber\\
&&+\lr{\pt,A\pt}_{\fH}
\sum\limits_{\l=0}^{a-1}\tilde{c}_\l\lN^{\l+1}\Bigg[\lrt{\Number}_N
 -\sum\limits_{m=0}^{a-\l-1}\lN^m\sum\limits_{n=0}^{2m}
\lrt{\Number}_{n,2m-n}\Bigg]\nonumber\,.
\end{eqnarray}
Lemma~\ref{lem:aux:new} and the bound on moments of $\Number+1$ from Lemma~\ref{lem:moments:ChilN:full:evolution} imply
\begin{eqnarray}
\big|\mathcal{R}_A^{(1)}\big|
&\leq&\lN^{a+\frac32}C(a)\onorm{A}\norm{(\Number+1)^{a+1}\ChiN(t)}_{\FN}\norm{\Number^\frac12\ChiN(t)}_{\FN}\nonumber\\
&\leq& \lN^{a+\frac32}\onorm{A}\,\e^{C(a)t}\,,\\[5pt]
\big|\mathcal{R}_A^{(3)}\big|
&\leq&\lN^{a+1}|R_{1,a}|\Big(\norm{\ChiN(t)}_{\FN}\Big\|\int\dx\dy A(x;y)\ad_xa_y\ChiN(t)\Big\|_{\FN} \nonumber\\
&&\qquad\qquad\qquad+ \lr{\pt,A\pt}_{\fH}\lr{\ChiN(t),\Number\ChiN(t)}\Big)\nonumber\\
&\leq& \lN^{a+1} \onorm{A} \,\e^{C(a)t}\,.
\end{eqnarray}
In order to bound $\mathcal{R}_A^{(2)}$ and $\mathcal{R}_A^{(4)}$, we bring $\Number^k$  into normal order by iteratively applying \eqref{eqn:aux:1}, i.e.,
\begin{eqnarray}
\Number^k
=\sum\limits_{\l_1=0}^{k-1}\tbinom{k-1}{\l_1} \int\dx_1\ad_{x_1}\Number^{\l_1}a_{x_1}=\sum_{\l=1}^k C_{k,\l} \int \dx^{(\l)} \, \ad_{x_1}\mycdots\,\ad_{x_\l} a_{x_1}\mycdots\, a_{x_\l}
\end{eqnarray}
for some constants $C_{k,\l}$. Consequently, for $\lrt{\cdot}\in\big\{\lrt{\cdot}_{\l,k}\,,\,\lrt{\cdot}_N\big\}$,
\begin{eqnarray}
&&\hspace{-1cm}\int\dx\overline{(A^*\varphi)(t,x)}\lrt{(\Number-1)^ka_x}\nonumber\\
&=& \sum\limits_{\l=1}^k  \tilde{C}_{k,\l}\int\dx^{(\l)}\dy^{(\l+1)}B(t,x^{(\l)};y^{(\l+1)})\lrt{\ad_{x_1}\mycdots\,\ad_{x_\l}a_{y_1}\mycdots\,a_{y_{\l+1}}}\,,
\end{eqnarray}
where $B(t,x^{(\l)};y^{(\l+1)}):=\delta(y_1-x_1)\mycdots\delta(y_\l-x_\l)\overline{(A^*\varphi)(t,x_{\l+1})}$ is the Schwartz integral  kernel of the  operator $B:\fH^{\l+1}\to\fH^\l$, $\psi\mapsto(B\psi)(x^{(\l)})=\int\dx_{\l+1}\overline{(A^*\varphi)(t,x_{\l+1})}\psi(x^{\l+1})$ with $\norm{B}_{\mathcal{L}(\fH^{\l+1},\fH^\l)}\leq\onorm{A}$. Hence, Corollary \ref{thm:correlation functions_simplified} leads to the bound
\begin{equation}
\big|\mathcal{R}_A^{(j)}\big|\leq \onorm{A}\,\e^{C(a)t}\lN^{a+1}\,,\qquad j=2,4\,.
\end{equation}
Finally, Theorem \ref{thm:RDM} follows by duality of compact and trace class operators.
\qed

\section*{Acknowledgments}
\begin{wrapfigure}{l}{0.088\textwidth}
 \vspace{-15pt}
\includegraphics[scale=0.27]{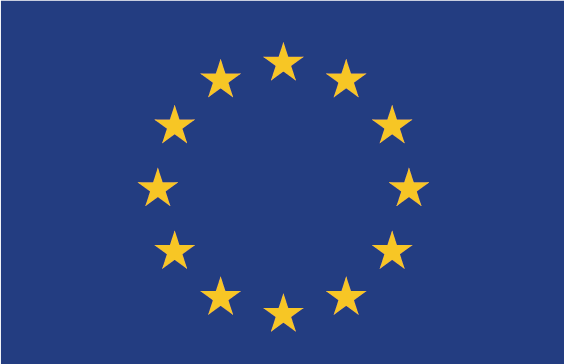}
  \vspace{-11pt}
\end{wrapfigure}
We are grateful for the hospitality of Central China Normal University (CCNU), where parts of this work were done, and thank Phan Th\`{a}nh Nam, Simone Rademacher, Robert Seiringer and Stefan Teufel for helpful discussions. L.B.\ gratefully acknowledges the support by the German Research Foundation (DFG) within the Research Training Group 1838 ``Spectral Theory and Dynamics of Quantum Systems'', and the funding from the European Union’s Horizon 2020 research and innovation programme under the Marie Sk{\textl}odowska-Curie Grant Agreement No.\ 754411.

\appendix
\section{Derivation of the Excitation Hamiltonian}\label{appendix:Hamiltonian}
Let us write the Hamiltonian \eqref{HN} as
\begin{equation}
\HN = \sum\limits_{j=1}^N \hpt_j + \lN\sum\limits_{1\leq i<j\leq N} \Wpt(x_i,x_j) \,,
\end{equation}
with $\Wpt$ as in \eqref{eqn:W(t,x,y)}.
Since $\ChiN(t) := \UNpt\PsiN(t)$, we find with \eqref{LNSS:map:U} that
\begin{eqnarray}
\i \partial_t \ChiN(t) &=& \i \partial_t \left( \bigoplus\limits_{j=0}^N\big(\qpt\big)^{\otimes j}\left(\frac{a(\pt)^{N-j}}{\sqrt{(N-j)!}}\PsiN(t)\right) \right) \nonumber\\
&=& \d\Gamma\big(\hpt\big) \ChiN(t) + \UNpt\left( \HN - \sum_{j=1}^N \hpt_j \right) \UNpt^* \ChiN(t)\,.
\end{eqnarray}
In some orthonormal basis $\{ \varphi_n(t) \}_{n\geq 0}$ with $\varphi_0(t) = \pt$, we define the matrix elements
\begin{equation}
W_{ijk\ell} := \int \dx \dy\, \overline{\varphi_i(t,x)} \, \overline{\varphi_j(t,y)} \Wpt(x,y) \varphi_k(t,x) \, \varphi_\ell(t,y)\,,
\end{equation}
omitting the time dependence for ease of notation. Then, denoting $a^\sharp_i := a^\sharp(\varphi_i(t))$,
\begin{eqnarray}
\HN - \sum_{j=1}^N \hpt_j &=& \frac{1}{2} \lN \sum_{i,j,k,\ell \geq 0} W_{ijk\ell} a^\dagger_i a^\dagger_j a_k a_\ell \nonumber\\
&=& \lN \sum_{j,k > 0} W_{0jk0}\, a^\dagger_0 a^\dagger_j a_k a_0 + \frac{\lN}{2} \left( \sum_{i,j > 0} W_{ij00}\, a^\dagger_i a^\dagger_j a_0 a_0 + \hc \right) \nonumber\\
&& + \lN \left( \sum_{i,j,k > 0} W_{ijk0} \, a^\dagger_i a^\dagger_j a_k a_0 + \hc \right) + \frac{\lN}{2} \sum_{i,j,k,\ell > 0} W_{ijk\ell}\, a^\dagger_i a^\dagger_j a_k a_\ell\qquad\;
\end{eqnarray}
since $W_{0000} = W_{000\ell} = W_{i0k0} = 0$ etc.\ for all $i,k,\ell > 0$ by definition of $\Wpt$. Finally, \eqref{eq_H_N_Fock} follows directly from the transformation rules \eqref{eqn:substitution:rules} because $W_{0jk0} = \lr{\varphi_j(t), \Kopt \varphi_k(t)}$, $W_{ij00} = \lr{\varphi_i(t) \otimes \varphi_j(t), \Ktpt}$, and $W_{ijk0} = \lr{\varphi_i(t) \otimes \varphi_j(t), \Kthpt(t) \varphi_k(t)}$ for all $i,j,k > 0$.

\section{Assumption~\ref{ass:initial:data} for Trapped Initial Data}\label{appendix:ini_data}

Let us discuss in more detail how the results of \cite{spectrum} provide a natural class of initial data satisfying  Assumption~\ref{ass:initial:data}. If $\PsiN(0)$ is the ground state or a low-energy eigenstate of $\HN_\trap$ as in \eqref{H:N:trap}, the initial condensate wave function $\varphi(0)$ is given by the normalized minimizer $\varphi_\mathrm{trap}$ of the Hartree energy functional
\begin{equation}
\mathcal{E}_\trap[\varphi]=\int\limits_{\R^d}\left(|\nabla\varphi(x)|^2+V_\trap(x)|\varphi(x)|^2\right)\dx + \tfrac12\int\limits_{\R^{2d}}v(x-y)|\varphi(x)|^2|\varphi(y)|^2\dx\dy\,,
\end{equation}
whose minimum is denoted by $e_\trap:=\mathcal{E}_\trap[\varphi_\trap]$. 
More precisely, Assumption \ref{ass:initial:data} is satisfied for each eigenstate $\PsiN_\trap$ of $\HN_\trap$ associated with an eigenvalue $E_\trap\in\mathfrak{E}_\trap^\zeta$ for some $\zeta\geq 0$. Here, $\mathfrak{E}_\trap^\zeta$ is the set of eigenvalues of $\HN_\trap$ such that 
$|E_\trap^{(n)}-N e_\trap|\leq \zeta$
for any $E_\trap^{(n)}\in\mathfrak{E}_\trap^\zeta$ and such  that
\begin{equation}\label{eqn:E:gamma:trap}
\lim\limits_{N\to\infty}(Ne_\trap-E_\trap^{(n_1)})\neq \lim\limits_{N\to\infty}(Ne_\trap-E_\trap^{(n_2)})\,
\end{equation}
for any $E_\trap^{(n_1)}, E_\trap^{(n_2)}\in\mathfrak{E}_\trap^\zeta$ with $n_1\neq n_2$,
where the eigenvalues are counted with multiplicity.
Hence, $\mathfrak{E}^\zeta_\trap$ contains all non-degenerate eigenvalues of $\HN_\trap$ with an energy of order one above the ground state energy, which additionally satisfy the condition that no two elements of $\mathfrak{E}^\zeta_\trap$ converge to the same eigenvalue of the Bogoliubov Hamiltonian corresponding to $\HN_\trap$\footnote{
The result in \cite{spectrum} is essentially a perturbative expansion of the spectral projectors of $\FockHN_\trap$ around the spectral projectors of the Bogoliubov Hamiltonian. Due to the strategy of proof (the spectral projectors are expressed as integrals of the resolvent of $\FockHN_\trap$ along a contour enclosing the appropriate Bogoliubov eigenvalue), one cannot separate projectors of $\FockHN$ corresponding to eigenvalues with equality in \eqref{eqn:E:gamma:trap}. Moreover, in the case of degenerate eigenvalues, there is no one-to-one correspondence between the spectral projector and a wave function. We expect that suitable functions $\Chil(t)$ can be defined in the context of degenerate perturbation theory.}. 
The state $\Chio(0)$ as in \eqref{def:Chi0:0} is a normalized eigenstate of this Bogoliubov Hamiltonian. If $\PsiN(0)$ is the ground state of $\HN_\trap$, it follows that $\tilde{\nu}=0$ in \eqref{def:Chi0:0}, i.e., $\Chio(0)$ is quasi-free.

\begin{proposition}\label{prop:H:trap}
Let $v:\R^d\to\R$ be measurable, even, and of positive type (i.e., $v$ has a non-negative Fourier transform). Let $V_\mathrm{trap}:\R^d\to\R$ be measurable, locally bounded and non-negative, and let $V_\mathrm{trap}$ tend to infinity as $|x|\to\infty$. 
Let $E_\trap\in\mathfrak{E}_\trap^\zeta$ for some $\zeta\geq 0$ and denote by $\PsiN_\trap\in\fH^N_\sym$ the associated normalized  eigenstate of $\HN_\trap$.
Let $\PsiN(0)=\PsiN_\trap$ and $\varphi(0)=\varphi_\trap$.
Then Assumption \ref{ass:initial:data} is satisfied for any $\tilde{a}\in\N_0$.
\end{proposition}

Proposition \ref{prop:H:trap} is proven in \cite[Theorem 3]{spectrum}. The coefficients $\mathfrak{a}^{(\l)}_{n,m,\mu}$ can be retrieved from \cite{spectrum}. For example, for $\PsiN(0)$ the ground state of $\HN_\mathrm{trap}$, the first order correction $\Chit(0)$ to the Bogoliubov ground state $\Chio(0)$ is given as
\begin{equation}
\Chit(0)=\frac{\id-|\Chio(0)\rangle\langle\Chio(0)|}{E^{(0)}_\mathrm{trap}-\FockH^{(0)}_\mathrm{trap}}\,\FockH^{(1)}_\mathrm{trap}\,\Chio(0)\,,
\end{equation}
where the operators $\FockH^{(0)}_\mathrm{trap}$ and $\FockH^{(1)}_\mathrm{trap}$ (corresponding to  $\HN_\mathrm{trap}$) are defined analogously to the operators $\FockH^{(0)}_\pz$ and $\FockH^{(1)}_\pz$  (corresponding to $\HN$), and where $E^{(0)}_\mathrm{trap}$ denotes the ground state energy of $\FockH^{(0)}_\mathrm{trap}$. For $\cU_\mathrm{trap}$ the Bogoliubov transform diagonalizing $\FockH^{(0)}_\mathrm{trap}$, it holds that 
\begin{equation}
\Chio(0)=\cU_\mathrm{trap}^*|\Omega\rangle\,
\end{equation}
and
\begin{equation}
\Chit(0)=\cU_\mathrm{trap}^*\Big(\bigoplus_{m\geq 0}\mathcal{O}^{(m)}\Big) \cU_\mathrm{trap}\, \FockH^{(1)}_\mathrm{trap}\,\cU_\mathrm{trap}^*|\Omega\rangle\,,
\end{equation}
where 
\begin{equation}
\cU_\mathrm{trap}\,\frac{\id-|\Chio(0)\rangle\langle\Chio(0)|}{E^{(0)}_\mathrm{trap}-\FockH^{(0)}_\mathrm{trap}}\,\cU_\mathrm{trap}^*
=:\mathcal{O}
=\bigoplus_{m\geq 0}\mathcal{O}^{(m)}
\end{equation}
because $\mathcal{O}$ is particle number preserving. The $m$-body operators $\mathcal{O}^{(m)}$ are given by
\begin{equation}
\mathcal{O}^{(m)} = 
-\sum\limits_{\substack{\bj\in\N^m\\j_1\leq\dots\leq j_m}}\frac{1}{\sum_{k=1}^m e_{j_k}}\ad(\xi_{j_1})\mycdots\ad(\xi_{j_m})|\Omega\rangle\langle
\Omega|a(\xi_{j_m})\mycdots a(\xi_{j_1})\,,
\end{equation}
where $\{\xi_j\}_{j\geq 1}$ denotes a complete set of normalized eigenfunctions of the one-body operator resulting from the diagonalization of $\FockH^{(0)}_\mathrm{trap}$. The corresponding eigenvalues are denoted by $e_j$.
One computes 
\begin{equation}
\cU_\mathrm{trap}\,\FockH^{(1)}_\mathrm{trap}\,\cU_\mathrm{trap}^*|\Omega\rangle 
= \int\dx f^{(1)}_\mathrm{trap}(x)\ad_x|\Omega\rangle +\int\dx^{(3)} f^{(3)}_\mathrm{trap}(x^{(3)})\ad_{x_1}\ad_{x_2}\ad_{x_3}|\Omega\rangle 
\end{equation}
with
\begin{subequations}
\begin{eqnarray}
f^{(1)}_\mathrm{trap}(x)&:=&\int\dy\left(\mathfrak{A}^{(-1,1,1)}_{\mathrm{trap};1,3}(y,y,x) + \mathfrak{A}^{(-1,1,1)}_{\mathrm{trap};1,3}(y,x,y)+\mathfrak{A}^{(1,-1,1)}_{\mathrm{trap};1,3}(x,y,y)\right),\qquad\\
f^{(3)}_\mathrm{trap}(x^{(3)})&:=&\mathfrak{A}^{(1,1,1)}_{\mathrm{trap};1,3}(x^{(3)})
\end{eqnarray}
\end{subequations}
for $\mathfrak{A}^{(\bj)}_{\mathrm{trap};1,3}$ defined analogously to \eqref{eqn:A:coeff}. Consequently, 
\begin{eqnarray}
\Chit(0)&=& \cU_\mathrm{trap}^*\,\int\dx\big(\mathcal{O}^{(1)}f_\mathrm{trap}^{(1)}\big)(x)\ad_x\,\cU_\mathrm{trap}\Chio(0)\nonumber\\
&&+ \,\cU_\mathrm{trap}^*\,\int\dx^{(3)}\big(\mathcal{O}^{(3)}\tilde{f}_\mathrm{trap}^{(3)}\big)(x^{(3)})\ad_{x_1}\ad_{x_2}\ad_{x_3}\,\cU_\mathrm{trap}\Chio(0)
\end{eqnarray}
for $\tilde{f}^{(3)}_\mathrm{trap}(x^{(3)})=(3!)^{-1/2}\sum_{\sigma\in\mathfrak{S}_3} f^{(3)}_\mathrm{trap}(x_{\sigma(1)},x_{\sigma(2)},x_{\sigma(3)})$ the symmetrized version of $f^{(3)}_\mathrm{trap}$. Finally, the coefficients $\mathfrak{a}^{(1)}_{1,m,\mu}$ follow from this by \eqref{eqn:trafo:ax:invers}, e.g., 
\begin{equation}
\mathfrak{a}^{(1)}_{1,1,0}(x) = -\big(\overline{V}^*_\mathrm{trap}\mathcal{O}^{(1)}f^{(1)}_\mathrm{trap}\big)(x)\,,\qquad
\mathfrak{a}^{(1)}_{1,1,1}(x) = \big(U^*_\mathrm{trap}\mathcal{O}^{(1)}f^{(1)}_\mathrm{trap}\big)(x)
\end{equation}
for $U_\mathrm{trap}$ and $V_\mathrm{trap}$ the matrix entries of $\cU_\mathrm{trap}$ as in \eqref{BogV:block:form}.

\section{Proof of Lemma \ref{lem:calculus}.}
\label{appendix:calculus}

We estimate the Taylor series remainders for the functions $f_0^{(0)}$ and $f_\mu^{(0)} f_0^{(0)}$, where
\begin{equation}
f_\mu^{(n)}:[-1,1-\mu]\to[0,\infty)\,,\qquad x\mapsto f_\mu^{(n)}(x):=(1-x-\mu)^{\frac12-n}\,,
\end{equation}
for $\mu\in[0,\frac14)$ and $n\geq0$. For $c_\l^{(n)}$ as in \eqref{eqn:taylor:coeff}, we find that
\begin{eqnarray}
f_\mu^{(n)}(x)&=:&\sum\limits_{\l=0}^a \frac{c^{(n)}_\l}{(1-\mu)^{n+\l-\frac12}} \,x^\l+R^{(n)}_{a,\mu}(x) \,,\label{eqn:proof:thm:calculus:2}\\
\big(f^{(0)}_\mu f^{(0)}_0\big)(x)
&=:&\sum\limits_{\l=0}^a \sum\limits_{m=0}^\l c_m^{(0)} c^{(0)}_{\l-m} \left(\frac{1}{1-\mu}\right)^{m-\frac12} \,x^\l+\tilde{R}_{a,\mu}(x)\,.
\label{eqn:proof:thm:calculus:1}
\end{eqnarray} 
\medskip

\noindent\textbf{Estimates for $f^{(0)}_0$.}
The remainder $R^{(n)}_{a,0}(x)$ is given as
\begin{equation}\label{eqn:proof:thm:calculus:3}
R^{(n)}_{a,0}(x)
 = (a+1)\,c^{(n)}_{a+1}\int\limits_0^x\frac{1}{(1-t)^{n+\frac12}}\left(\frac{x-t}{1-t}\right)^a \dt
=c^{(n)}_{a+1}\left(\frac{1}{1-\xi}\right)^{n+a+\frac12}x^{a+1}\,\end{equation}
for some $\xi\in(0,x)$. 
For $x\in[-1,\frac12]$, the second equality yields $|R^{(0)}_{a,0}(x)|\leq2^a|x|^{a+1}$ since $1-\xi>\frac12$ and by \eqref{eqn:cln:bound}. For $x\in(\frac12,1]$, $|R^{(0)}_{a,0}(x)|\leq 1\leq   2^{a+1}|x|^{a+1}$ by the first equality.
\medskip

\noindent\textbf{Estimates for $f^{(0)}_{\lN} f^{(0)}_0$.}
By \eqref{eqn:proof:thm:calculus:1} for $\mu=\lN$ and $x=\lN(k-1)$ with $0\leq k\leq N-1$,
\begin{equation}\label{eqn:proof:thm:calculus:4}
\big(f_{\lN}^{(0)}f_0^{(0)}\big)(\lN(k-1))
= \sum\limits_{\l=0}^a\sum\limits_{n=0}^\l c_n^{(0)}c_{\l-n}^{(0)} f^{(n)}_0(\lN)\lN^\l (k-1)^\l 
+ \tilde{R}_{a,\lN}(\lN(k-1))\,.
\end{equation}
By \eqref{eqn:proof:thm:calculus:2} and \eqref{eqn:proof:thm:calculus:3},
\begin{equation}
f_0^{(n)}(\lN) = \sum\limits_{\nu=0}^{a-\l} c_\nu^{(n)}
\lN^{\nu}
+c_{a-\l+1}^{(n)}(1-\xi)^{\l-a-n-\frac12}\lN^{a-\l+1}
\end{equation}
for some $\xi\in(0,\lN)$. Inserting this formula into \eqref{eqn:proof:thm:calculus:4} yields
\begin{eqnarray}
&&\hspace{-1cm}\big(f_{\lN}^{(0)}f_0^{(0)}\big)(\lN (k-1))\nonumber\\
&=&\sum\limits_{\l=0}^a\sum\limits_{n=0}^\l \sum\limits_{\nu=0}^{a-\l}
c_n^{(0)}c_{\l-n}^{(0)} c_\nu^{(n)}
(k-1)^\l \lN^{\nu+\l}\nonumber\\
&&+\lN^{a+1} \sum\limits_{\l=0}^a\sum\limits_{n=0}^\l 
c_n^{(0)}c_{\l-n}^{(0)} c_{a-\l+1}^{(n)}\left(\frac{1}{1-\xi}\right)^{a-\l+n+\frac12} (k-1)^\l
+ \tilde{R}_{a,\lN}(\lN(k-1))\nonumber\\
&=&\sum\limits_{\l=0}^a\lN^{\l}\sum\limits_{m=0}^\l d_{\l,m} (k-1)^m
+\lN^{a+1}r(k,a)+\tilde{R}_{a,\lN}(\lN(k-1))\,,
\label{eqn:proof:thm:calculus:8}
\end{eqnarray}
where we re-ordered the triple sum using $d_{\l,m}$ from \eqref{eqn:taylor:coeff:2}, and abbreviated the double sum as $r(k,a)$.
Since $\xi\in(0,\lN)$ and $n\leq\l\leq a$, it follows that $(1-\xi)^{-(a-\l+n+\frac12)}\leq (\frac{N-1}{N-2})^{a+\frac12}$ and consequently $|r(k,a)|\leq (a+1)^2 2^a\left(\tfrac{N-1}{N-2}\right)^{a+\frac12}(k+1)^a$.
The remainder $\tilde{R}_{a,\lN}$ is  given as
\begin{eqnarray}
\tilde{R}_{a,\lN}(x)
&=&\sum\limits_{n=0}^{a+1}c^{(0)}_nc^{(0)}_{a+1-n}\left(\frac{1}{1-\xi-\lN}\right)^{n-\frac12}\left(\frac{1}{1-\xi}\right)^{a-n+\frac12}x^{a+1}\nonumber\\
&=&(a+1)\sum\limits_{n=0}^{a+1}c^{(0)}_n c^{(0)}_{a+1-n}\int\limits_0^x\left(\frac{1}{1-t-\lN}\right)^{n-\frac12}\left(\frac{1}{1-t}\right)^{a-n+\frac12}(x-t)^a \dt\qquad
\end{eqnarray}
for some $\xi\in(0,x)$. 
For $x\in[-1,\frac12]$, the first equality yields
$|\tilde{R}_{a,\lN}(x)|\leq a\,4^a |x|^{a+1}$,
where we used that that $2>1-{\xi} >\frac12$ and $2>1-\xi-\lN>\frac12\frac{N-3}{N-1}$. 
For $x\in(\frac12,1-\lN]$, the second equality leads to the estimate
$|\tilde{R}_{a,\lN}(x)|
\leq (a+1)^2 2^{a+2}|x|^{a+1}$.
\\

\noindent\textbf{Proof of Lemma \ref{lem:calculus}.}
Since 
$\frac{\sqrt{[N-\Number]_+}}{N-1}=\lN^\frac12f^{(0)}_0(\lN(\Number-1))$ on $\FN$,
$\lN^{-1}=N-1\leq\Number-1$ on $\Fock^{\geq N}$, and
$\frac{\sqrt{[(N-\Number)(N-\Number-1)]_+}}{N-1} =\big(f^{(0)}_{\lN} f_0^{(0)}\big)\left(\lN(\Number-1)\right)$
on $\Fock^{\leq N-1}$, we find
\begin{subequations}
\begin{eqnarray}
\norm{R^{(3)}_a\bPhi}_{\FN}
&\leq& 2^{a+1}\norm{(\Number+1)^{a+1}\bPhi}_{\FN}\,,\\
\norm{R^{(3)}_a\bPhi}_{\Fock^{\geq N}} &\leq& (a+1) \norm{(\Number+1)^{a+1}\bPhi}_{\Fock^{> N}}\,,\\
\norm{R^{(2)}_a\bPhi}_{\Fock^{\leq N-1}}&\leq& 4^a(a+1)^2\norm{(\Number+1)^{a+1}\bPhi}_{\Fock^{\leq N-1}}\,,\\
\norm{R^{(2)}_a\bPhi}_{\Fock^{\geq N}}
&\leq& 2^a(a+1)^3\norm{(\Number+1)^{a+1}\bPhi}_{\Fock^{\geq N}}\,.
\end{eqnarray}
\end{subequations}

\renewcommand{\bibname}{References}
\bibliographystyle{abbrv}
    \bibliography{bib_file}
\end{document}